\documentclass[a4paper,UKenglish,pdfa,cleveref,autoref,numberwithinsect]{lipics-noline-v2021}
\usepackage{microtype}%if unwanted, comment out or use option "draft"

\title{Pebble transducers with unary output}

\author{Gaëtan Douéneau-Tabot}{Université Paris Cité, CNRS, IRIF, F-75013, Paris, France
\and Direction générale de l'armement - Ingénierie des projets, Paris, France
}{doueneau@irif.fr}{}{}

\authorrunning{G. Douéneau-Tabot}%mandatory. First: Use abbreviated first/middle names. Second (only in severe cases): Use first author plus 'et al.'

\Copyright{Gaëtan Douéneau-Tabot}%mandatory, please use full first names. LIPIcs license is "CC-BY";  http://creativecommons.org/licenses/by/3.0/

 %TODO mandatory: Please choose ACM 2012 classifications from https://dl.acm.org/ccs/ccs_flat.cfm 

%mandatory

\category{}%optional, e.g. invited paper

\acknowledgements{The author is grateful to Olivier Carton for discussing about this work.
He also thanks the reviewers  of MFCS 2021 for their helpful comments and remarks.}%optional

\EventEditors{}
\EventNoEds{1}
\EventLongTitle{}
\EventShortTitle{MFCS 2021}
\EventAcronym{}
\EventYear{}
\EventDate{}
\EventLocation{}
\EventLogo{}
\SeriesVolume{}
\ArticleNo{XX}
\usepackage{tikz}
\usepackage[utf8]{inputenc}
\usetikzlibrary{shapes,arrows,positioning,automata,shadows}
\usepackage{amsfonts, amsmath, amssymb,dsfont, mathrsfs}
\usepackage{amsthm}
\usepackage{hyperref}
\usepackage[ruled]{algorithm2e}
\usepackage{stmaryrd}
\usepackage{pifont}
\usepackage{soulutf8}
\usepackage{yhmath}
\usepackage{bm}

\newtheorem{sublemma}[theorem]{\bfseries Sublemma}

\newcommand{\mb}[1]{\mathbb{#1}}
\newcommand{\mc}[1]{\mathcal{#1}}
\newcommand{\mf}[1]{\mathfrak{#1}}

\newcommand{\vide}{\varepsilon}
\newcommand{\omar}[2]{#1{\uparrow}#2}
\newcommand{\defined}{:=}
\renewcommand{\phi}{\varphi}
\renewcommand{\epsilon}{\varepsilon}
\newcommand{\Nat}{\mb{N}}

\newcommand{\uu}{\textcolor{red}{u_1}}
\newcommand{\ud}{\textcolor{blue!90}{u_2}}

\newcommand{\mh}{\Lambda}

\newcommand{\lmove}{\triangleleft}
\newcommand{\rmove}{\triangleright}
\newcommand{{\lmark}}{\vdash}
\newcommand{{\rmark}}{\dashv}

\newcommand{\apar}[1]{{\wideparen{#1}}}
\newcommand{\facto}[2]{\operatorname{\normalfont \textsf{Fact}}(#2)}
\newcommand{\cro}[1]{\bm{\langle}\mathbf{#1}\bm{\rangle}}
\newcommand{\marq}[1]{\underline{#1}}
\newcommand{\neutral}{1_M}
\newcommand{\deutral}{2_M}

\newcommand{\nodi}{\mc{I}}
\newcommand{\nodj}{\mc{J}}
\newcommand{\nodb}{\mc{B}}
\newcommand{\nodp}{\mc{P}}
\newcommand{\noda}{\mc{A}}
\newcommand{\forest}{\mc{F}}

\newcommand{\up}[1]{\operatorname{\textsf{Up}}(#1)}
\newcommand{\nba}{\operatorname{\normalfont \textsf{nb-a}}}
\newcommand{\fr}[2]{\operatorname{\normalfont \textsf{Fr}}_{#2}(#1)}
\newcommand{\dep}[1]{\operatorname{\textsf{Dep}}(#1)}
\newcommand{\calls}{\operatorname{\normalfont \textsf{Calls}}}
\newcommand{\pro}[1]{\operatorname{\normalfont \textsf{prod}}(#1)}
\newcommand{\nodes}[1]{\operatorname{\normalfont \textsf{Nodes}}(#1)}
\newcommand{\bitype}[1]{\operatorname{\normalfont \textsf{bitype}}(#1)}
\newcommand{\type}[2]{\operatorname{\normalfont \textsf{type}}_{#2}(#1)}
\newcommand{\basis}[2]{\operatorname{\normalfont \textsf{basis}}_{#2}(#1)}
\newcommand{\itera}[1]{\operatorname{\normalfont \textsf{Iterable-nodes}}(#1)}
\newcommand{\mimi}[2]{\operatorname{\normalfont \textsf{middle}}_{#2}(#1)}
\newcommand{\valu}[1]{\operatorname{\normalfont \textsf{value}}(#1)}
\newcommand{\mine}[2]{\operatorname{\normalfont \textsf{min}}_{#2}(#1)}
\newcommand{\maxe}[2]{\operatorname{\normalfont \textsf{max}}_{#2}(#1)}
\newcommand{\lefe}[2]{\operatorname{\normalfont \textsf{left}}_{#2}(#1)}
\newcommand{\rige}[2]{\operatorname{\normalfont \textsf{right}}_{#2}(#1)}
\newcommand{\parti}[1]{\operatorname{\normalfont \textsf{Part}}(#1)}
\newcommand{\ltm}{\operatorname{\normalfont \textsf{Fact}}_{\le \mh}}

\newcommand{\regs}{\mf{X}}
\newcommand{\Som}{\operatorname{\textsf{Sum}}}
\newcommand{\Old}{\operatorname{\textsf{Old}}}

\newcommand{\SST}{{\textsf{SST}}}
\newcommand{\SSTs}{{\textsf{SSTs}}}

\newcommand{\oras}{\mf{F}}
\newcommand{\exte}{\mf{f}}
\newcommand{\exteg}{\mf{g}}

\newcommand{\trans}{\mc{T}}
\newcommand{\utrans}{\mc{U}}

\newcommand{\produc}{\operatorname{\normalfont \textsf{product}}}
\newcommand{\lproduc}{\operatorname{\normalfont \textsf{letter-product}}}

\newcommand{\sq}{\operatorname{\textsf{square}}}
\newcommand{\isq}{\operatorname{\textsf{iterated-square}}}
\newcommand{\trian}{\operatorname{\textsf{triangular-sum}}}
\newcommand{\mirror}{\operatorname{\textsf{reverse}}}

\begin{document}

\maketitle

\begin{abstract} 
 Boja{\'{n}}czyk recently initiated an intensive study
 of deterministic pebble transducers, which
 are two-way automata that can drop marks (named "pebbles") 
 on their input word, and produce an output word.
 They describe functions from words to words.
 Two natural restrictions of this definition have been investigated:
 marble transducers by Douéneau-Tabot et al.,
 and comparison-free pebble transducers (that we rename
 here "blind transducers") by Nguy{\^{e}}n et al.
 
  Here, we study the decidability of membership
  problems between the classes of functions computed by
  pebble, marble and blind transducers that produce
  a unary output. First, we
  show that pebble and marble transducers have the same
  expressive power when the outputs are unary (which is false over non-unary outputs).
  Then, we characterize $1$-pebble transducers with unary output
  that describe a function computable by a blind
  transducer, and show that the membership problem is decidable.
 These results can be interpreted in terms of automated simplification
of programs.
\end{abstract}

%%%%%%%%%%%%%%%%%%%%%%%%%%
% Introduction
%%%%%%%%%%%%%%%%%%%%%%%%%%

%%
%% Main
%%

\section{Introduction}

Regular languages can be described by several models such as deterministic, non-deterministic, or two-way (the reading head can move in two directions) finite automata \cite{shepherdson1959reduction}.
A natural extension consists in adding an output mechanism to finite automata.
Such machines, called \emph{transducers}, describe functions from words to words
(or relations when non-deterministic) and provide a natural way to model simple programs that
produce outputs.
The particular model of a \emph{two-way transducer} consists in a two-way automaton
enhanced with an output function.  It describes the class of \emph{regular
functions} which has been intensively studied for its fundamental properties:
closure under composition \cite{chytil1977serial}, logical characterization by
monadic second-order transductions \cite{engelfriet2001mso}, decidable
equivalence problem \cite{gurari1982equivalence}, etc.

\subparagraph*{Pebble transducers and their variants.}
The model of $k$-pebble transducer can be defined as
an inductive extension of two-way transducers.
A $0$-pebble transducer is just a two-way transducer.
For $k \ge 1$, a $k$-pebble transducer $\trans$ is a two-way transducer that,
when in a given configuration, can "call" an \emph{external function} $\exte$,
computed by some $(k{-}1)$-pebble transducer.\linebreak
$\trans$ gives as argument to
$\exte$ its input word together with a mark, named "pebble", on the
position from which the call was performed,
and uses the output of $\exte$ within its own output.

The behavior of a $1$-pebble transducer
is depicted in Figure \ref{fig:pebbles}.
Intuitively, a $k$-pebble transducer is some recursive program
whose recursion depth is at most $k{+}1$. Equivalently, it can be seen
as an iterative algorithm with "two-way for-loops",
such that the maximal depth of nested loops is $k{+}1$.
A $k$-pebble transducer can only produce an output whose length
is polynomial in its input's length, more precisely
$\mc{O}(n^{k+1})$ when $n$ is the input's length
(this is intuitive from the "nested loops" point of view).
The functions computed by a $k$-pebble transducer for some $k \ge 0$
are thus called \emph{polyregular functions} \cite{bojanczyk2018polyregular}.
Several properties of polyregular functions have been investigated:
closure under composition \cite{bojanczyk2018polyregular}, logical characterization by
monadic second-order interpretations \cite{bojanczyk2019string}, etc.
The equivalence problem (given two machines, do they compute the same function?)
is however still open.

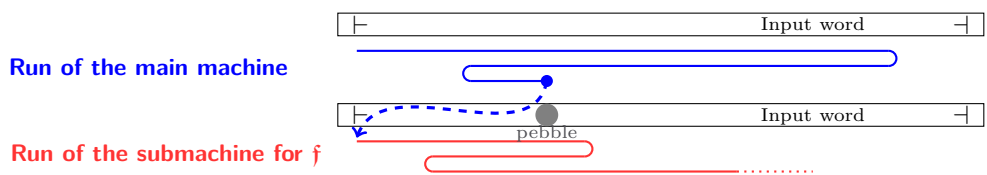
\begin{figure}[h!]

\centering
\begin{tikzpicture}{scale=1}

	\newcommand{\couleur}{blue}
	\newcommand{\texte}{\small \bfseries \sffamily \mathversion{bold} }

	% Input 1
	\draw (-0.25,4.2) rectangle (8.25,4.5);
	\node[above] at (6,4.1) {$\substack{\text{Input word}}$};
	\node[above] at (0.05,4.1) {$\lmark$};
	\node[above] at (7.95,4.1) {$\rmark$};

	% Run 1
	\node[above] at (-2.8,3.55) {\textcolor{\couleur}{\texte Run of the main machine}};
	\draw[-,thick,\couleur](0,4) -- (7,4);
	\draw[-,thick, \couleur] (7,4) arc (90:-90:0.1);
	\draw[-,thick,\couleur](7,3.8) -- (1.5,3.8);	
	\draw[-, thick,\couleur] (1.5,3.8) arc (90:270:0.1);
	\draw[-,thick,\couleur](1.5,3.6) -- (2.5,3.6);	
	\fill[fill = \couleur,even odd rule] (2.5,3.6) circle (0.08);
        \draw[->,very thick,\couleur,dashed](2.5,3.6) to[out= -100, in = 70] (0,2.85);

	% Input 2

	\draw (-0.25,3) rectangle (8.25,3.3);
	\node[above] at (6,2.9) {$\substack{\text{Input word}}$};
	\node[above] at (0.05,2.9) {$\lmark$};
	\node[above] at (7.95,2.9) {$\rmark$};
	\fill[fill = gray] (2.5,3.15) circle (0.15);

	% Run 2
	\node[above] at (-2.57,2.35) {\textcolor{red!80}{\texte Run of the submachine for $\exte$} };
	\draw[-,thick,red!80](0,2.8) -- (3,2.8);
	\draw[-,thick, red!80] (3,2.8) arc (90:-90:0.1);
	\draw[-,thick,red!80](3,2.6) -- (1,2.6);	
	\draw[-, thick,red!80] (1,2.6) arc (90:270:0.1);
	\draw[-,thick,red!80](1,2.4) -- (5,2.4);	
	\draw[-,thick,red!80,dotted](5,2.4) -- (6,2.4);

	\node[above,darkgray] at (2.5,2.67) {$\substack{\text{pebble}}$};
	
\end{tikzpicture}

\caption{\label{fig:pebbles} Behavior of a $1$-pebble transducer}

\end{figure}

Recently, two natural restrictions of pebble transducers have been introduced.
First, the \emph{$k$-marble transducers} of \cite{doueneau20} only
give as argument to their external function the prefix of the input word
which ends in the calling position (see Figure \ref{fig:marbles}).
Second, the \emph{$k$-blind transducers}\footnote{The original terminology of \cite{nguyen2020comparison} is \emph{comparison-free pebble transducers}, but we strongly believe that the term "blind" is more adapted, since there are no pebbles in this model.}
of \cite{nguyen2020comparison} give the whole input word,
but no pebble on the calling position (see Figure \ref{fig:blind}).
The classes of functions they compute are
strict subclasses of polyregular functions \cite{doueneau20,nguyen2020comparison}.

\subparagraph*{Class membership problems.}
These various models of transducers
raise several membership problems: given
a function computed by a machine of model $X$,
can it be computed by some machine of model $Y$?
When $Y$ is a restriction of $X$, this problem
reformulates as a program optimization question:
given a "complex" algorithm in a class $X$, can we build an
equivalent "simpler" one in class $Y$?
Thus it is of a foremost interest in practice.

Given a function $f$ computed by an $\ell$-pebble transducer,
one can ask whether it is computable by a $k$-pebble transducer
for a given $k < \ell$. The problem is open, but it is solved in
the case of marble \cite{doueneau20} and blind  \cite{nguyen2020comparison}
transducers and it turns out that a necessary and sufficient condition
for this membership is that $|f(w)| = \mc{O}(|w|^{k+1})$.
Using the "nested loops" interpretation of pebble transducers,
it means that an output of size $\mc{O}(|w|^{k+1})$
can always be produced with at most $k{+}1$ nested loops.

\subparagraph*{Contributions.} In this paper, we study a different membership problem: can a 
function given by a $k$-pebble transducer be computed
by a $k$-marble or $k$-blind transducer?
It turns out to be a more difficult question, since there is no
intuitive and machine-independent candidate for a membership condition (such as the size of the output).
In general, membership problems for transducers are difficult,
since contrary to regular languages, there is no "canonical" object known to
represent a regular function. Hence, there can be several seemingly unrelated
manners to produce the same function, and moving from one to another
can be technical.

We focus on \emph{transducers whose output alphabet is unary},
and our proof techniques are new.
The first main result is that (when the outputs are unary) $k$-pebble transducers
and $k$-marble transducers compute the same functions
(one direction is obvious since $k$-marble is a restriction of $k$-pebble).
The transformation is effective, but the way of producing the output
must sometimes be completely modified (the transformation modifies
the \emph{origin semantics}, in the sense of \cite{bojanczyk2014transducers}),
which creates an additional difficulty.
The correspondence
fails as soon as the output is not over a unary alphabet,
as detailed in Example \ref{ex:fail}.

\begin{example}[\cite{doueneau20,nguyen2020comparison}] \label{ex:fail}
The partial function $\{a,b\}^* \rightarrow \{a,b\}^*, a^mb^n \mapsto (b^na)^m$
can be computed by a $1$-pebble transducer,
but not by a $k$-marble for any $k \ge 0$.
\end{example}

Since the equivalence problem is decidable for marble transducers,
it follows from our result that it is also decidable for pebble transducers
with unary output.

As a second main result, we show how to decide (when the outputs
are unary) whether a function given by $1$-pebble ($\equiv$ $1$-marble) transducer
can be computed by a $1$-blind transducer, or more generally by a $k$-blind
transducer for some $k \ge 0$.
The technical proof also gives a syntactical characterization
of $1$-marble transducers whose function verify this property:
it describes a kind of "symmetry" in the production of the machine on its input.
Furthermore, the conversion is effective when possible, but once more
the manner of producing the output can be strongly modified.
Our techniques heavily rely on the theory of \emph{factorization forests};
this is, to our knowledge, the first time 
this notion is used for membership problems
of transducers, and we believe this approach to be fruitful.

Our results are summarized in red in Figure \ref{fig:conclu}.
We also give some examples of functions (their outputs
are non-negative integers, since we identify $\{a\}^*$ with $\Nat$).

%\begin{table}[h!]
%\begin{tabular}{|l | c | c | c |}
%
%\hline
%&$1$-pebble&$1$-marble&$1$-blind\\
%\hline
%$\sq : a^n \mapsto n^2$&\cmark&\cmark&\cmark\\
%\hline
%$\isq : a^{n_1} b a^{n_2} b \cdots  b a^{n_\ell} b \mapsto \sum_{i=1}^\ell  (n_i)^2$&\cmark&\cmark&\xmark\\
%\hline
%$\produc :  a^m b^n \mapsto mn$&\cmark&\cmark&\cmark\\
%\hline
%$\iproduc :  a^{m_1} b^{n_1} \cdots  a^{m_\ell} b^{n_\ell} \mapsto \sum_{i=1}^{\ell} m_i n_i$&\cmark&\cmark&\xmark\\
%\hline
%$\lproduc :  w\in \{a,b\}^* \mapsto |w|_a  |w|_b$&\cmark&\cmark&\cmark\\
%\hline
%$\trian :  a^{n_1} b a^{n_2} b \cdots  b a^{n_\ell} b \mapsto \sum_{k=1}^\ell k n_k $&\cmark&\cmark&\xmark\\
%\hline
%
%\end{tabular}
%
%\caption{\label{tab:examples} Some functions computable by unary $1$-pebble transducers}
%\end{table}

\begin{figure}[h!]

    \def\acircle{(0,0) circle (1)}
    \def\bcircle{(0,1.1) circle (2.1)}
    \def\ccircle{(0,1.6) circle (2.6)}
    \def\icircle{(0,1.7) circle (2.7)}
    
    	\begin{center}
	\hspace*{-0.4cm}
        \begin{tikzpicture}{scale=0.9}

	   % Remplissage rose
            \fill[red!15] \acircle;     
            %\fill[red!15] {(2.8,4.7) circle (1.5)};

            \begin{scope}
            \clip {(0,-1) rectangle (8,2)};
            \fill[red!15] \icircle;
            \end{scope}
            
            \begin{scope}
            \clip {\icircle};
            \fill[red!15] (-0.94,0.37) -- (1.404,5.75) -- (4.246,4.3) -- (3.25,1.4) -- cycle;
            \end{scope}

            \draw \acircle node {$\substack{0\text{-pebble} \\ =\\ 0\text{-marble} \\ = \\ 0\text{-blind}}$};
            \draw \bcircle ;
            %\draw \ccircle;
            \draw[dashed] \icircle;

            \node at (0.8,2) {$\substack{1\text{-blind}}$};
            \node at (-0.9,2.1) {$\textcolor{red}{\substack{1\text{-pebble} \\=\\ 1\text{-marble}}}$};
            
             \node at (0,0.86) {\tiny $\mc{O}(n)$};
             \node at (0,3.05) {\tiny $\mc{O}(n^2)$};
             %\node at (0,4.05) {\tiny $\mc{O}(n^3)$};
             \node at (0,4.25) {\tiny $\mc{O}(n^{k+1})$};
            
            %\node at (1.3,3.1) {$\substack{2\text{-blind}}$};
            %\node at (-1.2,3.2) {$\textcolor{red}{\substack{2\text{-pebble} \\ = \\ 2\text{-marble} }}$};
            \node at (-1,3.5) {$\textcolor{red}{\substack{k\text{-pebble} \\ = \\ k\text{-marble} }}$};
            \node at (1.3,3.4) {$\substack{k\text{-blind}}$};

             \draw[darkgray,dotted,very thick] (0.8,-0.2) -- (3.35,-0.2);
             \fill[darkgray] {(0.8,-0.2) circle (0.05)};
             \node[right] at (3.4,-0.15) {\small \textcolor{darkgray}{$\nba: w \in \{a,b\}^* \mapsto |w|_a$}};

             \draw[darkgray,dotted,very thick] (1.2,1) -- (3.35,1);
             \fill[darkgray] {(1.2,1) circle (0.05)};
             \node[right] at (3.4,1.05) {\small \textcolor{darkgray}{$\lproduc:  w\in \{a,b\}^* \mapsto |w|_a  |w|_b$}};
             
             \draw[darkgray,dotted,very thick] (1.4,0.5) -- (3.35,0.5);
             \fill[darkgray] {(1.4,0.5) circle (0.05)};
             \node[right] at (3.4,0.55) {\small \textcolor{darkgray}{$\sq : a^n \mapsto n^2$}};
             
             \draw[darkgray,dotted,very thick] (1,1.5) -- (3.35,1.5);
             \fill[darkgray] {(1,1.5) circle (0.05)};
             \node[right] at (3.4,1.55) {\small \textcolor{darkgray}{$\produc :  a^m b^n \mapsto mn$}};

             \draw[darkgray,dotted,very thick] (-0.5,2.7) -- (3.35,2.7);
             \fill[darkgray] {(-0.5,2.7) circle (0.05)};
             \node[right] at (3.4,2.75) {\small \textcolor{darkgray}{$\isq: a^{n_1} b a^{n_2} b \cdots  b a^{n_\ell} b \mapsto \sum_{i=1}^\ell  (n_i)^2$}};

             \draw[darkgray,dotted,very thick] (-0.3,2.2) -- (3.35,2.2);            
             \fill[darkgray] {(-0.3,2.2) circle (0.05)};
             \node[right] at (3.4,2.25) {\small \textcolor{darkgray}{$\trian :  a^{n_\ell} b a^{n_{\ell{-}1}} b \cdots  b a^{n_1} b \mapsto \sum_{i=1}^\ell i n_i$}};
             
              \draw[red,thick,->] (-1.2,1.55) -- (0.3,1.55);
            \node[right] at (-1.25,1.3) { $\substack{\text{\textcolor{red}{\textsc{decidable}}}\\ \text{\textcolor{red}{\textsc{membership}}} }$};

             %\draw[gray,dotted,thick] (0,3.5) -- (3.55,3.5);
            % \fill[gray!110] {(0,3.5) circle (0.05)};
             %\node[right] at (3.6,3.55) {\small \textcolor{gray}{$\textsf{iterated-cube}: a^{n_1} b a^{n_2} b \cdots  b a^{n_\ell} b \mapsto \sum_{i=1}^\ell  (n_i)^3$}};
             
        \end{tikzpicture}
        \end{center}
        
    \caption{\label{fig:conclu} Classes of functions with unary output and results of this paper}
\end{figure}
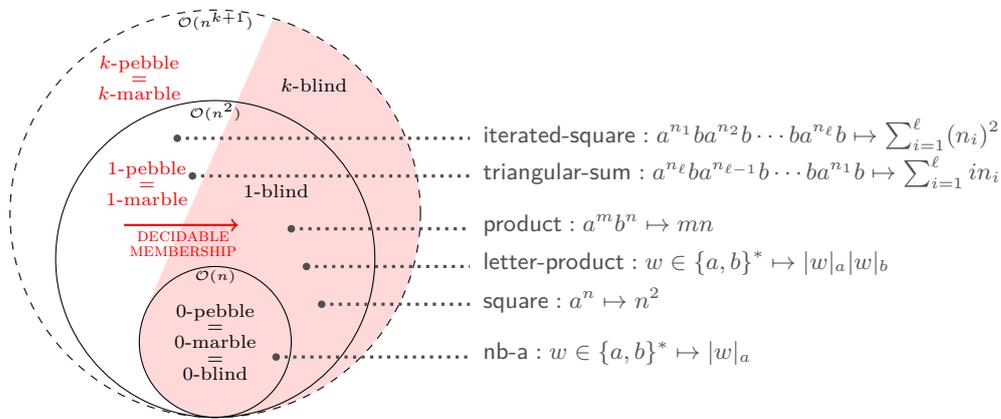

\subparagraph*{Outline.} We first recall in Section \ref{sec:prelim} the definitions
of $k$-pebble, $k$-marble and $k$-blind transducers, simplified for the case of unary outputs.
In Section \ref{sec:bimachines}, we define the notions of
$k$-pebble, $k$-marble and $k$-blind bimachines, and show their equivalence with the transducer
models. Bimachines are easier to handle in the proofs,
due to the fact that they avoid two-way moves.
In Section \ref{sec:marbles}, we show that $k$-marble and $k$-pebble transducers are equivalent.
Finally, we solve in Section \ref{sec:blind} the
class membership problem from $1$-pebble to $1$-blind. Due to space
constraints, several proofs are sketched in the main paper, and we focus on
the most significant lemmas and characterizations.

\section{Preliminaries}

\label{sec:prelim}

$\mb{N}$ is the set of nonnegative integers.
If $0 \le i \le j$, the set $[i{:}j]$ denotes $\{i, i{+}1, \dots, j\} \subseteq \Nat$
(empty if $j <i$).
Capital letters $A, B$ denote finite sets of letters (alphabets).
The empty word is denoted $\vide$. 
If $w \in A^*$, let $|w| \in \Nat$ be its length,
and for $1 \le i \le |w|$ let $w[i]$ be its $i$-th letter.
If $I = \{i_1 < \cdots < i_{\ell}\} \subseteq \{1, \dots, |w|\}$,
 let $w[I] \defined w[i_1] \cdots w[i_{\ell}]$.
 If $a \in A$, we denote by $|w|_a$ the number of letters $a$ occurring in $w$. 
 Given $A = \{a, \dots \}$, let $\marq{A} \defined \{\marq{a}, \dots\}$
be a disjoint copy of $A$.
For $1 \le i \le |w|$, we define $\omar{w}{i} \defined w[1{:}i{-1}]\marq{w[i]}w[i{+}1{:}|w|]$
as "$w$ in which position $i$ is underlined".
We assume that the reader is familiar with the basics of automata theory, in particular the notion of  two-way deterministic automaton.

\subparagraph*{Two-way transducers.} A deterministic two-way transducer is a deterministic two-way automaton enhanced with the ability to produce outputs along its run.  The class of functions described by these machines is known as \emph{regular functions} \cite{chytil1977serial, engelfriet2001mso}.

\begin{definition}[] \label{def:manymarbles}

	A \emph{(deterministic) two-way transducer}
	$(A,B,Q,q_0,F,\delta,\lambda)$ is:
	\item 
	\begin{itemize}
	\item an input alphabet $A$ and an output alphabet $B$;
	\item a finite set of states $Q$
	with an initial state $q_0 \in Q$
	and a set $F \subseteq Q$ of final states;
	\item a (partial) transition function $\delta: Q \times (A \uplus \{{\lmark}, {\rmark}\})
	\rightarrow Q \times \{\lmove, \rmove\}$;
	\item a (partial) output function $\lambda: Q \times (A \uplus \{{\lmark}, {\rmark}\}) \rightarrow B^*$ with same domain as $\delta$.

	\end{itemize}

\end{definition}

When given as input a word $w \in A^*$, the two-way transducer disposes of a read-only input tape containing ${\lmark} w {\rmark}$. The marks ${\lmark}$ and ${\rmark}$ are used to detect the borders of the tape, by convention we denote them as positions $0$ and $|w|{+}1$ of $w$.
Formally, a \emph{configuration} over  ${\lmark} w {\rmark}$ is a tuple $(q,i)$ where $q \in Q$ is the current state and $0 \le i \le |w|{+}1$ is the position of the reading head. The \emph{transition relation} $\rightarrow$ is defined as follows. Given a configuration $(q,i)$, let $(q',\star):= \delta(q,w[i])$. Then $(q, i) \rightarrow (q', i')$ whenever either $\star = \lmove$ and  $i' = i-1$ (move left), or $\star = \rmove$ and $i' = i+1$ (move right), with $0 \le i' \le |w|{+}1$. A \emph{run} is a sequence of configurations $(q_1,i_1) \rightarrow \cdots \rightarrow (q_n,i_n)$. Accepting runs are those that begin in $(q_0, 0)$ and end in a configuration of the form $(q, |w|{+}1)$ with $q \in F$ (and it never visits such a configuration before).
The function $f: A^* \rightarrow B^*$ computed by the machine is defined as follows.
Let $w \in A^*$, if there exists an accepting run on ${\lmark} w {\rmark}$,
then $f(w)$ is the concatenation of the $\lambda(q,w[i])$
along this unique run on ${\lmark} w {\rmark}$.
To make $f$ a total function, we let $f(w) \defined \epsilon$ if there is no accepting run
(the language of words having an accepting run in a two-way transducer is regular \cite{shepherdson1959reduction}, hence the domain does not matter).

\begin{example} \label{ex:mirror} $\mirror: A^* \rightarrow A^*, abac \mapsto caba$ can be computed by a two-way transducer.
\end{example}

From now on, the output alphabet of the machines will always
be a singleton. Up to identifying $\{a\}^*$ and $\Nat$, we assume that
 $\lambda: Q \times (A \uplus \{{\lmark}, {\rmark}\}) \rightarrow \Nat$
 and $f:A^* \rightarrow \Nat$.

\subparagraph*{External functions.} We now extend the notion
of output function $\lambda$: it will not give directly an integer, but performs
a call to an \emph{external function} which returns an integer.
For pebbles, the output of the external functions depends on the input word
and the current position.

\begin{definition}[] \label{def:two-external} 

A \emph{two-way transducer with external pebble functions}
$(A,Q,q_0,F,\delta,\oras,\lambda)$ is:

\item 

\begin{itemize}

\item an input alphabet $A$;

\item a finite set of states $Q$ with an initial state $q_0 \in Q$ and a set $F \subseteq Q$ of final states;

\item a (partial) transition function $\delta: Q \times (A \uplus \{{\lmark}, {\rmark}\}) \rightarrow Q \times \{\lmove, \rmove\}$;

\item a finite set $\oras$ of external functions  $\exte: (A \uplus \marq{A})^* \rightarrow \Nat$;

\item a (partial) output function $\lambda: Q \times (A \uplus \{{\lmark}, {\rmark}\}) \rightarrow \oras$ with same domain as $\delta$.

\end{itemize}

\end{definition}

Configurations $(q,i$) and runs of two-way transducers with external functions
are defined as for classical two-way transducers.
The function $f: A^* \rightarrow \Nat$ computed by the machine is defined as follows.
Let $w \in A^*$ such that there exists an accepting run on ${\lmark} w {\rmark}$.
If $\lambda(q,w[i]) = \exte \in \oras$,
we let $\nu(q,i) \defined \exte(\omar{w}{i})$,
that is the result of $\exte$ applied to $w$ marked in $i$.
Finally, $f(w)$ is defined as the sum of the $\nu(q,i)$ along this unique
accetping run on ${\lmark} w {\rmark}$.
We similarly set $f(w) = 0$ if there is no accepting run.

\begin{remark} If the external functions are constant,
we exactly have a two-way transducer.
\end{remark}

\begin{example}\label{ex:mul1} Let $a,b \in A$,
$\exte_b: w \in (A \uplus \marq{A})^* \mapsto |w|_{b}$
and $\exte_0: w \in (A \uplus \marq{A})^* \mapsto 0$.
The two-way transducer with external pebble functions,
which makes a single pass on its input and calls $\exte_b$ if reading $a$
and $\exte_0$ otherwise, computes $\lproduc: w \in A \mapsto |w|_a |w|_b$.
\end{example}

We define two other models. Their definition is nearly the same,
except that the external functions of $\oras$ have type $A^* \rightarrow \Nat$
and $\nu(q,i)$ is defined in a slightly different way:
\begin{itemize}
\item in a \emph{two-way transducer with external blind functions}, we define $\nu(q,i) \defined \exte(w)$.
The external function is applied to $w$ without marking the current position;
\item in a \emph{two-way transducer with external marble functions}, we define $\nu(q,i) \defined \exte(w[1{:}i])$. The external function is applied to the prefix of $w$ stopping at the current position.
\end{itemize}

\subparagraph*{Pebble, blind and marble transducers.}
We now describe the transducer models using the formalism
of external functions. These are not the original definitions
from \cite{bojanczyk2018polyregular,doueneau20,nguyen2020comparison},
but the correspondence is straightforward,
as soon as we know that pebble automata can only recognize regular
languages.

\begin{definition} \label{def:pt}
For $k \ge 0$, a \emph{$k$-pebble} (resp. \emph{$k$-blind, $k$-marble}) \emph{transducer} is:
\item
\begin{itemize}
\item if $k =0$, a two-way transducer;
\item if $k \ge 1$, a two-way transducer with external pebble (resp. blind, marble) functions
that are computed by $(k{-}1)$-pebble (resp. $(k{-}1)$-blind, $(k{-}1)$-marble) transducers.
\end{itemize}
\end{definition}

The intuitive behavior of a $1$-pebble transducer is
depicted in Figure \ref{fig:pebbles} in Introduction. 
We draw in Figure \ref{fig:blind} the behavior of a
$1$-blind transducer, which is the same except that 
the calling position is not marked for the machine
computing the external function.

\begin{figure}[h!]

\centering
\begin{tikzpicture}{scale=1}

	\newcommand{\couleur}{blue}
	\newcommand{\texte}{\small \bfseries \sffamily \mathversion{bold} }

	% Input 1
	
	\draw (-0.25,4.2) rectangle (8.25,4.5);
	\node[above] at (6,4.1) {$\substack{\text{Input word}}$};
	\node[above] at (0.05,4.1) {$\lmark$};
	\node[above] at (7.95,4.1) {$\rmark$};

	% Run 1
	\node[above] at (-1.8,3.55) {\textcolor{\couleur}{\texte Main machine}};
	\draw[-,thick,\couleur](0,4) -- (7,4);
	\draw[-,thick, \couleur] (7,4) arc (90:-90:0.1);
	\draw[-,thick,\couleur](7,3.8) -- (1.5,3.8);	
	\draw[-, thick,\couleur] (1.5,3.8) arc (90:270:0.1);
	\draw[-,thick,\couleur](1.5,3.6) -- (2.5,3.6);	
	\fill[fill = \couleur,even odd rule] (2.5,3.6) circle (0.08);
        \draw[->,very thick,\couleur,dashed](2.5,3.6) to[out= -120, in = 60] (0,2.85);

	% Input 2
	\draw (-0.25,3) rectangle (8.25,3.3);
	\node[above] at (6,2.9) {$\substack{\text{Input word}}$};
	\node[above] at (0.05,2.9) {$\lmark$};
	\node[above] at (7.95,2.9) {$\rmark$};

	% Run 2
	\node[above] at (-1.94,2.35) {\textcolor{red!80}{\texte Submachine}};
	\draw[-,thick,red!80](0,2.8) -- (3,2.8);
	\draw[-,thick, red!80] (3,2.8) arc (90:-90:0.1);
	\draw[-,thick,red!80](3,2.6) -- (1,2.6);	
	\draw[-, thick,red!80] (1,2.6) arc (90:270:0.1);
	\draw[-,thick,red!80](1,2.4) -- (5,2.4);	
	\draw[-,thick,dotted,red!80](5,2.4) -- (6,2.4);

\end{tikzpicture}
\caption{\label{fig:blind}  Behavior of a $1$-blind transducer}
\end{figure}
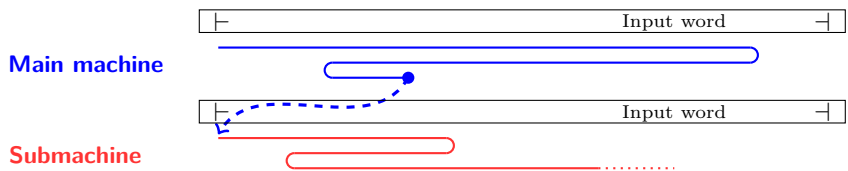

\begin{figure}[h!]

\centering
\begin{tikzpicture}{scale=1}

	\newcommand{\couleur}{blue}
	\newcommand{\texte}{\small \bfseries \sffamily \mathversion{bold} }

	% Input 1
	\draw (-0.25,4.2) rectangle (8.25,4.5);
	\node[above] at (6,4.1) {$\substack{\text{Input word}}$};
	\node[above] at (0.05,4.1) {$\lmark$};
	\node[above] at (7.95,4.1) {$\rmark$};

	% Run 1
	\node[above] at (-1.8,3.55) {\textcolor{\couleur}{\texte Main machine}};
	\draw[-,thick,\couleur](0,4) -- (3,4);
	\draw[-,thick, \couleur] (3,4) arc (90:-90:0.1);
	\draw[-,thick,\couleur](3,3.8) -- (1.5,3.8);	
	\draw[-, thick,\couleur] (1.5,3.8) arc (90:270:0.1);
	\draw[-,thick,\couleur](1.5,3.6) -- (4.7,3.6);	
	\fill[fill = \couleur,even odd rule] (4.7,3.6) circle (0.08);
        \draw[->,very thick,\couleur,dashed](4.7,3.6) to[out= -140, in = 50] (0,2.85);

	% Input 2
	\draw (-0.25,3) rectangle (5.25,3.3);
	\node[above] at (0.05,2.9) {$\lmark$};
	\node[above] at (5,2.9) {$\rmark$};

	% Run 2
	\node[above] at (-1.94,2.35) {\textcolor{red!80}{\texte Submachine}};
	\draw[-,thick,red!80](0,2.8) -- (4.7,2.8);
	\draw[-,thick, red!80] (4.7,2.8) arc (90:-90:0.1);
	\draw[-,thick,red!80](1,2.6) -- (4.7,2.6);	
	\draw[-, thick,red!80] (1,2.6) arc (90:270:0.1);
	\draw[-,thick,red!80](1,2.4) -- (2.5,2.4);	
	\draw[-,thick,dotted,red!80](3.5,2.4) -- (2.5,2.4);

\end{tikzpicture}

\caption{\label{fig:marbles} Behavior of a $1$-marble transducer}

\end{figure}
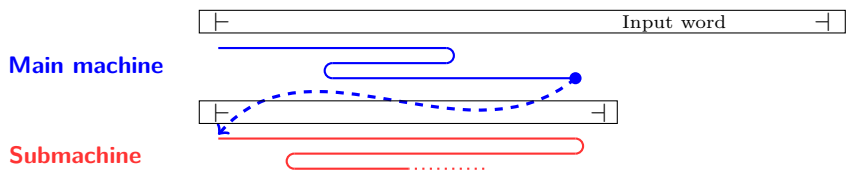

\begin{example} \label{ex:mul2} By restricting the functions $\exte_{b}$ and $\exte_0$
of Example \ref{ex:mul1} to $A^*$, we see that $\lproduc : w \in A^* \mapsto |w|_a |w|_b$
can be computed by a $1$-blind transducer.
\end{example}

The intuitive behavior of a $1$-marble transducer is depicted in Figure \ref{fig:marbles}.

\begin{example} \label{ex:mul3}
	The function  $\lproduc : w \in A^* \mapsto |w|_a |w|_b$ can
	be computed by a $1$-marble transducer as follows. Assume that $a \neq b$,
	let $\exte_{a}: w \mapsto |w|_a$ and $\exte_b: w \mapsto |w|_b$.
	The machine calls $\exte_b$ when reading $a$ and $\exte_a$
	when reading $b$. This way, each $a$ is "counted" $|w|_b$ times
	(from the call of $\exte_b$ starting in this $a$ which computes
	all the $b$ before it, plus each time it is seen in a call of $\exte_a$ starting
	from some $b$ after this $a$).
\end{example}

The strategy for computing $\lproduc$ is really different
between Examples \ref{ex:mul2} and \ref{ex:mul3}.
This illustrates the difficulty to obtain a "canonical" form for
a transduction.

\section{From two-way transducers to bimachines}

\label{sec:bimachines}

Since we consider a commutative output monoid, the order
in which the production is performed does not matter.
It is thus tempting to simplify a two-way transducer
in a one-way machine which visits each position only
once. This is exactly what we do with bimachines,
with the subtlety that they are able to check regular properties
of the prefix (resp. suffix) starting  (resp. ending) in the current position.
From now on, we  consider only total functions
of type $A^+ \rightarrow \Nat$ (the output on $\epsilon$
can be treated separately and does not matter).

\begin{definition}[] \label{def:bimachine-external}
	A \emph{bimachine with external pebble functions} $(A,M,\mu,\oras,\lambda)$ consists of:
	\item
	\begin{itemize}
	\item an input alphabet $A$;
	\item a morphism into a finite monoid $\mu: A^* \rightarrow M$;
	\item a finite set $\oras$ of external functions $\exte: (A \uplus \marq{A})^+ \rightarrow \Nat$;
	\item a total output function  $\lambda: M\times A \times M \rightarrow \oras$.
	\end{itemize}
\end{definition}

Given $1 \le i \le |w|$ a position of $w \in A^*$, let
 $\exte_i \defined \lambda(\mu(w[1{:}{i{-}1}]),w[i],\mu(w[i{+}1{:}|w|])) \in \oras$.
The bimachine defines a function $f: A^+ \rightarrow \Nat$ as follows:
\begin{equation*}
f(w) \defined \sum_{1 \le i \le |w|} \exte_i(\omar{w}{i}).
\end{equation*}

As before, we can define \emph{bimachines with external
blind} (resp. \emph{marble}) functions (in this case we have $\exte_i: A^+ \rightarrow \Nat$).
We then let:
\begin{equation*}
f(w) \defined \sum_{1 \le i \le |w|} \exte_i(w) \left(
\text{resp. }f(w) \defined \sum_{1 \le i \le |w|} \exte_i(w[1{:}i])\right).
\end{equation*}

As for two-way transducers, we define bimachines
(without external functions)
 by setting $\lambda: M \times A \times M \rightarrow \Nat$.
Equivalently, it corresponds to bimachines with external constant functions
$\exte: w \mapsto n$. Going further, we define $k$-pebble bimachines by induction.

\begin{definition}
For $k \ge 0$, a \emph{$k$-pebble} (resp. \emph{$k$-blind, $k$-marble}) \emph{bimachine} is:
\item
\begin{itemize}
\item if $k=0$, a bimachine (without external functions);
\item if $k \ge 1$, a bimachine with external pebble (resp. blind, marble)
 functions which are computed by
$(k{-}1)$-pebble (resp $(k{-}1)$-blind, $(k{-}1)$-marble) bimachines.
\end{itemize}

\end{definition}

\begin{example} \label{ex:trian} The function
$\trian:  a^{n_\ell} b a^{n_{\ell-1}} b \cdots  b a^{n_1} b \mapsto \sum_{i=1}^\ell i n_i$
can be computed by a $1$-marble bimachine.
It uses the singleton monoid $M = \{\neutral\}$ and
the morphism $\mu: a,b \mapsto \neutral$
in all its bimachines. The output function
of the main bimachine is defined by
$\lambda(\neutral,a,\neutral) \defined \exte_a$ and $\lambda(\neutral,b,\neutral) \defined \exte_b$.
For computing $\exte_a: w \mapsto 0$ we use output
$\lambda_{\exte_a}(\neutral,a,\neutral) \defined 0$ and
$\lambda_{\exte_a}(\neutral,b,\neutral) \defined 0$, and
for computing $\exte_b: w \mapsto |w|_a$ we use output
$\lambda_{\exte_b}(\neutral,a,\neutral) \defined 1$ and
$\lambda_{\exte_b}(\neutral,b,\neutral) \defined 0$.
\end{example}

Standard proof techniques allow to relate
bimachines and transducers.

\begin{proposition} \label{lem:pebble-bim}
 $k$-pebble (resp. $k$-blind, $k$-marble) bimachines
 and  $k$-pebble (resp. $k$-blind, $k$-marble) transducers
 compute the same functions, and both conversions are effective.
\end{proposition}

\begin{proof}[Proof sketch.]
Both directions are treated by induction.
From bimachines to transducers, 
we show that
a bimachine with external
pebble functions can be transformed in
an equivalent two-way transducer
with the same external pebble functions
(we use a \emph{lookaround} \cite{engelfriet2001mso}
to simulate $\mu$).
From transducers to bimachines,
the induction step shows that
a two-way transducer with external
pebble functions can be transformed in
an equivalent bimachine
with external pebble functions,
by adapting the classical reduction
from two-way to one-way automata \cite{shepherdson1959reduction}.
However the new external functions
can be linear combinations
of the former ones, since we
produce "all at once" the results of
several visits in a position.
We only need to use a finite number of
combinations, since in its accepting runs, a two-way transducer
can only visit each position a bounded number of times.
\end{proof}

\section{Equivalence between $k$-pebble and $k$-marble transducers}

\label{sec:marbles}

The main goal of this section is to show
equivalence between $k$-pebble and $k$-marble
transducers, over unary outputs.
We shall use another model
which is equivalent to marble transducers \cite{doueneau20}:
a \emph{streaming string transducer
(with unary output)}, which consists in a deterministic
automaton with a finite set $\regs$ of
registers that store integers.
At each letter read, the values of the registers
are updated by doing a linear combination
of their former values, whose
coefficients depend on the current state of the automaton.
In our definition we focus on the registers
and forget about the states, which corresponds to
a weighted automaton over the semiring $(\Nat, +, \times)$
(it is shown in \cite{doueneau20} that both models
are equivalent). The update is represented by a matrix
from $\Nat^{\regs \times \regs}$, which is chosen depending
on the letter read.

\begin{definition}A \emph{streaming string transducer} ($\SST$) $\trans = (A,\regs, I, T, F)$ is:
\item
\begin{itemize}

\item an input alphabet $A$ and a finite set $\regs$ of registers;

\item an initial row vector $I \in \Nat^{\regs}$;

\item a register update function $T: A \rightarrow \Nat^{\regs \times \regs}$;

\item an output column vector $F \in \Nat^\regs$.
\end{itemize}
\end{definition}

$T$ can be extended as a monoid morphism from $A^*$ to $(\Nat^{\regs \times \regs}, \times)$.
Given $w \in A^*$, the vector $I T(w)$ intuitively describes
the values of the registers after reading $w$.
To define the function $f: A^* \rightarrow \Nat$ computed by $\trans$,
we combine these values by the output vector:
\begin{equation*}
f(w) \defined I T(w) F.
\end{equation*}

\begin{example} The function
$\trian:  a^{n_\ell} b a^{n_{\ell-1}} b \cdots  b a^{n_1} b \mapsto \sum_{i=1}^\ell i n_i$
can be computed by an $\SST$. We use two registers $x,y$ and allow constants in the updates for
more readability: $x$ is initialized to $0$ and updated $x \leftarrow x+1$ on $a$ and $x \leftarrow x$ on $b$,
and $y$ is initialized to $0$ and updated $y \leftarrow y$ on $a$ and $y \leftarrow y+x$ on $b$.
Finally we output $y$.
\end{example}

We are now ready to state the main results of this section.

\begin{theorem} \label{theo:pebble:sst}
Given a $k$-pebble bimachine, 
one can build an equivalent  $\SST$.
\end{theorem}

The proof is done by induction on $k \ge 0$.
Consider a bimachine whose external functions
are computed by $(k{-}1)$-pebble bimachines.
By hypothesis, we can compute these functions by $\SSTs$.
The induction step is shown by Lemma \ref{lem:bimoversst},
which uses new proof techniques.

\begin{lemma} \label{lem:bimoversst}
Given a bimachine with external pebble functions
computed by $\SSTs$, one
can build an equivalent $\SST$
(with no external functions).
\end{lemma}

\begin{proof}[Proof idea.]
Let $\trans$ be the $\SST$ computing
an external function $\exte$.
On input $w \in A^+$, the bimachine calls $\exte$ on
several positions $1\le i_1 < \cdots < i_{\ell} \le |w|$,
which induces executions of $\trans$ on $\omar{w}{i_1}, \dots, \omar{w}{i_\ell}$.
These executions are very similar: they only
differ when reading the marked letter.
Thus we build an $\SST$ which computes "simultaneously"
all these executions, by keeping track of the sum
of the values of the registers of $\trans$ along them.
\end{proof}

As a consequence of Theorem \ref{theo:pebble:sst},
we obtain equivalence
between pebbles and marbles over unary outputs.
The result is false over non-unary output alphabets
\cite{doueneau20,nguyen2020comparison}.
We also relate these functions with
 those computed by $\SST$, assuming that
 the output is bounded by a polynomial 
 in the input's length.

\begin{corollary} \label{cor:pebmar}For all $k \ge 0$ and $f:A^* \rightarrow \Nat$, the following conditions
are equivalent:
\item
\begin{enumerate}
\item \label{po:pebble} $f$ is computable by a $k$-pebble transducer;
\item \label{po:marble} $f$ is computable by a $k$-marble transducer;
\item \label{po:sst} $f$ is computable by an $\SST$ and $f(w) = \mc{O}(|w|^{k+1})$.
\end{enumerate}
Furthermore the transformations are effective.
\end{corollary}

\begin{proof}
Clearly a $k$-pebble transducer can
simulate a $k$-marble transducer,
hence $\ref{po:marble} \Rightarrow \ref{po:pebble}$.
Let $f$ be computed by a $k$-pebble transducer,
we have $f(w) = \mc{O}(|w|^{k+1})$ and
by Theorem \ref{theo:pebble:sst} one can build an $\SST$ for $f$.
Thus $\ref{po:pebble} \Rightarrow \ref{po:sst}$.
Finally $\ref{po:sst} \Rightarrow \ref{po:marble}$
is shown in \cite{doueneau20}.
\end{proof}

Another important consequence is that we can decide equivalence
of pebble transducers with unary output,
since we can do so for marble transducers \cite{doueneau20}.

\begin{corollary} One can decide if two pebble transducers
compute the same function.
\end{corollary}

This has been an open question since \cite{bojanczyk2018polyregular},
and it is still open for generic output alphabets.

\section{Deciding if $1$-pebble is $1$-blind}

\label{sec:blind}

Since the equivalence between marbles and
pebbles is established, we now compare
$1$-pebble (which are $1$-marble) transducers
with $1$-blind transducers.
It turns out that $1$-pebble are strictly more expressive;
 the main goal of this section is to show Theorem \ref{theo:membership}.

\begin{theorem}[Membership] \label{theo:membership} One can decide if a function given
by a $1$-marble (or $1$-pebble) transducer can be computed by
a $k$-blind transducer for some $k \ge 0$.
If this condition holds, one can
 build a $1$-blind transducer which computes it.
\end{theorem}

Let us fix a function $f: A^+ \rightarrow \Nat$
described by a $1$-marble bimachine
$\trans = (A,M, \mu, \oras, \lambda)$.
For $\exte \in \oras$, let $\trans_{\exte} \defined (A,M, \mu, \lambda_{\exte})$
be the bimachine which computes it.
We enforce the morphism $\mu$ to be surjective (up to considering the co-restriction to its image)
and the same in all machines (up to taking the product of all morphisms used).
Our goal is to give a decidable condition on $\trans$
for $f$ to be computable by a $1$-blind transducer.
For this purpose, we define the notion of \emph{bitype}.
Intuitively, it describes two disjoint factors in an input word,
together with a finite abstraction of their "context".

Let $\mh \defined 3|M|$ (it will be justified by Theorem \ref{theo:simon}).

\begin{definition} \label{def:bitype}
A \emph{bitype} $\Phi \defined m \cro{\uu} m' \cro{\ud} m''$
consists in $m,m',m'' \in M$, $\uu, \ud \in A^+$.
\end{definition}

We can define "the production performed in $\uu$ by the calls from $\ud$", in $\Phi$.
For $1 \le i \le |\uu|$ and $1 \le j \le |\ud|$, let
$\Phi(i,j) \defined \lambda_{\exte_{j}}(m\mu(\uu[1{:}{i}{-}1]), \uu[i], \mu(\uu[i{+1}{:}|\uu|])m'\mu(\ud[1{:}j])) \in \Nat $
where $\exte_{j} \defined \lambda(m\mu(\uu)m'\mu(\ud[1{:}{j}{-}1]), \ud[j],\mu(\ud[j{+1}{:}|\ud|])m'')$. Then we set:
\begin{equation*}
 \pro{m \cro{\uu} m' \cro{\ud}m''} \defined \sum_{\substack{1 \le i \le |\uu| \\ 1 \le j \le |\ud|}}
\Phi(i,j) \in \Nat.
\end{equation*}

\begin{definition} \label{def:symmetrical} The $1$-marble bimachine $\trans$
is \emph{symmetrical} whenever
 $\forall m,n,m_1,n_1,m_2,n_2 {\in} M$ and $\uu, \ud \in A^+$
such that $|\uu|, |\ud| \le 2^{\mh}$, $e_1 {\defined}\mu(\uu)$,  $e_2 {\defined} \mu(\ud)$
and $e {\defined} m_1 e_1 n_1 {=} m_2 e_2 n_2$ are idempotents,
there exists $K \ge 0$ such that  $\forall p \in M$:
\item
\begin{itemize}
\item  if
$m_1 e_1 p e_2 n_2 = e$, $em_1e_1p e_2= em_2e_2$ and $e_1p e_2 n_2 e= e_1n_1e$,

then $\pro{mem_1 e_1 \cro{\uu} e_1 p e_2 \cro{\ud} e_2 n_2 en} = K$;

\item if
$m_2 e_2 p e_1 n_1 = e$, $em_2e_2pe_1 = em_1e_1$ and $ e_2pe_1 n_1e = e_2n_2e$,

then $\pro{mem_2 e_2 \cro{\ud} e_2 p e_1 \cro{\uu} e_1 n_1 en} = K$.
\end{itemize}

\end{definition}

\begin{figure}[h!]
\centering
\begin{subfigure}{1\linewidth}
\centering\begin{tikzpicture}[scale=0.7]

	% Input 1

	\draw (-4.5,5) rectangle (4.5,4.5);
	\node[above] at (0,4.45) {$e$};
	\draw (-6.5,5) rectangle (-4.5,4.5);
	\node[above] at (-5.5,4.45) {$e$};
	\draw (6.5,5) rectangle (4.5,4.5);
	\node[above] at (5.5,4.45) {$e$};
	\draw (-8.5,5) rectangle (-6.5,4.5);
	\node[above] at (-7.5,4.45) {$m$};
	\draw (8.5,5) rectangle (6.5,4.5);
	\node[above] at (7.5,4.45) {$n$};

	\draw (-4.5,4.5) rectangle (-3.5,4);
	\node[above] at (-4,3.9) {$m_1$};	
	\draw (-3.5,4.5) rectangle (-2.5,4);
	\node[above] at (-3,3.9) {$e_1$};	
		\draw[fill=red!80]  (-2.5,3.5) rectangle (-1.5,4);
		\node[above] at (-2,3.4) {$u_1$};
	\draw (-2.5,4.5) rectangle (-1.5,4);
	\node[above] at (-2,3.9) {$e_1$};
	\draw (-1.5,4.5) rectangle (-0.5,4);
	\node[above] at (-1,3.9) {$e_1$};	
	
	\draw (-0.5,4.5) rectangle (0.5,4);
	\node[above] at (0,3.9) {$p$};	
	
	\draw (4.5,4.5) rectangle (3.5,4);	
	\node[above] at (4,3.9) {$n_2$};	
	\draw (3.5,4.5) rectangle (2.5,4);
	\node[above] at (3,3.9) {$e_2$};
		\draw[fill=blue!70]  (2.5,3.5) rectangle (1.5,4);
		\node[above] at (2,3.4) {$u_2$};	
	\draw (2.5,4.5) rectangle (1.5,4);
	\node[above] at (2,3.9) {$e_2$};	
	\draw (1.5,4.5) rectangle (0.5,4);
	\node[above] at (1,3.9) {$e_2$};
	
	\draw[thick,dashed](-1.5,3.3) --(-1.5,3) -- (6.5,3) -- (6.5,4.3) ;
	\node[above] at (3.5,2.4) {$e_1 n_1 e$};
	
	\draw[thick,dashed](1.5,2.8) --(1.5,2.5) -- (-6.5,2.5) -- (-6.5,4.3);
	\node[above] at (-3.5,2.4) {$e m_2 e_2$};

\end{tikzpicture}

\subcaption{Bitype $mem_1 e_1 \cro{\uu} e_1 p e_2 \cro{\ud} e_2 n_2 en$.}
\end{subfigure}

\vspace{1\baselineskip}

\begin{subfigure}{1\linewidth}
\centering\begin{tikzpicture}[scale=0.7]

	% Input 1

	\draw (-4.5,5) rectangle (4.5,4.5);
	\node[above] at (0,4.45) {$e$};
	\draw (-6.5,5) rectangle (-4.5,4.5);
	\node[above] at (-5.5,4.45) {$e$};
	\draw (6.5,5) rectangle (4.5,4.5);
	\node[above] at (5.5,4.45) {$e$};
	\draw (-8.5,5) rectangle (-6.5,4.5);
	\node[above] at (-7.5,4.45) {$m$};
	\draw (8.5,5) rectangle (6.5,4.5);
	\node[above] at (7.5,4.45) {$n$};

	\draw (-4.5,4.5) rectangle (-3.5,4);
	\node[above] at (-4,3.9) {$m_2$};	
	\draw (-3.5,4.5) rectangle (-2.5,4);
	\node[above] at (-3,3.9) {$e_2$};	
		\draw[fill=blue!70]  (-2.5,3.5) rectangle (-1.5,4);
		\node[above] at (-2,3.4) {$u_2$};
	\draw (-2.5,4.5) rectangle (-1.5,4);
	\node[above] at (-2,3.9) {$e_2$};
	\draw (-1.5,4.5) rectangle (-0.5,4);
	\node[above] at (-1,3.9) {$e_2$};	
	
	\draw (-0.5,4.5) rectangle (0.5,4);
	\node[above] at (0,3.9) {$p$};	
	
	\draw (4.5,4.5) rectangle (3.5,4);	
	\node[above] at (4,3.9) {$n_1$};	
	\draw (3.5,4.5) rectangle (2.5,4);
	\node[above] at (3,3.9) {$e_1$};
		\draw[fill=red!80]  (2.5,3.5) rectangle (1.5,4);
		\node[above] at (2,3.4) {$u_1$};	
	\draw (2.5,4.5) rectangle (1.5,4);
	\node[above] at (2,3.9) {$e_1$};	
	\draw (1.5,4.5) rectangle (0.5,4);
	\node[above] at (1,3.9) {$e_1$};
	
	\draw[thick,dashed](-1.5,3.3) --(-1.5,3) -- (6.5,3) -- (6.5,4.3) ;
	\node[above] at (3.5,2.4) {$e_2 n_2 e$};
	
	\draw[thick,dashed](1.5,2.8) --(1.5,2.5) -- (-6.5,2.5) -- (-6.5,4.3);
	\node[above] at (-3.5,2.4) {$e m_1 e_1$};

\end{tikzpicture}

\subcaption{Bitype $mem_2 e_2 \cro{\ud} e_2 p e_1 \cro{\uu} e_1 n_1 en$}
\end{subfigure}

\caption{\label{fig:bitypes-figg}  The bitypes used to define a symmetrical $1$-marble bimachine}
\end{figure}
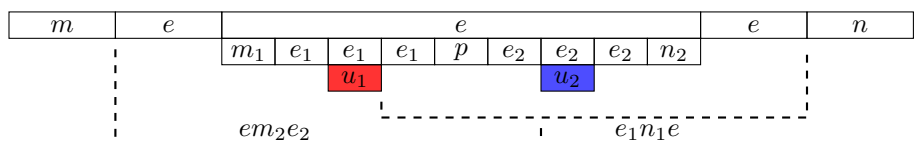
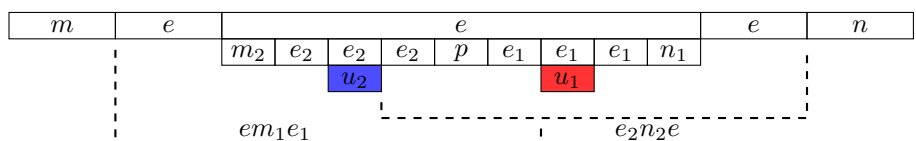

Symmetry means that, under some idempotent conditions,
$\pro{m \cro{\uu} m' \cro{\ud}m''}$ only depends on $m,m'', m' \mu(\ud) m''$ and $m \mu(\uu) m'$,
that are the "contexts" of $\uu$ and $\ud$, but not on the element $m'$
which separates them. The same holds if we swap $\uu$ and $\ud$.
The bitypes considered in Definition \ref{def:symmetrical} are depicted
in Figure \ref{fig:bitypes-figg}, together with the equations they satisfy.

Symmetry is the decidable condition we are looking for, as shown in Theorem \ref{theo:characterization}.
Recall that $f$ is the function computed by the $1$-marble bimachine $\trans$.

\begin{theorem}[Characterization] \label{theo:characterization} The following conditions are equivalent:
\item
\begin{enumerate}
\item \label{it:unt} $f$ is computable by a $k$-blind transducer for some $k \ge 0$;
\item \label{it:det} $f$ is computable by a $1$-blind transducer;
\item \label{it:trt} $\trans$ is symmetrical.
\end{enumerate}
\end{theorem}

Theorem \ref{theo:membership} follows from Theorem \ref{theo:characterization}, since
it suffices to check whether the machine is symmetrical, which can be decided
by ranging over all monoid elements (including idempotents) and  words of length at most $2^{\mh}$.

\begin{example} Let us show that the bimachine of Example \ref{ex:trian}
computing $\trian$ is not symmetrical.
Let $\uu \defined a, \ud \defined b$,
$m,n,m_1,n_1,m_2,n_2,p,e = \neutral$, $e_1 {\defined}\mu(\uu) = \neutral$
and $e_2 {\defined} \mu(\ud) = \neutral$. 
Then $\pro{mem_1 e_1 \cro{\uu} e_1 p e_2 \cro{\ud} e_2 n_2 en} = \pro{\neutral \cro{a} \neutral \cro{b} \neutral} = 1$
and  $\pro{mem_2 e_2 \cro{\ud} e_2 p e_1 \cro{\uu} e_1 n_1 en} = \pro{\neutral \cro{b} \neutral \cro{a} \neutral} = 0$.
Furthermore the equations of Definition \ref{def:symmetrical} hold,
thus $\trian$ is not computable by a $k$-blind bimachine.
\end{example}

Lemma \ref{lem:blind-symm} shows $\ref{it:unt} \Rightarrow \ref{it:trt}$ in Theorem \ref{theo:characterization}.
It allows to show that some function
cannot be computed by a $k$-marble transducer.
Its proof is technical; a coarse intuition is that 
a $1$-blind bimachine which makes a production
on $\uu$ when called from
$\ud$ cannot see the monoid element $m'$ between $\uu$
and $\ud$ (since $\ud$ is not marked, its position is "forgotten").

\begin{lemma} \label{lem:blind-symm} If $f$ is computable by a $k$-blind bimachine,
then $\trans$ is symmetrical.
\end{lemma}

Since $\ref{it:det} \Rightarrow \ref{it:unt}$ in Theorem \ref{theo:characterization} is obvious, it
remains to show that if $\trans$ is symmetrical, then
$f$ is effectively computable by a $1$-blind bimachine.
This is the goal of the two following subsections.
The main tool for the proof is the notion of factorization forest:
using \mbox{Lemma \ref{lem:nodes-prod}}, it allows us to compute the function $f$
without directly referring to a machine.

\subsection{Factorization forests}

Recall that $\mu:A^+ \rightarrow M$ is a fixed monoid morphism.
A \emph{factorization forest} \cite{bojanczyk2009factorization}
of $w \in A^+$ is an unranked tree structure which decomposes $w$
following the image of its factors by $\mu$.

\begin{definition}[\cite{simon1990factorization,bojanczyk2009factorization}] \label{def:facto}
A \emph{factorization (forest)} of $w \in A^+$ is a tree defined as follows:
\item
\begin{itemize}
\item if $w = a \in A$, it is a leaf $a$;
\item if $|w| \ge 2$, then $(\forest_1) \cdots (\forest_n)$ is a factorization of
$w$ if each $\forest_i$ is a factorization of some $w_i \in A^+$
such that $w = w_1 \cdots w_n$, and either:
\item[]
\begin{itemize}
\item $n=2$: the root is a \emph{binary node};
\item or $n \ge 3$ and $\mu(w_1) = \dots = \mu(w_n)$ is an idempotent:
the root is an \emph{idempotent node}.
\end{itemize}
\end{itemize}
\end{definition}

The set of factorizations over $w$
is denoted $\facto{\mu}{w}$. Recall that $\mh = 3|M|$.

\begin{theorem}[\cite{simon1990factorization,bojanczyk2009factorization}] \label{theo:simon}
For all $w \in A^+$, there is $\forest \in \facto{\mu}{w}$ of height at most $\mh$.
\end{theorem}

Let $\apar{A} \defined A \uplus \{(,)\}$.
We have defined $\facto{\mu}{w}$ as a set of tree structures,
but we can assume that $\facto{\mu}{w} \subseteq \apar{A}^+$.
Indeed, in Definition \ref{def:facto}, 
a factorization of $w$ can also be seen as "the word $w$ with parentheses".
There exists a rational function which computes factorizations,
under this formalism.
We reformulate this statement in Proposition \ref{prop:effective-factorizations}
using a two-way transducer (which, exceptionally
in this paper, has a non-unary output alphabet $\apar{A}$).

\begin{proposition}[Folklore] \label{prop:effective-factorizations}
One can build a two-way transducer which computes a 
function $A^+ \rightarrow \apar{A}^+, w \mapsto \forest \in \facto{\mu}{w}$
for some $\forest$ of height at most $\mh$.
\end{proposition}

We denote by $\nodes{\forest}$ the set of (idempotent or binary) nodes of $\forest$.
In order to simplify the statements, we identify a node with
the subtree rooted in this node. Thus $\nodes{\forest}$ can also
be seen as the set of subtrees of $\forest$, and $\forest \in \nodes{\forest}$.
We shall use the standard tree vocabulary of "height" (a leaf is a tree of
height $1$), "parent node", "descendant" and "ancestor"
(defined in a non-strict way: a node is itself one of its ancestors),
"branch", etc.

\begin{example} \label{ex:factorization} Let $A = \{a,b,c\}$,
$M {=} \{\neutral,\deutral,3_M\}$ with $\deutral^2 {=} \neutral$, $3_M$ absorbing,
$\mu(a) {\defined} \deutral$
and $\mu(b) {\defined} \mu(c) {\defined} 3_M$.
Then
$\forest \defined (aa)(bc(a(cbbcb))b) \in \facto{\mu}{aabcaccbcbcb}$
(we dropped the parens around single letters for more readability)
is depicted in Figure \ref{fig:iterable}.
Idempotent nodes are drawn
using a horizontal line.
\end{example}

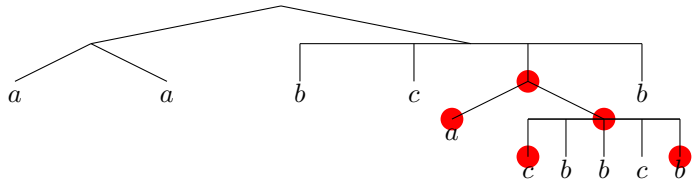
\begin{figure}[h!]

\centering
\begin{tikzpicture}{scale=1}

	\newcommand{\couleur}{blue}
	\newcommand{\texte}{\small \bfseries \sffamily \mathversion{bold} }

	% Input 1
	\fill[fill=red]  (4.75,8.5)  circle (0.15);
	\fill[fill=red]  (3.75,8)  circle (0.15);
	\fill[fill=red]  (5.75,8)  circle (0.15);
	\fill[fill=red]  (4.75,7.5)  circle (0.15);
	\fill[fill=red]  (6.75,7.5)  circle (0.15);
		
	\draw (1.5,9.5) -- (-1,9);
	\draw (1.5,9.5) -- (4,9);
	\draw (1.75,9) -- (6.25,9);
	\draw (6.25,8.5) -- (6.25,9);
	\node[above] at (6.25,8.1) {$b$};
	\draw (4.75,8.5) -- (4.75,9);

	\draw (3.25,8.5) -- (3.25,9);
	\node[above] at (3.25,8.1) {$c$};
	\draw (1.75,8.5) -- (1.75,9);
	\node[above] at (1.75,8.1) {$b$};

	\draw (-1,9) -- (-2,8.5);
	\node[above] at (-2,8.1) {$a$};
	\draw (-1,9) -- (0,8.5);
	\node[above] at (0,8.1) {$a$};
	
	\draw (4.75,8.5) -- (3.75,8);
	\node[above] at (3.75,7.6) {$a$};	
	\draw (4.75,8.5) -- (5.75,8);

	\draw (4.75,8) -- (6.75,8);
	\draw (4.75,8) -- (6.75,8);
	
	\draw (6.75,7.5) -- (6.75,8);
	\node[above] at (6.75,7.1)  {$b$};	
	\draw (6.25,7.5) -- (6.25,8);
	\node[above] at (6.25,7.1)  {$c$};	
	\draw (5.75,7.5) -- (5.75,8);
	\node[above] at (5.75,7.1)  {$b$};	
	\draw (5.25,7.5) -- (5.25,8);
	\node[above] at (5.25,7.1) {$b$};	
	\draw (4.75,7.5) -- (4.75,8);
	\node[above] at  (4.75,7.1)  {$c$};

\end{tikzpicture}

\caption{\label{fig:iterable}  The factorization $(aa)(bc(a(cbbcb))b)$ of $aabcacbbcbb$}

\end{figure}

We define $\itera{\forest} \subseteq \nodes{\forest}$
as the set of nodes which are the middle child 
of an idempotent node. Intuitively, such nodes
can be copied without modifying their "context".

\begin{definition} Let $\forest \in \facto{\mu}{w}$, we define the set of iterable nodes
of $\forest$ by induction:
\item
\begin{itemize}
\item if $\forest = a\in A$ is a leaf, $\itera{\forest} \defined \varnothing$;
\item if $\forest = (\forest_1) \cdots (\forest_n)$ is a binary or idempotent node, then:
\begin{equation*}
\itera{\forest} \defined \{\forest_i: 2 \le i \le n{-}1\}  \biguplus_{1 \le i \le n} \itera{\forest_i}.
\end{equation*}
\end{itemize}
\end{definition}

On the contrary, we now define sets of nodes which
cannot be duplicated individually.

\begin{definition} Let $\forest \in \facto{\mu}{w}$, we define
the \emph{dependency} of $\forest$ as follows:
\item
\begin{itemize}
\item if $\forest = a\in A$ is a leaf, then $\dep{\forest} \defined \{ a \}$;
\item if $\forest = (\forest_1) \cdots (\forest_n)$ is binary or idempotent, then
$\dep{\forest} \defined \{\forest\} \cup \dep{\forest_1} \cup \dep{\forest_n}$.
\end{itemize}
\end{definition}

Intuitively, $\dep{\forest} \subseteq \nodes{\forest}$ contains all the nodes of $\forest$
except those which are descendant of a middle child.
If $\nodi \in \nodes{\forest}$, we consider
$\dep{\nodi} \subseteq \nodes{\nodi}$ as a subset of $\nodes{\forest}$.
We then define the frontier of $\nodi$, denoted
$\fr{\nodi}{\forest} \subseteq \{1, \dots, |w|\}$
as the set of positions of $w$ which belong to $\dep{\nodi}$
(when seen as leaves of $\forest$).

\begin{example} In Figure \ref{fig:iterable},
the top-most red node $\nodi$ is iterable.
Furthermore $\dep{\nodi}$ is the set of red nodes,
$\fr{\nodi}{\forest} = \{5,6,10\}$
and $w[\fr{\nodi}{\forest} ] = acb$.
\end{example}

The relationship between iterable nodes and dependencies is detailed below.
We denote by $\parti{\forest} \defined \itera{\forest} \uplus  \{\forest\} $,
the set of iterable nodes plus the root.

\begin{lemma} \label{lem:partition} Let $\forest \in \facto{\mu}{w}$, then
$ \{\dep{\nodi}: \nodi \in \parti{\forest}\}$ is a partition of $\nodes{\forest}$;
and $ \{\fr{\nodi}{\forest}: \nodi \in \parti{\forest}\}$
is a partition of $\{1, \dots, |w|\}$.
\end{lemma}

We define $\pro{i, j}$ in $w$ as "the production performed
in $i$ when called from $j$".

\begin{definition} Let $w \in A^+$ and $1 \le i \le j \le |w|$ two positions
of $w$. We define $\pro{i, j} \in \Nat$ as $\lambda_{\exte_{j}}(\mu(w[1{:}{i}{-}1]), w[i], \mu(w[i{+1}{:}j]))$,
where $\exte_{j} \defined \lambda(\mu(w[1{:}{j}{-}1]), w[j],\mu(w[j{+1}{:}|w|]))$.
\end{definition}

We extend this definition to pairs of nodes:
given $\nodi, \nodj \in \nodes{\forest}$, we define $\pro{\nodi, \nodj}$
 "the sum of all productions performed in the frontier of $\nodi$,
when called from the frontier of $\nodj$" as follows
(we have to ensure that the calling positions are "on the right").

\begin{definition} \label{def:prodnodes} Let $w \in A^+$,
$\forest \in \facto{\mu}{w}$ and $\nodi, \nodj \in \nodes{\forest}$. We define:
\begin{equation*}
 \pro{\nodi,\nodj} \defined \sum_{\substack{i \in \fr{\nodi}{\forest} \\ j \in \fr{\nodj}{\forest} \\ {i \le j}}}
\pro{i,j} \in \Nat.
\end{equation*}
\end{definition}

If $\nodi$ is an ancestor of $\nodj$ (or the converse)
then $\fr{\nodi}{\forest}$ and $\fr{\nodj}{\forest}$ are interleaved, hence
we can have both $\pro{\nodi, \nodj} \neq 0$ and $\pro{\nodj, \nodi} \neq 0$.
However, if $\nodi$ and $\nodj$ are not on the same branch,
we have either $\pro{\nodi, \nodj} = 0$ or $\pro{\nodj, \nodi} = 0$.

Applying Lemma \ref{lem:partition}, it is not hard to compute
$f(w)$ using the $\pro{\nodi, \nodj}$.

\begin{lemma} \label{lem:nodes-prod}
Let $w \in A^+$, $\forest \in \facto{\mu}{w}$. Then:
\begin{equation*}
f(w) =  \sum_{\nodi ,\nodj \in \parti{\forest}} \pro{\nodi,\nodj}.
\end{equation*}
\end{lemma}

\subsection{Typology of pairs of nodes}

We intend to compute (if possible) $f$
using a $1$-blind transducer. Following Lemma \ref{lem:nodes-prod},
it is enough to consider the productions performed on the pairs of nodes
of a factorization.
For this study, we split the pairs depending on their relative position in the tree.

\subparagraph*{Pairs separated by the frontier of the root.}
The frontier of the root $\fr{\forest}{\forest}$ plays a very specific
role with respect to blind transducers. Indeed,
over factorizations of height at most $\mh$, the size
of the frontier is bounded, hence it
splits the word in a bounded number of
distinguishable "blocks". Formally, we define the
notion of \emph{basis}.

\begin{definition} An idempotent node is a 
\emph{basis} if it belongs to the dependency of the root.
\end{definition}

The following result is shown by induction.

\begin{lemma} \label{lem:dep-root} Let $w \in A^+$ and $\forest \in \facto{\mu}{w}$.
Given $\nodi \in \itera{\forest}$, there exists a unique basis,
 denoted $\basis{\nodi}{\forest}$, such that $\nodi$ is the 
descendant of a middle child of $\basis{\nodi}{\forest}$.
\end{lemma}

\begin{definition} Given $w \in A^+$ and $\forest \in \facto{\mu}{w}$, we define
$D(\forest) \subseteq \parti{\forest} {\times} \parti{\forest}$ by:
\begin{equation*}
D(\forest) \defined \{(\nodi, \nodj): \nodi, \nodj \in \itera{\forest} \text{ and } \basis{\nodi}{\forest} \neq \basis{\nodj}{\forest}\}.
\end{equation*}
\end{definition}

Intuitively $\basis{\nodi}{\forest} \neq \basis{\nodj}{\forest}$ means that
$\fr{\nodi}{\forest}$ and $\fr{\nodj}{\forest}$
belong to two different "blocks" of the input. Lemma \ref{lem:distinguishable} is shown
by building a $1$-blind bimachine which visits successively each basis $\nodb$,
and for each iterable $\nodj$ such that $\basis{\nodj}{\forest} = \nodb$, calls a submachine which visits
the $\nodi$ such that  $\basis{\nodi}{\forest} \neq \nodb$ and produces $\pro{\nodi, \nodj}$.
The key element for doing this operation without pebbles 
is that the number of bases is bounded.

\begin{lemma} \label{lem:distinguishable} One can build a $1$-blind bimachine computing:
\begin{equation*}
f_D: (\apar{A})^+ \rightarrow \Nat, \forest \mapsto
\left\{
    \begin{array}{l}
        \displaystyle \sum_{(\nodi,\nodj) \in D(\forest)} \pro{\nodi,\nodj} \text{ if } \forest \text{ factorization of height at most } \mh; \\
        0 \text{ otherwise.}\\
    \end{array}
\right.
\end{equation*}

\end{lemma}

\subparagraph*{Linked pairs.}
Let $U(\forest) \defined \parti{\forest} {\times} \parti{\forest} \smallsetminus D(\forest)$,
it corresponds to the pairs of $\itera{\forest}$ which have the same basis,
plus all the pairs $(\forest, \nodi)$ and $(\nodi, \forest)$ for $\nodi \in \parti{\forest}$.
We now study the pairs of $U(\forest)$ which are 
"linked", in the sense that one node is (nearly) the
ancestor of the other.

\begin{definition} Let $w \in A^+$, $\forest \in \facto{\mu}{w}$. Let $L(\forest)$ be the
set of all $(\nodi, \nodj)\in U(\forest)$ such that $\nodi$ (or $\nodj)$
is either the ancestor of, or the right/left sibling of an ancestor
of $\nodj$ (or $\nodi$).
\end{definition}

In particular, we have
$(\forest, \forest), (\nodi, \forest), (\forest, \nodi) , (\nodi, \nodi) \in L(F)$ for all $\nodi \in \parti{\forest}$.
If $\forest$ has height at most $\mh$, there are at most $3 \mh$
nodes which are either an ancestor or the right/left sibling
of an ancestor of $\nodi$. Lemma \ref{lem:linked} follows from this boundedness.

\begin{lemma} \label{lem:linked} One can build a $0$-blind bimachine computing:
\begin{equation*}
f_L: (\apar{A})^+ \rightarrow \Nat, \forest \mapsto
\left\{
    \begin{array}{l}
        \displaystyle \sum_{(\nodi,\nodj) \in L(\forest)} \pro{\nodi,\nodj} \text{ if } \forest \text{ factorization of height at most } \mh; \\
        0 \text{ otherwise.}\\
    \end{array}
\right.
\end{equation*}

\end{lemma}

\subparagraph*{Independent nodes.}
The remaining sum is the most interesting,
since it is the only case where we use the
assumption that $\trans$ to be symmetrical (and this assumption is crucial).
Let $\forest \in \facto{\mu}{w}$,
we define the set $I(\forest) \defined U(\forest) \smallsetminus L(\forest)$.
It contains the pairs $(\nodi,\nodj)$ of iterable nodes such that
$\basis{\nodi}{\forest} = \basis{\nodj}{\forest}$ (i.e. they descend
from a common "big" idempotent), and
$\nodi$ (or $\nodj$) is not an ancestor of $\nodj$ (or $\nodi$),
nor the left or right sibling of its ancestor.

\begin{lemma} \label{lem:independent} If $\trans$ is symmetrical,
one can build a $1$-blind bimachine computing:
\begin{equation*}
f_I: (\apar{A})^+ \rightarrow \Nat, \forest \mapsto
\left\{
    \begin{array}{l}
        \displaystyle \sum_{(\nodi,\nodj) \in I(\forest)} \pro{\nodi,\nodj} \text{ if } \forest \text{ factorization of height at most } \mh; \\
        0 \text{ otherwise.}\\
    \end{array}
\right.
\end{equation*}
\end{lemma}

\begin{proof}[Proof idea.]
We define $\type{\nodi}{\forest}$ for $\nodi \in \itera{\forest}$
as a bounded abstraction of $\nodi$ which describes
the frontier and the location of $\nodi$ in $\forest$
and in $\basis{\nodi}{\forest}$.
Using symmetry, we show that for $(\nodi, \nodj) \in I(\forest)$,
$\pro{\nodi, \nodj}$ only depends
on $\type{\nodi}{\forest}$ and $\type{\nodj}{\forest}$,
but not on their relative positions.
Hence we build a $1$-blind bimachine,
whose main bimachine ranges over all possible $\nodj$
and computes $\type{\nodj}{\forest}$,
and whose submachines range over all possible $\nodi$
(a special treatment has to be done to avoid $\nodi$ such that
$(\nodi,\nodj) \in L(\forest)$),
compute $\type{\nodi}{\forest}$ and output $\pro{\nodi, \nodj}$.
The submachines do not need to "see" $\nodj$.
\end{proof}

We finally show $\ref{it:trt} \Rightarrow \ref{it:det}$
in Theorem \ref{theo:characterization}.
Given $w \in A^+$ we first compute $\forest$ of height at most $\mh$
by Proposition \ref{prop:effective-factorizations}.
Then we use the machines from Lemmas
\ref{lem:distinguishable}, \ref{lem:linked} and \ref{lem:independent} 
and build a $1$-blind transducer computing the sum of their outputs.
The original function can be recovered since
$1$-blind transducers are closed under composition
with two-way \cite{nguyen2020comparison}.

\section{Conclusion and outlook}

As a conclusion, we discuss future work.
This paper introduces new proof techniques,
in particular the use of factorization forests
to study the productions of transducers.
We believe that these techniques
give a step towards other membership problems
concerning pebble transducers.
Among them, let us mention the membership problem from $k$-pebble to $k$-blind,
at first over unary alphabets.
Similarly, the membership  from $k$-pebble to $k$-marble
over non-unary alphabets is worth being studied
(the answer seems to rely on combinatorial properties of the output,
since unary outputs can always be produced using marbles).

%%
%% Bibliography
%%

\newpage

\bibliographystyle{plain}% the recommended bibstyle
\bibliography{up}

\newpage

\appendix

\section{Proof of Lemma \ref{lem:bimoversst}}

We show that given a bimachine with external pebble functions,
which are computed by $\SSTs$,
one can build an equivalent $\SST$.

\subsection{$\SST$ with lookaround}

We first define a variant of $\SST$ with the same
expressive power. Intuitively, this model
is similar to bimachines, in the sense
that the register update not only depends 
on the current letter, but also on a finite
abstraction of the prefix and suffix.

\begin{definition}An \emph{$\SST$ with lookaround} $\trans = (A,\regs,  M, \mu,I, \lambda, F)$ is:
\item
\begin{itemize}

\item an input alphabet $A$ and a finite set $\regs$ of registers;

\item a morphism into a finite monoid $\mu: A^* \rightarrow M$;

\item an initial row vector $I \in \Nat^{\regs}$;

\item a register update function $\lambda: M \times A \times M\rightarrow \Nat^{\regs \times \regs}$;

\item an output column vector $F \in \Nat^\regs$.
\end{itemize}
\end{definition}

Let us define its semantics. Intuitively, in position $i$ of $w \in A^+$,
we perform the register update $\lambda(\mu(w[1{:}i{-1}]),w[i],]\mu(w[i{+1}{:}|w|]))$.
Formally, for $0 \le i \le |w|$, we define $\trans^{w,i} \in \Nat^{\regs}$ ("the values of the registers after reading $w[1{:}i]$"\footnote{Due to the fact that $\lambda$ looks "on the right", $\trans^{w,i}$ depends on the whole $w$ and not only on $w[1{:}i]$.}) as follows:
\begin{itemize}

\item $\trans^{w,0} \defined I$;

\item for $i \ge 1$, $\trans^{w,i} \defined \trans^{w,i{-}1}\times \lambda(\mu(w[1{:}i{-1}]),w[i],]\mu(w[i{+1}{:}|w|]))$.

\end{itemize}

To define the function $f: A^+ \rightarrow \Nat$ computed by $\trans$,
we combine the final values by the output vector:
\begin{equation*}
f(w) \defined \trans^{w,|w|} F.
\end{equation*}

It is known that $\SST$ with lookaround are equivalent to $\SST$
(the proof is roughly a "determinisation" procedure for eliminating the rightmost argument of $\lambda$,
and an encoding of the monoid in the registers for eliminating the leftmost argument).

\begin{lemma}[\cite{filiot2017copyful}] \label{lem:stateless}
Given an $\SST$ with lookaround, we can build an equivalent $\SST$.
\end{lemma}

Hence, it is sufficient to build an $\SST$ with lookaround.

\subsection{Main proof of Lemma \ref{lem:bimoversst}}

Let $\trans = (A, M, \mu, \lambda , \oras)$ be the bimachine with external pebble functions.
Each $\exte \in \oras$ is computed by an
$\SST$ $\trans_\exte \defined (A \uplus \marq{A}, \regs_\exte, I_\exte, T_\exte,F_\exte)$.

\begin{example}[Running example]
Let $\trans = (A, M, \mu, \lambda , \{\exte\})$ with $M = \{\neutral\}$ singleton
and $\lambda(\neutral,a,\neutral) = \exte$ for $a \in A$.
Let $\trans_{\exte}$ have two registers $x,y$ with $x$ initialized to $1$
and $y$ initialized to $0$. When reading $a \in A \uplus \marq{A} $ it performs $x \leftarrow x, y \leftarrow y+x$. Finally it outputs $y$.

Then $\exte(w) = |w|$ and $\trans$ computes $f: w \mapsto |w|^2$.
\end{example}

\begin{definition} Let $w \in A^*$ and $\exte \in \oras$.
We define $\calls(w,\exte)$ as the set of positions of $w$
in which $\exte$ is called, that is $\{ 1 \le j \le |w|:
\lambda(\mu(w[1{:}j{-}1]),w[j],\mu(w[j{+}1{:}|w|]) = \exte\}$.
\end{definition}

For $1 \le j \le i \le |w|$, $I_{\exte} T_{\exte}(\omar{w[1{:}i]}{j})$
corresponds to the value of the registers of $\trans_\exte$
in position $i$ when the call to $\exte$ is performed from position $j$.

\begin{claim} \label{claim:somsom}For all $w \in A^*$, the following holds:
\begin{equation*}
\begin{aligned}
f(w) = \sum_{\exte \in \oras} \sum_{\substack{1 \le j \le |w| \\ j \in \calls(w,\exte)}} I_{\exte} T_{\exte}(\omar{w}{j})F_\exte
\end{aligned}
\end{equation*}
\end{claim}

\begin{proof} By definition of external pebble functions we have:
\begin{equation*}
\begin{aligned}
f(w) \defined \sum_{1 \le j \le |w|} \exte_j(\omar{w}{j})
\end{aligned}
\end{equation*}

where $\exte_j \defined \lambda(\mu(w[1{:}j{-}1]),w[j],\mu(w[j{+}1{:}|w|]))$ is "the external function called in $j$".
Hence by partitioning the sum depending on the external functions it follows:
\begin{equation*}
\begin{aligned}
f(w) \defined \sum_{\exte \in \oras} \sum_{\substack{1 \le j \le |w| \\ j \in \calls(w,\exte)}} \exte(\omar{w}{j})
\end{aligned}
\end{equation*}

And finally we note that
 $\exte(\omar{w}{j})
 = I_{\exte} T_{\exte}(\omar{w}{j})F_{\exte}.$
\end{proof}

\subparagraph*{Idea of the construction.}
Let us fix an external function $\exte$.
Following Claim \ref{claim:somsom}, we want to
build an $\SST$ with lookaround $\utrans$ which computes the values of the vector:
\begin{equation*}
\begin{aligned}
\sum_{\substack{1 \le j \le |w| \\ j \in \calls(w,\exte)}} I_{\exte} T_{\exte}(\omar{w}{j}) \in \Nat^\regs.
\end{aligned}
\end{equation*}

For this, it  will keep track when in position $i$ of the values of:
\begin{equation*}
\begin{aligned}
\sum_{\substack{1 \le j \le i \\ j \in \calls(w,\exte)}} I_{\exte} T_{\exte}(\omar{w[1{:}i]}{j}) \in \Nat^{\regs}
\end{aligned}
\end{equation*}

and these values will be updated when going from $i$ to $i+1$.

\begin{example}[Running example]

We have $I_{\exte} T_{\exte}(\omar{w[1{:}i]}{j})(x) =1$ and  $I_{\exte} T_{\exte}(\omar{w[1{:}i]}{j})(y) = i$.

Hence $\displaystyle \sum_{\substack{1 \le j \le i \\ j \in \calls(w,\exte)}}I_{\exte} T_{\exte}(\omar{w[1{:}i]}{j})(x) = i$ and $\displaystyle \sum_{\substack{1 \le j \le i \\ j \in \calls(w,\exte)} }I_{\exte} T_{\exte}(\omar{w[1{:}i]}{j})(y) = i^2$.

\end{example}

\subparagraph*{Formal construction.}
Let $\utrans = (A,\regs,  M, \mu,I, \kappa, F)$ be an $\SST$
with lookaround with:
\begin{itemize}
\item the set $\regs \defined  \biguplus_{\exte \in \oras} \{\Som_x:x \in \regs_{\exte}\} \uplus \{\Old_x:x \in \regs_{\exte}\}$ of registers;

\item the morphism $\mu: A^* \rightarrow M$ used in $\trans$;

\item an initial column vector $I \in \Nat^{\regs}$ such that for all $\exte \in \oras$ and $x \in \regs_{\exte} $:
\begin{itemize}
\item $I(\Som_x) = 0$;
\item $I(\Old_x) = I_{\exte}(x)$;
\end{itemize}

\item the update $\kappa: M \times A \times M\rightarrow \Nat^{\regs \times \regs}$ as follows.
Let $(m,a,n) \in M \times A \times N$ and $\exte \defined \lambda(m,a,n)$.
Then $\kappa(m,a,n)$ performs the following updates:
\begin{itemize}

\item for all $\exteg \neq \exte$ and $x \in \regs_{\exteg}$:
\begin{itemize}
\item $\Som_x \leftarrow  \sum_{y \in \regs_\exteg} \alpha_y \Som_y$;
\item $\Old_x \leftarrow  \sum_{y \in \regs_\exteg} \alpha_y \Old_y$;
\end{itemize}

where $x \leftarrow  \sum_{y \in \regs_\exteg} \alpha_y y$ is the
update performed by $\trans_{\exteg}$ when reading $a$;

\item for all $x \in \regs_\exte$:
\begin{itemize}
\item $\Som_x \leftarrow \sum_{y \in \regs_\exte} \alpha_y \Som_y + \sum_{y \in \regs_\exte} \beta_y \Old_y$;
\item $\Old_x \leftarrow  \sum_{y \in \regs_\exte} \alpha_y \Old_y$;
\end{itemize}

where $x \leftarrow  \sum_{y \in \regs_\exte} \alpha_y y$ is the
update performed by $\trans_{\exte}$ when reading $a$;

and $x \leftarrow  \sum_{y \in \regs_\exte} \beta_y y$ is the
update performed by $\trans_{\exte}$ when reading $\marq{a}$.

Intuitively, the sum with the $\beta_y$
corresponds to what is "added" by the new call to $\exte$.
\end{itemize}

\item the output line vector $F \in \Nat^\regs$ such that for all $\exte \in \oras$ and $x \in \regs_{\exte} $:
\begin{itemize}
\item $F(\Som_x) = F_{\exte}(x)$;
\item $F(\Old_x) = 0$.
\end{itemize}

\end{itemize}

\begin{example}[Running example] $\utrans$ performs the following updates:
\begin{itemize}
\item $\Old_x \leftarrow \Old_x$, $\Old_y \leftarrow \Old_y + \Old_x$;
\item $\Som_x \leftarrow \Som_x + \Old_x$, $\Som_y \leftarrow \Som_y + \Som_x + \Old_y + \Old_x $.
\end{itemize}
We can check that  $\utrans^{w,i}(\Old_x) = 1$, $\utrans^{w,i}(\Old_y) = i$
and $\utrans^{w,i}(\Som_x) = i$, $\utrans^{w,i}(\Som_y) = i^2$.
\end{example}

\subparagraph*{Correctness of the construction.}

As the registers $\Old_x$ for $x \in \regs_{\exte}$ are updated following the updates
of $\trans_{\exte}$, it follows immediately that:

\begin{claim} \label{claim:old} Given $x \in \regs_{\exte}$, for all $w \in A^+$ and $1 \le i \le |w|$ we have:
\begin{equation*}
\begin{aligned}
\utrans^{w,i}(\Old_x) = I_{\exte} T_{\exte}(w[1{:}i])(x).
\end{aligned}
\end{equation*}
\end{claim}

We can finally show that the registers $\Som_x$ store
the information we wanted.

\begin{claim} Given $x \in \regs_{\exte}$, for all $w \in A^+$ and $0 \le i \le |w|$ we have:
\begin{equation*}
\begin{aligned}
\utrans^{w,i}(\Som_x) = \sum_{\substack{1 \le j \le i \\ j \in \calls(w,\exte)}} I_{\exte} T_{\exte}(\omar{w[1{:}i]}{j})(x).
\end{aligned}
\end{equation*}
\end{claim}

\begin{proof} We proceed by induction on $0 \le i \le |w|$.
For $i=0$ both terms equal $0$. For the induction step with $i \ge 1$ 
let $\exte \in \oras$ and $x \in \regs_{\exte}$.

Suppose that $\lambda(\mu(w[1{:}i{-}1]),w[i],\mu(w[i{+}1{:}|w|])) = \exte$
(the case when they differ is similar and even easier), then:
\begin{equation}
\label{eq:firre}
\begin{aligned}
\sum_{\substack{1 \le j \le i \\ j \in \calls(w,\exte)}} I_{\exte} T_{\exte}(\omar{w[1{:}i]}{j})(x)
= I_{\exte} T_{\exte}(\omar{w[1{:}i]}{i})(x) +\sum_{\substack{1 \le j \le i{-1} \\ j \in \calls(w,\exte)}} I_{\exte} T_{\exte}(\omar{w[1{:}i]}{j})(x).
\end{aligned}
\end{equation}

\begin{itemize}
\item Let $x \leftarrow  \sum_{y \in \regs_\exte} \alpha_y y$
be the update performed by $\trans_{\exte}$ when reading $a$.

Then for $j\le i-1$, $ I_{\exte} T_{\exte}(\omar{w[1{:}i]}{j})(x) = \sum_{y \in \regs_{\exte}} \alpha_y \times I_{\exte} T_{\exte}(\omar{w[1{:}i{-}1]}{j})(y)$.
\item Let $x \leftarrow  \sum_{y \in \regs_\exte} \beta_y y$
be the update performed by $\trans_{\exte}$ when reading $\marq{a}$.

Then $I_{\exte}  T_{\exte}(\omar{w[1{:}i]}{i})(x) = \sum_{y \in \regs_{\exte}} \beta_y  \times I_{\exte} T_{\exte}(w[1{:}i{-}1])(y)$.
\end{itemize}

Hence we can rewrite Equation \ref{eq:firre} using
values in position $i{-}1$. By Claim \ref{claim:old} and the induction hypothesis,
this sum coincides with the update in $\utrans$.
\end{proof}

The fact that $\utrans$ computes $f$
follows from the definition of $F$
and Claim \ref{claim:somsom}.

\section{Non-computability of $\isq$ by a $k$-blind bimachine}

In order to simplify the description, we do not
deal with $\isq$ but a slight variant:
\begin{equation*}
f: a^{n_1} b \cdots  ba^{n_\ell}b \mapsto \sum_{i=1}^{\ell} n_i(n_i +1) = \isq + \nba.
\end{equation*}
Indeed, by Claim \ref{claim:sss}, if $\isq$ is computable by a $k$-blind
transducer, then so is $f$ (since $\nba$ is computable by a $0$-blind).

\subparagraph*{Description with a $1$-marble bimachine.}
Let us first give a $1$-marble bimachine for $f$.
The monoid is $M = \{\neutral, 0_M\}$
where $\neutral$ is neutral and $0_M$ absorbing (both are idempotents).
The morphism is $\mu: a\mapsto \neutral, b \mapsto 0_M$.
In the main bimachine, we define the output function
$\lambda(m,a,m) \defined \exte_a$ and $\lambda(m,b,m) \defined \exte_b$
for all $m \in M$.

Function $\exte_b$ is constant to $0$, its output
is $\lambda_{\exte_b}(m,a,m) \defined 0$ and
$\lambda_{\exte_b}(m,b,m) \defined 0$
for all $m \in M$.
Function $\exte_a$ outputs twice the number of $a$
after the last $b$. Its output
is $\lambda_{\exte_a}(m,a,\neutral) \defined 2$,
$\lambda_{\exte_a}(m,a,0_M) \defined 0$,
and $\lambda_{\exte_a}(m,b,m) \defined 0$
for all $m \in M$.

\subparagraph*{Non-computability by $1$-blind.}
We set $\uu \defined a, \ud \defined a$ and
$m,n, m_1,n_1,m_2,n_2 = 0_M$.
Let $e_1 {\defined}\mu(\uu) = \neutral$,  $e_2 {\defined} \mu(\ud) = \neutral$,
note that $e {\defined} m_1 e_1 n_1 {=} m_2 e_2 n_2 = 0_M$.
Then:
\begin{itemize}
\item if $p = 0_M$, then
$\pro{mem_1 e_1 \cro{\uu} e_1 p e_2 \cro{\ud} e_2 n_2 en} = \pro{0_M \cro{a} 0_M \cro{a} 0_M} = 0$

and $m_1 e_1 p e_2 n_2 =e = 0_M$, $em_1e_1p e_2= em_2e_2 = 0_M$ and $e_1p e_2 n_2 e= e_1n_1e = 0_M$;

\item if $p = \neutral$, then
$\pro{mem_1 e_1 \cro{\uu} e_1 p e_2 \cro{\ud} e_2 n_2 en}= \pro{0_M \cro{a} \neutral \cro{a} 0_M} = 2$

and $m_2 e_2 p e_1 n_1 = e = 0_M$, $em_2e_2pe_1 = em_1e_1 = 0_M$ and $ e_2pe_1 n_1e = e_2n_2e = 0_M$.

\end{itemize}

It follows from Definition \ref{def:symmetrical} that the bimachine
is not symmetrical, hence by Lemma \ref{lem:blind-symm} the function
$f$ is not computable by a $k$-blind bimachine for any $k \ge 0$.

\section{Proof of Lemma \ref{lem:blind-symm} - blind $\Rightarrow$ symmetrical}

Lemma \ref{lem:blind-symm} is the core of the negative membership result,
since it describes a necessary condition for $f:A^+ \rightarrow \Nat$ given
by the $1$-marble bimachine $\trans$,
to be computable by a $k$-blind bimachine.
Assume that it is the case,
we want to show that $\trans$ is symmetrical.

The elements $m,n \in M$ of Definition \ref{def:symmetrical} define a "context",
which should be reported everywhere in the proof,
but it is never used in the results. Hence to simplify the notations,
we suppose that $m = n$ is the neutral element of $M$
(it will no longer be mentioned).

\subsection{The production does not depend on $p$}

This subsection is devoted to showing Lemma \ref{lem:revert1}.
Intuitively, it says that the production on $\uu$ called from $\ud$ only
depends their respective "contexts" (which are $e_1, m_1,n_1$ and $e_2, m_2,n_2$)
but not on the $p$ which separates them.

\begin{lemma}  \label{lem:revert1}
Let $m_1,n_1,m_2,n_2 \in M$ and $\uu , \ud \in A^+$
such that $e_1 {\defined}\mu(\uu)$,  $e_2 {\defined} \mu(\ud)$
and $e {\defined} m_1 e_1 n_1 {=} m_2 e_2 n_2$ are idempotents.
Let $p \in M$ be such that $m_1 e_1 p e_2 n_2 {=} e$, $em_1e_1p e_2{=} em_2e_2$ and $e_1p e_2 n_2 e{=} e_1n_1e$. Then:
\begin{equation*}
\pro{em_1 e_1 \cro{\uu} e_1 p e_2 \cro{\ud} e_2 n_2 e} = \pro{em_1 e_1 \cro{\uu} e_1 n_1 e m_2 e_2 \cro{\ud} e_2 n_2 e}.
\end{equation*}
\end{lemma}

We denote $\Phi \defined em_1 e_1 \cro{\uu} e_1 p e_2 \cro{\ud} e_2 n_2 e$ and
$\Phi' \defined em_1 e_1 \cro{\uu} e_1 n_1 e m_2 e_2 \cro{\ud} e_2 n_2 e$
the bitypes considered in Lemma \ref{lem:revert1}.
They are depicted in Figure \ref{fig:bitypes-p1}.

\begin{figure}[h!]
\centering
\begin{subfigure}{1\linewidth}
\centering
\begin{tikzpicture}[scale=0.7]

	% Input 

	\draw (-4.5,5) rectangle (4.5,4.5);
	\node[above] at (0,4.45) {$e$};
	\draw (-6.5,5) rectangle (-4.5,4.5);
	\node[above] at (-5.5,4.45) {$e$};
	\draw (6.5,5) rectangle (4.5,4.5);
	\node[above] at (5.5,4.45) {$e$};

	\draw (-4.5,4.5) rectangle (-3.5,4);
	\node[above] at (-4,3.9) {$m_1$};	
	\draw (-3.5,4.5) rectangle (-2.5,4);
	\node[above] at (-3,3.9) {$e_1$};	
		\draw[fill=red!80]  (-2.5,3.5) rectangle (-1.5,4);
		\node[above] at (-2,3.4) {$u_1$};
	\draw (-2.5,4.5) rectangle (-1.5,4);
	\node[above] at (-2,3.9) {$e_1$};
	\draw (-1.5,4.5) rectangle (-0.5,4);
	\node[above] at (-1,3.9) {$e_1$};	
	
	\draw (-0.5,4.5) rectangle (0.5,4);
	\node[above] at (0,3.9) {$p$};	
	
	\draw (4.5,4.5) rectangle (3.5,4);	
	\node[above] at (4,3.9) {$n_2$};	
	\draw (3.5,4.5) rectangle (2.5,4);
	\node[above] at (3,3.9) {$e_2$};
		\draw[fill=blue!70]  (2.5,3.5) rectangle (1.5,4);
		\node[above] at (2,3.4) {$u_2$};	
	\draw (2.5,4.5) rectangle (1.5,4);
	\node[above] at (2,3.9) {$e_2$};	
	\draw (1.5,4.5) rectangle (0.5,4);
	\node[above] at (1,3.9) {$e_2$};
	
	\draw[thick,dashed](-1.5,3.3) --(-1.5,3) -- (6.5,3) -- (6.5,4.3) ;
	\node[above] at (3.5,2.4) {$e_1 n_1 e$};
	
	\draw[thick,dashed](1.5,2.8) --(1.5,2.5) -- (-6.5,2.5) -- (-6.5,4.3);
	\node[above] at (-3.5,2.4) {$e m_2 e_2$};

\end{tikzpicture}
\subcaption{The bitype $\Phi = em_1 e_1 \cro{\uu} e_1 p e_2 \cro{\ud} e_2 n_2 e$}
\end{subfigure}

\vspace{1\baselineskip}

\begin{subfigure}{1\linewidth}
\centering
\begin{tikzpicture}[scale=0.7]

	% Input 

	\draw (-6.5,5) rectangle (-4.5,4.5);
	\node[above] at (-5.5,4.45) {$e$};		
	\draw (-4.5,5) rectangle (0.5,4.5);
	\node[above] at (-2,4.45) {$e$};
	\draw (2.5,5) rectangle (0.5,4.5);
	\node[above] at (1.5,4.45) {$e$};
	\draw (2.5,5) rectangle (7.5,4.5);
	\node[above] at (5,4.45) {$e$};
	\draw (7.5,5) rectangle (9.5,4.5);
	\node[above] at (8.5,4.45) {$e$};

	\draw (-4.5,4.5) rectangle (-3.5,4);
	\node[above] at (-4,3.9) {$m_1$};	
	\draw (-3.5,4.5) rectangle (-2.5,4);
	\node[above] at (-3,3.9) {$e_1$};	
		\draw[fill=red!80]  (-2.5,3.5) rectangle (-1.5,4);
		\node[above] at (-2,3.4) {$u_1$};
	\draw (-2.5,4.5) rectangle (-1.5,4);
	\node[above] at (-2,3.9) {$e_1$};
	\draw (-1.5,4.5) rectangle (-0.5,4);
	\node[above] at (-1,3.9) {$e_1$};	
	
	\draw (-0.5,4.5) rectangle (0.5,4);
	\node[above] at (0,3.9) {$n_1$};	
	
	\draw (7.5,4.5) rectangle (6.5,4);
	\node[above] at (7,3.9) {$n_2$};
	\draw (3.5,4.5) rectangle (2.5,4);	
	\node[above] at (3,3.9) {$m_2$};	
	\draw (6.5,4.5) rectangle (5.5,4);
	\node[above] at (5,3.9) {$e_2$};
		\draw[fill=blue!70]  (5.5,3.5) rectangle (4.5,4);
		\node[above] at (5,3.4) {$u_2$};	
	\draw (5.5,4.5) rectangle (4.5,4);
	\node[above] at (4,3.9) {$e_2$};	
	\draw (4.5,4.5) rectangle (3.5,4);
	\node[above] at (6,3.9) {$e_2$};

\end{tikzpicture}
\subcaption{The bitype $\Phi' = em_1 e_1 \cro{\uu} e_1 n_1 e m_2 e_2 \cro{\ud} e_2 n_2 e$}
\end{subfigure}

\caption{\label{fig:bitypes-p1}  The bitypes considered in this subsection}
\end{figure}
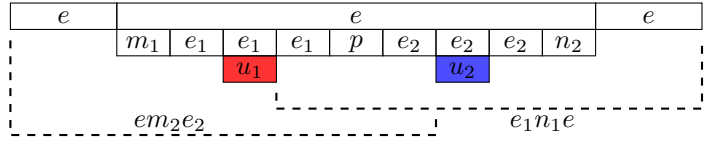
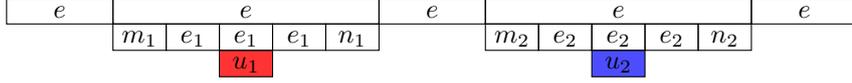

Let us prove Lemma \ref{lem:revert1}.
Since $\mu$ is surjective, there exists $\alpha, \beta, \gamma \in A^*$ such that
$\mu(\alpha) = em_1e_1$,
$\mu(\beta) = e_1pe_2$,
$\mu(\gamma) = e_2n_2e$.
They describe an instance of the bitype $\Phi$.

We are going to apply Lemma \ref{lem:forbid} below with $\alpha, \beta, \gamma, \uu$ and $\ud$.
Intuitively, this lemma creates a family of words in which both the bitypes $\Phi$ and $\Phi'$ occurs,
and describes what is the output of a $k$-blind bimachine on this family.

\begin{lemma} \label{lem:forbid} Let
 $\alpha, \beta, \gamma \in A^*$,
 $\uu, \ud \in A^+$, and $f$ computed by a $k$-blind bimachine.
 There exists $\omega \ge1$ such that the following holds.
 Let $w \defined \alpha \uu^\omega \beta \ud^\omega \gamma$
 and $G: \Nat^4 \rightarrow \Nat$ defined by:
   \begin{equation*}
   G(X,Y,X',Y') \defined  f(w^{2\omega-1} (\alpha \uu^{\omega X} \beta \ud^{\omega Y} \gamma)  w^{\omega-1} (\alpha \uu^{\omega X'} \beta \ud^{\omega Y'}  \gamma) w^\omega)
   \end{equation*}
 
Then there exists a polynomial $P: \Nat^2 \rightarrow \Nat$ such that
for all $X,Y,X',Y' \ge 2$:
  \begin{equation*}
  G(X,Y,X',Y') = P(X+X', Y+Y').
  \end{equation*}
\end{lemma}

\begin{remark} $P$ has degree at most $k$.
\end{remark}

\begin{remark} Here we shall fix $X' = 2$, but the lemma will be re-used lated with variable $X'$.
\end{remark}

Let $\omega$ be fixed by Lemma \ref{lem:forbid}. It allows to describe the output of
a $k$-blind bimachine on words of the form:
  \begin{equation*}
  w^{2\omega-1} (\alpha \uu^{\omega X} \beta \ud^{\omega Y} \gamma)  w^{\omega-1} (\alpha u_1^{2\omega} \beta \ud^{\omega Y'}  \gamma) w^\omega
  \end{equation*}

We now want to understand what is the output of our $1$-marble bimachine $\trans$,
and relate this output with $\pro{\Phi}$ and $\pro{\Phi'}$.
This is the purpose of Lemma \ref{lem:prod-iter} below, we
set $w_0 \defined w^{2\omega-1}\alpha$, $v_1 \defined \uu$, $w_1 = \beta$, $v_2 \defined \ud$, $w_2 \defined \gamma w^{\omega-1}\alpha u_1^{2\omega}\beta$, $v_3 \defined \ud$ and $w_3 \defined \gamma w^{\omega}$.

\begin{lemma} \label{lem:prod-iter} Let $w_0, w_1, w_2, w_3 \in A^*$ and $v_1, v_2, v_3 \in A^+$.
Let $f_1 \defined \mu(v_1)$, $f_2 \defined \mu(v_2)$ and $f_3 \defined \mu(v_3)$.
We suppose that $f_1, f_2, f_3$ are idempotents.
Let $H: \Nat^3 \rightarrow \Nat$ defined by:
  \begin{equation*}
  H(Z_1,Z_2,Z_3) = f(w_0 v_1^{Z_1} w_1 v_2^{Z_2} w_2 v_3^{Z_3} w_3).
  \end{equation*}

Define the following bitypes:
\item
\begin{itemize}
\item $\Psi_1^2 \defined \mu(w_0)f_1 \cro{v_1} f_1  \mu(w_1) f_2 \cro{v_2} f_2  \mu(w_2) f_3  \mu(w_3)$;
\item $\Psi_1^3 \defined   \mu(w_0)f_1 \cro{v_1} f_1  \mu(w_1) f_2  \mu(w_2) f_3 \cro{v_3}  f_3  \mu(w_3)$;
\item $\Psi_2^3 \defined  \mu(w_0)f_1  \mu(w_1) f_2 \cro{v_2} f_2    \mu(w_2) f_3 \cro{v_3}  f_3  \mu(w_3)$.
\end{itemize}

There exists a polynomial $Q(Z_1,Z_2,Z_3)$ without any term in $Z_1Z_2$, $Z_1Z_3 $ nor $Z_2Z_3$ such that for all $Z_1,Z_2, Z_3 \ge 4$:
  \begin{equation*}
  H(Z_1,Z_2,Z_3) = Q(Z_1,Z_2,Z_3) + \pro{\Psi_1^2}Z_1Z_2 + \pro{\Psi_1^3}Z_1Z_3 + \pro{\Psi_2^3} Z_2Z_3
  \end{equation*}
\end{lemma}

\begin{remark} Here we shall not use $\Psi_2^3$, but the lemma will be re-used lated.
\end{remark}

Let $T_1, T_2, T'_2 \ge 4$,
we have:
   \begin{equation*}
   w_0 v_1^{\omega T_1} w_1 v_2^{\omega T_2} w_2 v_3^{\omega T'_2} w_3 = w^{2\omega-1} (\alpha \uu^{\omega T_1} \beta \ud^{\omega T_2} \gamma)  w^{\omega-1} (\alpha u_1^{2\omega} \beta \ud^{\omega T'_2}  \gamma) w^\omega.
   \end{equation*}

Thus by applying $f$ to each member of the equality:
  \begin{equation*}
  H(\omega T_1, \omega T_2, \omega T'_2) = G(T_1, T_2, 2, T'_2).
  \end{equation*}
  
Hence by using the two previous lemmas:
  \begin{equation*}
   P(T_1 +  2, T_2 + T'_2)= Q(\omega T_1, \omega T_2, \omega T'_2) + \omega^2\pro{\Psi_1^2}T_1T_2 +  \omega^2 \pro{\Psi_1^3}T_1T'_2 +  \omega^2 \pro{\Psi_2^3} T_2T'_2. 
  \end{equation*}

Since $Q(\omega T_1, \omega T_2, \omega T'_2)$ has no term in $T_1 T_2$ nor $T_1 T'_2$,
and $P$ depends on $T_2 +T'_2$, it follows by identification that $\pro{\Psi_1^2} = \pro{\Psi_1^3}$.
Finally $\pro{\Phi} = \pro{\Phi'}$ since:

\begin{claim}
$\Psi_1^2 = \Phi$ and $\Psi_1^3 = \Phi'$.
\end{claim}

\begin{proof}
With the elements chosen to apply Lemma \ref{lem:prod-iter}:
\begin{itemize}
\item $f_1 = \mu(u_1) = e_1$, $f_2 = e_2$, $f_3 = e_2$;
\item $\mu(w) = \mu(\alpha u_1^\omega \beta u_2^\omega \gamma) = e$;
\item $\mu(w_0) = e^{2 \omega - 1}  e m_1 e_1 = e m_1 e_1$;
\item $\mu(w_1) = e_1 p e_2 $;
\item $\mu(w_2) = e_2 n_2 (e m_1 e_1 p e_2) = e_2 n_2 e m_2 e_2$;
\item $\mu(w_3) = e_2 n_2 e$
\item $\mu(w_0) f_1 = e m_1 e_1$;
\item $f_1 \mu(w_1) f_2 = e_1 pe_2$;
\item $f_2 \mu(w_2) f_3 \mu(w_3) = e_2 n_2 (e m_2 e_2 n_2 e) = e_2 n_2 e$;
\item $f_1 \mu(w_1) f_2  \mu(w_2) f_3= (e_1 p e_2  n_2 e) m_2 e_2 = e_1 n_1 e m_2 e_2$;
\item $f_3 \mu(w_3) = e_2 n_2 e$. 
\end{itemize}
\item
And the $5$ last equalities give the equalities between the bitypes.
\end{proof}

\subsection{The production does not change when swapping $\uu$ and $\ud$}

This subsection is devoted to Lemma \ref{lem:revert2}.
Intuitively, it says that the production from $\ud$ to $\uu$ is the same
as that from $\uu$ to $\ud$, if we keep the same "contexts"
(which are independent from $p$).

\begin{lemma}  \label{lem:revert2}
Let $m_1,n_1,m_2,n_2 \in M$ and $\uu, \ud \in A^+$
such that $e_1 {\defined}\mu(\uu)$,  $e_2 {\defined} \mu(\ud)$
and $e {\defined} m_1 e_1 n_1 {=} m_2 e_2 n_2$ are idempotents.
Let $p \in M$ such that $m_1 e_1 p e_2 n_2 {=} e$, $em_1e_1p e_2{=} em_2e_2$
and $e_1p e_2 n_2 e{=} e_1n_1e$. Then:
\begin{equation*}
\pro{em_1 e_1 \cro{\uu} e_1 p e_2 \cro{\ud} e_2 n_2 e} = \pro{ e m_2 e_2 \cro{\ud} e_2 n_2e m_1 e_1 \cro{\uu} e_1 n_1 e}.
\end{equation*}
\end{lemma}

We denote $\Phi \defined em_1 e_1 \cro{\uu} e_1 p e_2 \cro{\ud} e_2 n_2 e$ and
$\Phi' \defined  e m_2 e_2 \cro{\ud} e_2 n_2e m_1 e_1 \cro{\uu} e_1 n_1 e$
the bitypes considered in Lemma \ref{lem:revert2}.
They are depicted in Figure \ref{fig:bitypes-p}.

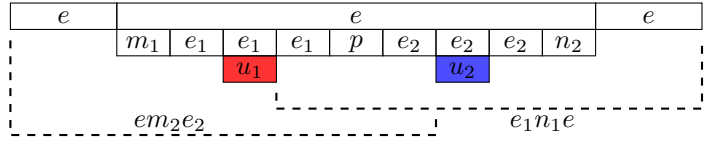
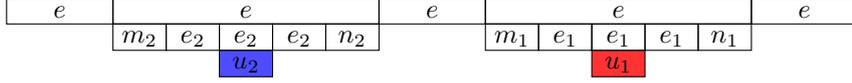
\begin{figure}[h!]
\centering
\begin{subfigure}{1\linewidth}
\centering
\begin{tikzpicture}[scale=0.7]

	% Input 

	\draw (-4.5,5) rectangle (4.5,4.5);
	\node[above] at (0,4.45) {$e$};
	\draw (-6.5,5) rectangle (-4.5,4.5);
	\node[above] at (-5.5,4.45) {$e$};
	\draw (6.5,5) rectangle (4.5,4.5);
	\node[above] at (5.5,4.45) {$e$};

	\draw (-4.5,4.5) rectangle (-3.5,4);
	\node[above] at (-4,3.9) {$m_1$};	
	\draw (-3.5,4.5) rectangle (-2.5,4);
	\node[above] at (-3,3.9) {$e_1$};	
		\draw[fill=red!80]  (-2.5,3.5) rectangle (-1.5,4);
		\node[above] at (-2,3.4) {$u_1$};
	\draw (-2.5,4.5) rectangle (-1.5,4);
	\node[above] at (-2,3.9) {$e_1$};
	\draw (-1.5,4.5) rectangle (-0.5,4);
	\node[above] at (-1,3.9) {$e_1$};	
	
	\draw (-0.5,4.5) rectangle (0.5,4);
	\node[above] at (0,3.9) {$p$};	
	
	\draw (4.5,4.5) rectangle (3.5,4);	
	\node[above] at (4,3.9) {$n_2$};	
	\draw (3.5,4.5) rectangle (2.5,4);
	\node[above] at (3,3.9) {$e_2$};
		\draw[fill=blue!70]  (2.5,3.5) rectangle (1.5,4);
		\node[above] at (2,3.4) {$u_2$};	
	\draw (2.5,4.5) rectangle (1.5,4);
	\node[above] at (2,3.9) {$e_2$};	
	\draw (1.5,4.5) rectangle (0.5,4);
	\node[above] at (1,3.9) {$e_2$};
	
	\draw[thick,dashed](-1.5,3.3) --(-1.5,3) -- (6.5,3) -- (6.5,4.3) ;
	\node[above] at (3.5,2.4) {$e_1 n_1 e$};
	
	\draw[thick,dashed](1.5,2.8) --(1.5,2.5) -- (-6.5,2.5) -- (-6.5,4.3);
	\node[above] at (-3.5,2.4) {$e m_2 e_2$};

\end{tikzpicture}
\subcaption{The bitype $\Phi = em_1 e_1 \cro{\uu} e_1 p e_2 \cro{\ud} e_2 n_2 e$}
\end{subfigure}

\vspace{1\baselineskip}

\begin{subfigure}{1\linewidth}
\centering
\begin{tikzpicture}[scale=0.7]

	% Input 

	\draw (-6.5,5) rectangle (-4.5,4.5);
	\node[above] at (-5.5,4.45) {$e$};		
	\draw (-4.5,5) rectangle (0.5,4.5);
	\node[above] at (-2,4.45) {$e$};
	\draw (2.5,5) rectangle (0.5,4.5);
	\node[above] at (1.5,4.45) {$e$};
	\draw (2.5,5) rectangle (7.5,4.5);
	\node[above] at (5,4.45) {$e$};
	\draw (7.5,5) rectangle (9.5,4.5);
	\node[above] at (8.5,4.45) {$e$};

	\draw (-4.5,4.5) rectangle (-3.5,4);
	\node[above] at (-4,3.9) {$m_2$};	
	\draw (-3.5,4.5) rectangle (-2.5,4);
	\node[above] at (-3,3.9) {$e_2$};	
		\draw[fill=blue!70]  (-2.5,3.5) rectangle (-1.5,4);
		\node[above] at (-2,3.4) {$u_2$};
	\draw (-2.5,4.5) rectangle (-1.5,4);
	\node[above] at (-2,3.9) {$e_2$};
	\draw (-1.5,4.5) rectangle (-0.5,4);
	\node[above] at (-1,3.9) {$e_2$};	
	
	\draw (-0.5,4.5) rectangle (0.5,4);
	\node[above] at (0,3.9) {$n_2$};	
	
	\draw (7.5,4.5) rectangle (6.5,4);
	\node[above] at (7,3.9) {$n_1$};
	\draw (3.5,4.5) rectangle (2.5,4);	
	\node[above] at (3,3.9) {$m_1$};	
	\draw (6.5,4.5) rectangle (5.5,4);
	\node[above] at (5,3.9) {$e_1$};
		\draw[fill=red!80]  (5.5,3.5) rectangle (4.5,4);
		\node[above] at (5,3.4) {$u_1$};	
	\draw (5.5,4.5) rectangle (4.5,4);
	\node[above] at (4,3.9) {$e_1$};	
	\draw (4.5,4.5) rectangle (3.5,4);
	\node[above] at (6,3.9) {$e_1$};

\end{tikzpicture}
\subcaption{The bitype $\Phi' = e m_2 e_2 \cro{\ud} e_2 n_2e m_1 e_1 \cro{\uu} e_1 n_1 e$}
\end{subfigure}

\caption{\label{fig:bitypes-p}  The bitypes considered in this subsection}
\end{figure}

The proof of Lemma \ref{lem:revert2} is essentially similar
to that of Lemma \ref{lem:revert1}, with the subtlety that we have
to swap $\uu$ and $\ud$ at some point.
Since $\mu$ is surjective, there exists $\alpha, \beta, \gamma \in A^*$ such that
$\mu(\alpha) = em_1e_1$,
$\mu(\beta) = e_1pe_2$,
$\mu(\gamma) = e_2n_2e$.
They describe an instance of the bitype $\Phi$.

We first apply Lemma \ref{lem:forbid} with $\alpha, \beta, \gamma, \uu$ and $\ud$.
We shall fix $Y' = 2$, but $X'$ is variable.
Let $\omega$ be fixed by Lemma \ref{lem:forbid} and similarily
$w \defined \alpha \uu^\omega \beta \ud^\omega \gamma$.
It allows to describe the output of
a $k$-blind bimachine on words of the form:
\begin{equation*}
w^{2\omega-1} (\alpha \uu^{\omega X} \beta \ud^{\omega Y} \gamma)  w^{\omega-1} (\alpha \uu^{\omega X'} \beta u_2^{2\omega}  \gamma) w^\omega
\end{equation*}

We now want to understand
what is the output of our $1$-marble bimachine
$\trans$ on such words
and relate this output with $\pro{\Phi}$ and $\pro{\Phi'}$.

For this purpose, we set $w_0 \defined w^{2\omega-1}\alpha$, $v_1 \defined \uu$, $w_1 = \beta$, $v_2 \defined \ud$, $w_2 \defined \gamma w^{\omega-1}\alpha$, $v_3 \defined \uu$ and $w_3 \defined  \beta u_2^{2 \omega} \gamma w^{\omega}$ and apply Lemma \ref{lem:prod-iter}.
We shall use $\Psi_2^3$, but not $\Psi_1^3$.

Let $T_1, T_2, T'_1 \ge 4$,
we have:
   \begin{equation*}
   w_0 v_1^{\omega T_1} w_1 v_2^{\omega T_2} w_2 v_3^{\omega T'_1} w_3 = w^{2\omega-1} (\alpha \uu^{\omega T_1} \beta \ud^{\omega T_2} \gamma)  w^{\omega-1} (\alpha \uu^{\omega T'_1} \beta u_2^{2\omega}  \gamma) w^\omega.
   \end{equation*}

Thus by applying $f$ to each member of the equality:
  \begin{equation*}
  H(\omega T_1, \omega T_2, \omega T'_1) = G(T_1, T_2, T'_1, 2).
  \end{equation*}
Hence by using Lemma \ref{lem:forbid} and Lemma \ref{lem:prod-iter}:
  \begin{equation*}
  P(T_1 +  T'_1, T_2 + 2)= Q(\omega T_1, \omega T_2, \omega T'_1) + \omega^2\pro{\Psi_1^2}T_1T_2 +  \omega^2 \pro{\Psi_1^3}T_1T'_1 +  \omega^2 \pro{\Psi_2^3} T_2T'_1. 
  \end{equation*}

Since $Q(\omega T_1, \omega T_2, \omega T'_1)$ has no term in $T_1 T_2$ nor $T_2 T'_1$
and $P$ depends on $T_1 +T'_1$,  it follows by identification that $\pro{\Psi_1^2} = \pro{\Psi_2^3}$.
Finally $\pro{\Phi} = \pro{\Phi'}$ since:

\begin{claim}
$\Psi_1^2 = \Phi$ and $\Psi_2^3 = \Phi'$.
\end{claim}

\subsection{Conclusion and proof of Lemma \ref{lem:blind-symm}}

Let $m_1,n_1,m_2,n_2 \in M$ and $\uu, \ud \in A^+$
such that $e_1 {\defined}\mu(\uu)$,  $e_2 {\defined} \mu(\ud)$
and $e {\defined} m_1 e_1 n_1 {=} m_2 e_2 n_2$ are idempotents.
\begin{description}

\item[- by applying  Lemma \ref{lem:revert1}:]
if $p \in M$ is such that 

$m_1 e_1 p e_2 n_2 = e$, $em_1e_1p e_2= em_2e_2$ and $e_1p e_2 n_2 e= e_1n_1e$,

then $\pro{em_1 e_1 \cro{\uu} e_1 p e_2 \cro{\ud} e_2 n_2 e}$
$ = \pro{em_1 e_1 \cro{\uu} e_1 n_1 e m_2 e_2 \cro{\ud} e_2 n_2 e}$;

\item[- by applying  Lemma \ref{lem:revert2} (after swapping $1$ and $2$):]
if $p \in M$ is such that 

$m_2 e_2 p e_1 n_1 = e$, $em_2e_2pe_1 = em_1e_1$ and $ e_2pe_1 n_1e = e_2n_2e$,

then $\pro{em_2 e_2 \cro{\ud} e_2 p e_1 \cro{\uu} e_1 n_1 e}$
$ = \pro{em_1 e_1 \cro{\uu} e_1 n_1 e m_2 e_2 \cro{\ud} e_2 n_2 e} $.

\end{description}

Finally, we just set $K \defined  \pro{em_1 e_1 \cro{\uu} e_1 n_1 e m_2 e_2 \cro{\ud} e_2 n_2 e} $.
The production is independent from $p$ and of whether $\uu$ is before or after $\ud$.

\subsection{Technical tools: proofs of Lemma \ref{lem:forbid} and Lemma \ref{lem:prod-iter}}

\subsubsection{ Proof of Lemma \ref{lem:forbid}}

We assume that $f: A^+ \rightarrow \Nat$ can be computed
by a $k$-blind bimachine for some $k \ge 1$.
Let $\nu: A^* \rightarrow N$ be the morphism
of this bimachine (and all its sub-bimachines, without loss of generality).
In this proof we forget about the fact that $f$ was given by a $1$-marble
bimachine with morphism $\mu$: here,
the term "idempotent" will always mean "idempotent of $N$".

Let $\omega$ be the idempotent index of $N$,
that is the smallest integer such that $m^\omega$ is idempotent
for all $m \in N$.
In particular, $f_1 \defined \nu(\uu^\omega)$, $f_2 \defined \nu(\ud^\omega) $
are idempotents.
Let $w \defined \alpha \uu^\omega \beta \ud^\omega \gamma $,
then $f \defined \nu(w^\omega) $ is also an idempotent.
Note that forall $X,Y \ge 1$, $\nu(\alpha \uu^{\omega X} \beta \ud^{\omega Y}  \gamma) = \nu(w)$.

Now for $X,Y, X', Y' \ge 2$ let:
  \begin{equation*}
W(X,Y,X',Y') \defined w^{2\omega-1} (\alpha \uu^{\omega X} \beta \ud^{\omega Y} \gamma)  w^{\omega-1} (\alpha \uu^{\omega X'} \beta \ud^{\omega Y'}  \gamma) w^\omega) \in A^+
  \end{equation*}

We show Lemma \ref{lem:forbid} by induction on $k \ge 0$, starting with
the case $k=0$.

\begin{sublemma} \label{sub:cas} If $\utrans = (A,N,\nu,\lambda)$ is a $0$-blind bimachine
computing a function $g: A^+ \rightarrow \Nat$,
there exists a polynomial $P: \Nat^2 \rightarrow \Nat$ of degree at most $1$ such that:
  \begin{equation*}
  \forall X,Y, X', Y' \ge 2, f(W(X,Y,X',Y') ) = P(X+X', Y+Y').
  \end{equation*}
\end{sublemma}

\begin{proof} Intuitively, we only show that since
$\uu^{X}$ and $\uu^{X'}$ (resp. $\ud^{Y}$ and $\ud^{Y'}$)
are in the same "context", the productions are the same
in both blocks.
Formally, we extend the definition of $\lambda: N \times A \times N \rightarrow \Nat$
to $\lambda: N \times A^* \times N \rightarrow \Nat$ with:
  \begin{equation*}
  \lambda(m,w,n) \defined \sum_{i=1}^{|w|} \lambda(m\nu(w[1{:}i{-}1]),w[i],\nu(w[i{+}1{:}|w|])n).
  \end{equation*}

It follows by partitioning the sum defining $g(w)$ that:
\begin{claim}
If $w = w_1 \cdots w_{\ell}$ with $w_j \in A^*$, then:
\begin{equation*}
g(w) = \sum_{j=1}^\ell \lambda(\nu(w_1 \cdots w_{j-1}),w_j,\nu(w_{j+1}\cdots w_{\ell})).
\end{equation*}
\end{claim}
By slicing the input $W(X,Y,X',Y')$ in several factors, we
thus show that $g(W(X,Y,X',Y'))$ is the sum of several $\lambda(m,w_i,n)$:
\begin{itemize}
\item $\lambda(\nu(\vide), w^{2\omega-1} \alpha \uu^\omega, \nu(\uu^{\omega (X-1)} \beta \ud^{\omega Y} \gamma  w^{\omega-1} \alpha \uu^{\omega X'} \beta \ud^{\omega Y'}  \gamma w^\omega))$

$= \lambda(\nu(\vide), w^{2\omega-1} \alpha \uu^\omega, f_1 \nu(\beta) f_2 \nu(\gamma) f )$
$=C_1$  constant with respect to $X,Y,X',Y'$;

\vspace*{0.5\baselineskip}

\item $\displaystyle \sum_{i=1}^{X-2} \lambda(\nu(w^{2\omega-1} \alpha \uu^{\omega i}),\uu^{\omega}, \nu(\uu^{\omega (X-1-i)} \beta \ud^{\omega Y} \gamma  w^{\omega-1} \alpha \uu^{\omega X'} \beta \ud^{\omega Y'}  \gamma w^\omega))$

$\displaystyle = \sum_{i=1}^{X-2} \lambda(f \nu(w^{\omega-1}\alpha) f_1,\uu^{\omega},f_1 \nu(\beta) f_2 \nu(\gamma) f)$

$= (X-2) \lambda(f \nu(w^{\omega-1}\alpha) f_1,\uu^{\omega},f_1 \nu(\beta) f_2 \nu(\gamma) f) =: (X-2) C_X$;

\vspace*{0.5\baselineskip}

\item $\lambda(\nu(w^{2\omega-1} \alpha \uu^{\omega (X-1)}), \uu^{\omega} \beta \ud^{\omega}, \nu(\ud^{\omega (Y-1)} \gamma  w^{\omega-1} \alpha \uu^{\omega X'} \beta \ud^{\omega Y'}  \gamma w^\omega))$

$=\lambda(f \nu(w^{\omega-1}\alpha) f_1, \uu^{\omega} \beta \ud^{\omega}, f_2 \nu(\gamma) f) = C_2$;

\vspace*{0.5\baselineskip}

\item $\displaystyle \sum_{i=1}^{Y-2} \lambda(\nu(w^{2\omega-1} \alpha \uu^{\omega X} \beta \ud^{\omega i}),\ud^{\omega}, \nu(\ud^{\omega (Y-1-i)} \gamma  w^{\omega-1} \alpha \uu^{\omega X'} \beta \ud^{\omega Y'}  \gamma w^\omega))$

$\displaystyle = \sum_{i=1}^{Y-2} \lambda(f \nu(w^{\omega-1}\alpha) f_1 \nu(\beta) f_2,\ud^{\omega}, f_2 \nu(\gamma)  f)$

$= (Y-2)\lambda(f \nu(w^{\omega-1}\alpha) f_1 \nu(\beta) f_2,\ud^{\omega}, f_2 \nu(\gamma)  f) =: (Y-2) C_Y$;

\vspace*{0.5\baselineskip}

\item $ \lambda(\nu(w^{2\omega-1} \alpha \uu^{\omega X} \beta \ud^{\omega (Y-1)}),\ud^{\omega} \gamma  w^{\omega-1} \alpha \uu^{\omega}, \nu(\uu^{\omega(X'-1)} \beta \ud^{\omega Y'}  \gamma w^\omega))$
$=C_3$;

\vspace*{0.5\baselineskip}

\item $\displaystyle \sum_{i=1}^{X'-2} \lambda(\nu(w^{2\omega-1} \alpha \uu^{\omega X} \beta  \ud^{\omega Y} \gamma  w^{\omega-1} \alpha  \uu^{\omega i}),\uu^{\omega}, \nu(\uu^{\omega (X'-1-i)}\beta \ud^{\omega Y'}  \gamma w^\omega))$

$= \displaystyle \sum_{i=1}^{X'-2} \lambda(f \nu(w^{\omega-1} \alpha)  f_1),\uu^{\omega}, f_1 \nu(\beta) f_2 \nu(\gamma)f) = (X'-2) C_X$;

\vspace*{0.5\baselineskip}

\item $\lambda(\nu(w^{2\omega-1} \alpha \uu^{\omega X} \beta  \ud^{\omega Y} \gamma  w^{\omega-1} \alpha  \uu^{\omega (X'-1)}),\uu^{\omega} \beta \ud^{\omega}, \nu(\ud^{\omega(Y'-1)}  \gamma w^\omega)) = C_4$;

\vspace*{0.5\baselineskip}

\item $\displaystyle \sum_{i=1}^{Y'-2} \lambda(\nu(w^{2\omega-1} \alpha \uu^{\omega X} \beta  \ud^{\omega Y} \gamma  w^{\omega-1} \alpha  \uu^{\omega X'} \beta \ud^{i}), \ud, \nu(\ud^{\omega (Y'-1-i)} \gamma w^\omega))$

$= \displaystyle \sum_{i=1}^{Y'-2} \lambda(f  \nu(w^{\omega-1} \alpha)  f_1 \nu(\beta) f_2), \ud, f_2 \nu(\gamma) f) = (Y'-2) C_Y$;

\vspace*{0.5\baselineskip}

\item $\lambda(\nu(w^{2\omega-1} \alpha \uu^{\omega X} \beta  \ud^{\omega Y} \gamma  w^{\omega-1} \alpha  \uu^{\omega X'} \beta \ud^{\omega(Y'-1)}), \ud^{\omega} \gamma w^\omega, \nu(\vide)) = C_5$.
\end{itemize}

Putting the terms together we get a polynomial of degree at most $1$ in $X+X',Y+Y'$.
\end{proof}

 For the induction step, let us consider a bimachine
 computing $f$
 with external blind functions $\exte_1, \dots, \exte_n$
 computed by $(k{-}1)$-blind bimachines
 (with morphism $\nu$). Then 
 by induction we get polynomials $Q_1, \dots, Q_n: \Nat^2 \rightarrow \Nat$
 of degree at most $k{-}1$
 such that:
  \begin{equation*}
  \forall X,Y, X', Y' \ge 2, \exte_i(W(X,Y,X',Y') ) = Q_i(X+X', Y+Y').
   \end{equation*}

By definition of external blind functions,
there exists $g_1, \dots, g_n$ computable
by a $0$-blind bimachine with monoid morphism $\nu$
(we define $g_i$ as the function which counts
"the number of times $\exte_i$ is called
by the main machine")
such that forall $w \in A^+$:
 \begin{equation*}
 f(w) = \sum_{i=1}^n g_i(w) \exte_i(w).
   \end{equation*}

By Sublemma \ref{sub:cas}, we get
polynomials $Q'_1, \dots, Q'_n: \Nat^2 \rightarrow \Nat$ such that:
 \begin{equation*}
 \forall X,Y, X', Y' \ge 2, g_i(W(X,Y,X',Y') ) = Q'_i(X+X', Y+Y').
  \end{equation*}
  
The polynomial of degree at most $k$ for $f$ is:
 \begin{equation*}
 Q \defined \sum_{i=1}^n Q'_i Q_i.
 \end{equation*}

\subsubsection{Proof of Lemma \ref{lem:prod-iter}}

Let $W(Z_1, Z_2, Z_3) \defined w_0 v_1^{Z_1} w_1 v_2^{Z_2} w_2 v_3^{Z_3} w_3$.
We want to describe $f(W(Z_1, Z_2,Z_3))$ as a polynomial.
Recall that $f$ is computed by the $1$-marble bimachine
$\trans = (A,M, \mu, \oras, \lambda)$. For $\exte \in \oras$,
let $\trans_{\exte} \defined (A,M, \mu, \lambda_{\exte})$
be the auxiliary bimachine which computes it.

\subparagraph*{Decomposing bitypes.} In order to study
the production performed by $\trans$ on $W(Z_1, Z_2, Z_3)$,
we first show how the productions on bitypes behave when
iterating an idempotent.

\begin{claim}\label{claim:Zw} Let $m,m',m'' \in M$ and $v,w \in A^+$ such that $e \defined \mu(v)$
is an idempotent.

Then $\forall Z \ge 2$:
\begin{align*}
\pro{m\cro{v^Z}m'\cro{w}m''} &= \pro{m\cro{v} e m'\cro{w}m''} +  \pro{me\cro{v} m'\cro{w}m''}\\
&+ (Z-2) \pro{m e\cro{v} e m'\cro{w}m''}.
\end{align*}
\end{claim}

\begin{proof}[Proof sketch.] We consider the word $v^Z$ as a concatenation
$v \cdots v$ and look at the production performed on each $v$ from $w$.
We regroup the production depending on the bitypes.
\end{proof}

The following claim is obtained symmetrically.

\begin{claim} \label{claim:wZ} Let $m,m',m'' \in M$ and $v,w \in A^+$ such that $\mu(v)$
is an idempotent.

Then  $\forall Z \ge 2$:
\begin{align*}
\pro{m\cro{w}m'\cro{v^Z}m''} &= \pro{m\cro{w} m' e \cro{v}m''} +  \pro{m\cro{w} m'\cro{w}em''}\\
&+ (Z-2) \pro{m \cro{w} m'e\cro{v}em''}.
\end{align*}
\end{claim}

Finally we can study the case of two idempotents iterated independently.

\begin{claim} \label{claim:ZZ} Let $m,m',m'' \in M$ and $v_1, v_2 \in A^+$ such that $e_1 \defined \mu(v_1)$
and $e_2 \defined \mu(v_2)$ are idempotents. There exists $C_1, C_2, C \in \Nat$ such that $\forall Z_1, Z_2 \ge 2$:
\begin{align*}
\pro{m\cro{v_1^{Z_1}}m'\cro{v_2^{Z_2}}m'} =  &\pro{me_1\cro{v_1}e_1 m'e_2\cro{v_2}e_2m''} \times Z_1 Z_2 \\
&+ C_{1} Z_1 + C_{2} Z_2 + C.
\end{align*}
\end{claim}

\begin{proof}  Apply successively claims \ref{claim:Zw} and \ref{claim:wZ}.
\end{proof}
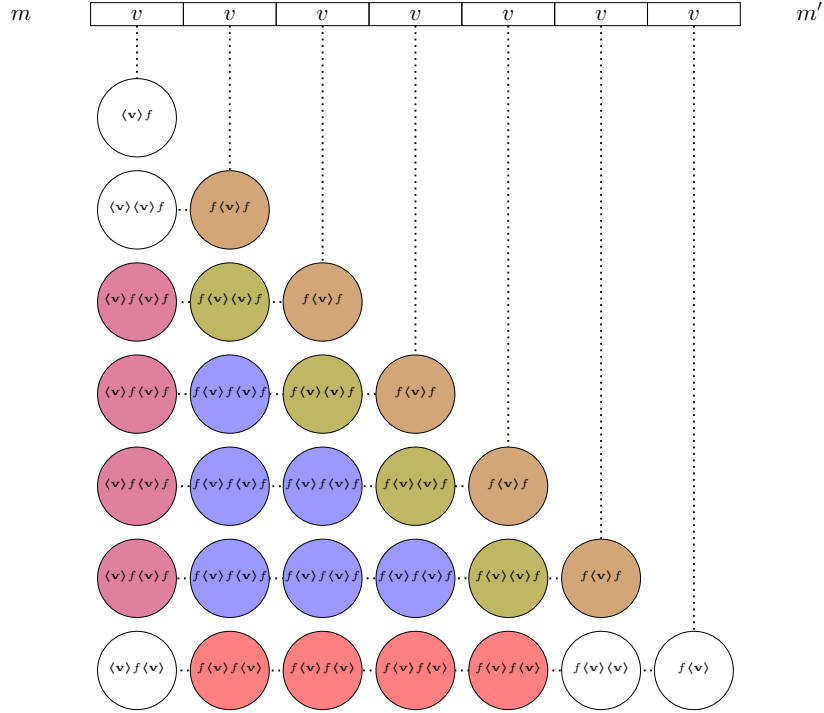
\begin{figure}[]
\centering
\resizebox{11cm}{!}{
\begin{tikzpicture}[scale=0.7]

	\node[above] at (-1.5,8.95) {$m$};	
		
	\draw (0,9) rectangle (2,9.5);
	\node[above] at (1,8.95) {$v$};	
	\draw (2,9) rectangle (4,9.5);
	\node[above] at (3,8.95) {$v$};	
	\draw (4,9) rectangle (6,9.5);
	\node[above] at (5,8.95) {$v$};	
	\draw (6,9) rectangle (8,9.5);
	\node[above] at (7,8.95) {$v$};	
	\draw (8,9) rectangle (10,9.5);
	\node[above] at (9,8.95) {$v$};	
	\draw (10,9) rectangle (12,9.5);
	\node[above] at (11,8.95) {$v$};	
	\draw (12,9) rectangle (14,9.5);
	\node[above] at (13,8.95) {$v$};

	\node[above] at (15.5,8.955) {$m'$};	
	
	\draw[thick, dotted](1,9) --(1,7) -- (0.5,7) ;
	\draw[thick, dotted](3,9) --(3,5) -- (0.5,5) ;
	\draw[thick, dotted](5,9) --(5,3) -- (0.5,3) ;
	\draw[thick, dotted](7,9) --(7,1) -- (0.5,1) ;
	\draw[thick, dotted](9,9) --(9,-1) -- (0.5,-1) ;
	\draw[thick, dotted](11,9) --(11,-3) -- (0.5,-3) ;
	\draw[thick, dotted](13,9) --(13,-5) -- (0.5,-5) ;

	\newcommand{\noeud}[3]%
	{\draw[fill = \col] (#1,#2) circle (0.85);
	\node[above] at (#1,#3) {\tiny$\numero$};}
	
	\def\col{white};	
	
	\def\col{brown!70};
	\def\numero{f \cro{v} f};
		\noeud{3}{5}{4.75}
		\noeud{5}{3}{2.75}
		\noeud{7}{1}{0.75}	
		\noeud{9}{-1}{-1.25}
		\noeud{11}{-3}{-3.25}	
		
	\def\col{white};	
		
	\def\numero{\cro{v} f};
		\noeud{1}{7}{6.75}

	\def\numero{f \cro{v}};
		\noeud{13}{-5}{-5.25}
		
	\def\col{olive!60};	

	\def\numero{f \cro{v} \cro{v} f};

		\noeud{3}{3}{2.75}
		\noeud{5}{1}{0.75}
		\noeud{7}{-1}{-1.25}
		\noeud{9}{-3}{-3.25}

	\def\col{white};	

	\def\numero{f \cro{v} \cro{v}};
		\noeud{11}{-5}{-5.25}		

	\def\numero{\cro{v} \cro{v} f};
		\noeud{1}{5}{4.75}	
		
	\def\col{blue!40};		
		
	\def\numero{f \cro{v} f\cro{v} f};

		\noeud{3}{1}{0.75}
		\noeud{5}{-1}{-1.25}
		\noeud{7}{-3}{-3.25}
		
		\noeud{3}{-1}{-1.25}
		\noeud{5}{-3}{-3.25}
		
		\noeud{3}{-3}{-3.25}
		
	\def\col{white};	
	\def\numero{f \cro{v} f\cro{v} f};

	\def\col{red!50};		
	\def\numero{f \cro{v} f \cro{v}};
		\noeud{9}{-5}{-5.25}
		\noeud{7}{-5}{-5.25}	
		\noeud{5}{-5}{-5.25}	
		\noeud{3}{-5}{-5.25}	
		
	\def\col{white};	
	\def\numero{\cro{v} f\cro{v}};	
		\noeud{1}{-5}{-5.25}		

	\def\col{purple!50};		
	\def\numero{\cro{v} f \cro{v}f};
		\noeud{1}{3}{2.75}
		\noeud{1}{1}{0.75}
		\noeud{1}{-1}{-1.25}
		\noeud{1}{-3}{-3.25}

\end{tikzpicture}}
\caption{\label{fig:carre} A decomposition of production on the monotype $m\cro{v^{7}}m'$}
\end{figure}

\subparagraph*{Monotypes.} In order to obtain a complete
description of the production performed on $W(Z_1, Z_2, Z_3)$,
we also need to describe "the production on $v$ from itself".
This notion is covered by monotypes,
which are defined below similarly to bitypes.

\begin{definition} A \emph{monotype} $\Pi \defined m \cro{v} m'$ consists in $m,m' \in M$ and $v \in A^+$.
\end{definition}

As for bitypes,
we define "the production performed in
$v$ by the calls from $v$", in $\Pi$.
For $1 \le i \le j \le |v|$, we define
$\Pi(i,j) \defined \lambda_{\exte_{j}}(m\mu(v[1{:}{i}{-}1]), v[i], \mu(v[i{+1}{:}j])m') \in \Nat $
where $\exte_{j} \defined \lambda(m\mu(v[1{:}{j}{-}1]), v[j],\mu(v[j{+1}{:}|v|])m')$. Then we set:
\begin{equation*}
\pro{\Pi} \defined \sum_{\substack{1 \le i \le j \le |v|}}
\Pi(i,j) \in \Nat.
\end{equation*}

We now describe the production
on a monotype given by an iterated idempotent.

\begin{lemma} \label{lem:Z} Let $m,m' \in M$ and $v \in A^+$ such that $\mu(v)$
is an idempotent. There exists a polynomial $P: \Nat \rightarrow \Nat$
of degree at most $2$ such that $\forall Z \ge 4$:
\begin{equation*}
\pro{m\cro{v^Z}m'} = P(Z).
\end{equation*}
\end{lemma}

\begin{proof} We claim that for all $Z \ge 4$:

\begin{equation}
\begin{aligned}
\label{eq:carre}
\pro{m\cro{v^Z}m'} &= \pro{m \cro{v} f m'} +  \pro{m \cro{v} \cro{v} f m'} + \pro{m \cro{v} f \cro{v} f m'}  \\
&+\pro{m f \cro{v} \cro{v} m'}  +\pro{m f \cro{v}}\\
&+(Z-3)\pro{m f \cro{v}f \cro{v} m'}  + (Z-3)\pro{m \cro{v}f \cro{v} fm'}\\
&+(Z-3)\pro{m f \cro{v}\cro{v}f m'}  + (Z-2)\pro{m f\cro{v} fm'}\\
&+\dfrac{(Z-4)(Z-3)}{2} \pro{mf \cro{v} f \cro{v}f m'}.
\end{aligned}
\end{equation}

In order to show Equation \ref{eq:carre}, 
we first decompose the production of
$v^Z$ in several blocks, which are
for $1 \le i \le j \le Z$ "the production on the $i$-th $v$
when called from the $j$-th $v$":
\begin{equation}
\begin{aligned}
\label{eq:carre2}
\pro{m\cro{v^Z}m'} &= \sum_{\substack{1\le i< j \le Z}} \pro{m \mu(v^{i-1})\cro{v} \mu(v^{j-i-1}) \cro{v} \mu(v^{Z-j})m'}\\
& +\sum_{1 \le j \le Z} \pro{m \mu(v^{j-1})\cro{v} \mu(v^{Z-j})m'}\\
\end{aligned}
\end{equation}

Finally, we can obtain Equation \ref{eq:carre}
by regrouping the terms that occur in the right member
of Equation \ref{eq:carre2},
depending on their bitype.
In order to convince the
reader without burdening the proof, 
we decomposed $\pro{m\cro{v^7}m'}$ in Figure \ref{fig:carre},
using a matrix which describes the terms of Equation \ref{eq:carre2}:
in cell $(i,j)$ for  $i \le j$, we have written "the bitype
of the $i$-th $v$ when called from the $j$-th $v$".
The bitypes which are equal are identified with the same colors.
It is not hard to see that we get the coefficients
of Equation \ref{eq:carre}.
\end{proof}

\subparagraph*{Proof of Lemma \ref{lem:prod-iter}.} 

We have $\mu(v_1) = f_1$, $\mu(v_2) = f_2$, $\mu(v_3) = f_3$.
Define $m_0 \defined \mu(w_0)$, $m_1 \defined \mu(w_1)$,
$m_2 \defined \mu(w_2)$, $m_3 \defined \mu(w_3)$.
By batching the productions $\pro{i,j}$
performed when computing $f(W(Z_1, Z_2, Z_3))$, we obtain Equation \ref{eq:fz}
 for all $Z_1, Z_2, Z_3 \ge 4$.
Applying the claims \ref{claim:Zw}, \ref{claim:wZ}, \ref{claim:ZZ} and  Lemma \ref{lem:Z} to the various
terms of this equation yields Lemma \ref{lem:prod-iter}.
The framed terms give
the monomials in $Z_1Z_2$, $Z_1Z_3$ and $Z_2Z_3$.

\begin{equation}
\label{eq:fz}
\begin{aligned}
f(W(Z_1, Z_2, Z_3)) &= \pro{\cro{w_0} f_1 m_1 f_2 m_2 f_3 m_3} + \pro{\cro{w_0} \cro{v_1^{Z_1}} m_1 f_2 m_2 f_3 m_3} \\
&+ \pro{\cro{w_0} f_1 \cro{w_1}  f_2 m_2 f_3 m_3} + \pro{\cro{w_0}f_1 m_1 \cro{v_2^{Z_2}} m_2 f_3 m_3}\\
&+ \pro{\cro{w_0} f_1 m_1 f_2 \cro{w_2} f_3 m_3} + \pro{\cro{w_0}f_1 m_1 f_2 m_2 \cro{v_3^{Z_3}}m_3}\\
&+ \pro{\cro{w_0} f_1 m_1 f_2 m_2 f_3 \cro{w_3}}\\
%
%\end{aligned}
%\end{equation}
%
%\begin{equation*}
%\begin{aligned}
&+\pro{m_0 \cro{v_1^{Z_1}} m_1 f_2 m_2 f_3 m_3} + \pro{m_0 \cro{v_1^{Z_1}}  \cro{w_1}  f_2 m_2 f_3 m_3}\\
&+ \text{\framebox{$\pro{m_0 \cro{v_1^{Z_1}}  m_1 \cro{v_2^{Z_2}} m_2 f_3 m_3}$}}+ \pro{m_0 \cro{v_1^{Z_1}} m_1 f_2 \cro{w_2} f_3 m_3}\\
&+ \text{\framebox{$\pro{m_0 \cro{v_1^{Z_1}}  m_1 f_2 m_2 \cro{v_3^{Z_3}}m_3}$}}+ \pro{m_0 \cro{v_1^{Z_1}} m_1 f_2 m_2 f_3 \cro{w_3}}\\
&+\pro{m_0 f_1 \cro{w_1}  f_2 m_2 f_3 m_3} + \pro{m_0 f_1 \cro{w_1} \cro{v_2^{Z_2}} m_2 f_3 m_3}\\
&+\pro{m_0 f_1 \cro{w_1}  f_2 \cro{w_2} f_3 m_3}+ \pro{m_0 f_1 \cro{w_1}  f_2 m_2 \cro{v_3^{Z_3}}m_3}\\
&+ \pro{m_0 f_1 \cro{w_1}  f_2 m_2 f_3 \cro{w_3}}\\
\\
&+\pro{m_0 f_1 m_1 \cro{v_2^{Z_2}}  m_2 f_3 m_3} +\pro{m_0 f_1 m_1 \cro{v_2^{Z_2}}  \cro{w_2} f_3 m_3}\\
&+\text{\framebox{$\pro{m_0 f_1 m_1 \cro{v_2^{Z_2}} m_2 \cro{v_3^{Z_3}}m_3}$}} + \pro{m_0 f_1 m_1 \cro{v_2^{Z_2}} m_2 f_3 \cro{w_3}}\\
\\
&+\pro{m_0 f_1 m_1 f_2 \cro{w_2} f_3 m_3}+ \pro{m_0 f_1 m_1 f_2 \cro{w_2} \cro{v_3^{Z_3}}m_3} \\
&+ \pro{m_0 f_1 m_1 f_2 \cro{w_2}  f_3 \cro{w_3}}\\
\\
&+\pro{m_0 f_1 m_1 f_2 m_2 \cro{v_3^{Z_3}}m_3} + \pro{m_0 f_1 m_1 f_2 m_2 \cro{v_3^{Z_3}} \cro{w_3}}\\
\\
&+ \pro{m_0 f_1 m_1 f_2 m_2  f_3 \cro{w_3}}\\
\end{aligned}
\end{equation}

\section{Proof of Lemma \ref{lem:partition}}

We show that $\dep{\forest} \uplus \{\dep{\nodi}: \nodi \in \itera{\forest}\}$
partitions $\nodes{\forest}$. Applying this result specifically to the leaves
will directly give a partition of the positions of $w$
by the frontiers of these nodes.
The proof is done by induction on the structure of $\forest$:
\begin{itemize}
\item if $\forest = (a)$ is a leaf the result is true;
\item if $\forest = (\forest_1) \dots (\forest_n)$ then
$\dep{\forest} = \{\forest\} \uplus \dep{\forest_1} \uplus \dep{\forest_n}$.
\begin{equation*}
\itera{\forest} =\{\forest_i: 2 \le i \le n{-}1\}  \biguplus_{1 \le i \le n} \itera{\forest_i} \\
\end{equation*}
Furthermore, by induction hypothesis
$\dep{\forest_i} \uplus \{\dep{\nodi}: \nodi \in \itera{\forest_i}\}$ partitions $\nodes{\forest_i}$
for all $1 \le i \le n$, hence:
\begin{equation*}
\begin{aligned}
\nodes{\forest} &= \{\forest\} \biguplus_{1 \le i \le n} \nodes{\forest_i}.\\
& = \{\forest\}  \biguplus_{1 \le i \le n}\left( \dep{\forest_i} \uplus \{\dep{\nodi}: \nodi \in \itera{\forest_i}\}\right)\\
& = \{\forest\} \uplus \dep{\forest_1} \uplus  \dep{\forest_n} \\
&\biguplus_{2 \le i \le n{-}1} \dep{\forest_i} \biguplus_{1 \le i \le n}\{\dep{\nodi}: \nodi \in \itera{\forest_i}\}\\
& = \dep{\forest} \uplus \{\dep{\nodi}: \nodi \in \itera{\forest}\}\\
\end{aligned}
\end{equation*}
\end{itemize}

\section{Proof of Lemma \ref{lem:nodes-prod}}

Let $w \in A^+$, $\forest \in \facto{\mu}{w}$, we show that:
\begin{equation*}
\begin{aligned}
f(w) =  \sum_{\nodi ,\nodj \in \parti{\forest}} \pro{\nodi,\nodj}.
\end{aligned}
\end{equation*}

We first note that by definition:
\begin{equation*}
\begin{aligned}
 \sum_{\nodi ,\nodj \in \parti{\forest}} \pro{\nodi,\nodj}
= \sum_{\substack{\nodi \in \parti{\forest} \\ i \in \fr{\nodi}{\forest}}}
\sum_{\substack{\nodj \in \parti{\forest} \\ j \in \fr{\nodj}{\forest} \\ i \le j}} \pro{i,j}
\end{aligned}
\end{equation*}

But since by Lemma \ref{lem:partition} $\{\fr{\nodi}{\forest}: \nodi \in \parti{\forest}\}$ partitions the positions of $w$:
\begin{equation*}
\begin{aligned}
 \sum_{\nodi ,\nodj \in \parti{\forest}} \pro{\nodi,\nodj} = \sum_{1 \le i \le |w|}
\sum_{i \le  j \le |w|} \pro{i,j}
\end{aligned}
\end{equation*}

Finally by definition of a $1$-blind bimachine it follows:
\begin{equation*}
\begin{aligned}
 \sum_{1 \le i \le j \le |w|} \pro{i,j} = f(w).
\end{aligned}
\end{equation*}

\section{Proof of Lemma \ref{lem:dep-root}}

We show that each iterable node is
the middle child of a unique basis.
The result is obtained by induction on the structure of the factorization:
\begin{itemize}
\item if $\forest = (a)$ is a leaf the result is true by emptiness;
\item if $\forest = (\forest_1) \dots (\forest_n)$ then
$\dep{\forest} = \{\forest\} \uplus \dep{\forest_1} \uplus \dep{\forest_n}$.
\item[]
\begin{itemize}
\item by induction hypothesis, the iterable nodes of $\forest_1$ (resp. $\forest_n$)
are the descendant of a middle child of a unique
basis of $\forest_1$ (resp. $\forest_n$).
These nodes are still bases of $\forest$.
Hence the iterable nodes of $\forest_1$ and $\forest_n$
are the descendant of a unique basis of $\forest$.
\item the iterable nodes from $\forest_i$ for $2 \le i \le n{-}1$
are the descendant of a middle child of a unique
basis of $\forest$, which is $\forest$ itself (if $n \ge 2$ it is an idempotent node).
\end{itemize}
\end{itemize}

\section{Further properties of factorizations}

In this section, we develop several tools which are useful in our
factorization techniques.

\subsection{Representation of factorizations}

The set $(\apar{A})^+$ not only contains factorizations of height at most $\mh$,
it also contains factorizations of bigger height, and other words which are not
factorizations. However, such inputs will often be considered as "invalid" in
our constructions, thus we need to isolate them.

\begin{definition} Let $\ltm$ be the set of factorizations of height
at most $\mh$.
\end{definition}

$\ltm$ can also be seen as a language of $(\apar{A})^*$.
Using this point of view, the following lemma follows from an easy
 inductive construction.

\begin{lemma}[Folklore] $\ltm$ is a regular language of $(\apar{A})^*$.
\end{lemma}

As a consequence, a function defined over $\ltm$ can
be extended to $(\apar{A})^+$, while preserving its computability
by some bimachine. Thus in the proofs of lemmas
\ref{lem:distinguishable}, \ref{lem:linked} and \ref{lem:independent} 
we shall always consider that the inputs belong to $\ltm$.

\begin{lemma} \label{lem:restriction} Let $k \ge 0$, $g: \ltm \rightarrow \Nat$ computed
by a $k$-pebble (resp. $k$-blind, $k$-marble) bimachine, and:
\begin{equation*}
h: (\apar{A})^+ \rightarrow \Nat, \forest \mapsto
\left\{
    \begin{array}{l}
        g(\forest) \text{ if } \forest \in \ltm; \\
        0 \text{ otherwise.}\\
    \end{array}
\right.
\end{equation*}
Then one can build a $k$-pebble (resp. $k$-blind, $k$-marble) bimachine
which computes $h$.
\end{lemma}

\begin{proof}[Proof idea.] Using its monoid morphism, the main bimachine
can check a regular property of their input and have a different behavior
depending on it.
\end{proof}

\subsection{Frontiers and dependencies}

\label{subs:nota}

\subparagraph*{Size of the sets.}
As evoked in the main body of this paper, over factorizations of height at most $\mh$,
dependencies and frontiers have a bounded size.

\begin{lemma} \label{lem:sizef} Let $\forest \in \ltm$ and $\nodi \in \nodes{\forest}$.
Then $|\dep{\nodi}|, |\fr{\nodi}{\forest}| \le 2^{\mh}$.
\end{lemma}

\begin{proof} We show by induction that if $\nodi$ has
height at most $h$, then $|\dep{\nodi}| \le 2^h$. Thus 
here $|\dep{\nodi}| \le 2^{\mh}$.
Considering the leaves gives the result for the frontier.
\end{proof}

\subparagraph*{Minimum and maximum.}
Intuitively, taking the dependency of a node $\nodi$
consists in pruning the subtree $\nodi$, in a way that
shall "preserve" some information on it.
First, it preserves the "borders" of the words.

\begin{definition} Let $w \in A^+$, $\forest \in \facto{\mu}{w}$
and $\nodi \in \nodes{\forest}$. By construction,
$\nodi$ is a factorization of some portion of $w$
between two positions denoted $\mine{\nodi}{\forest}$ and $\maxe{\nodi}{\forest}$.
\end{definition}

In other words, $\nodi \in \facto{\mu}{w[\mine{\nodi}{\forest}{:}\maxe{\nodi}{\forest}]}$.

\begin{lemma} \label{lem:preserv}
Let $w \in A^+$, $\forest \in \facto{\mu}{w}$
and $\nodi \in \nodes{\forest}$. Let
 $\{i_1 < \cdots < i_\ell\} \defined \fr{\nodi}{\forest}$. Then
 $\mine{\nodi}{\forest} = i_1$
and $\maxe{\nodi}{\forest} = i_{\ell}$.
\end{lemma}

\begin{proof}[Proof idea.]
By induction since $\dep{\nodi}$ contains its rightmost and leftmost children.
\end{proof}

\begin{remark} Thus $\mine{\forest}{\forest} = 1$ and $\maxe{\forest}{\forest} = |w|$
since $\forest$ is a factorization of the whole $w$.
\end{remark}

\subparagraph*{Values.}
More interestingly, the frontier preserves the
image by $\mu$. This explains why considering the frontier
of $\nodi$ is a sufficient abstraction of $\nodi$ when
looking at monoids elements.

\begin{lemma} Let $w \in A^+$, $\forest \in \facto{\mu}{w}$
and $\nodi \in \nodes{\forest}$, then:
\begin{equation*}
\mu(w[\mine{\nodi}{\forest}{:}\maxe{\nodi}{\forest}]) = \mu(w[\fr{\nodi}{\forest}]).
\end{equation*}
We denote this element $\valu{\nodi}$.
\end{lemma}

\begin{proof}[Proof idea.]
By induction since $\dep{\nodi}$ only removes "middle children" of $\nodi$.
\end{proof}

\section{Proof of Lemma \ref{lem:distinguishable} - pairs separated by the frontier of the root}

We build a $1$-blind bimachine computing:
\begin{equation*}
f_D: (\apar{A})^+ \rightarrow \Nat, \forest \mapsto
\left\{
    \begin{array}{l}
        \displaystyle \sum_{(\nodi,\nodj) \in D(\forest)} \pro{\nodi,\nodj} \text{ if } \forest \text{ factorization of height at most } \mh; \\
        0 \text{ otherwise.}\\
    \end{array}
\right.
\end{equation*}

Using Lemma \ref{lem:restriction}, we
restrict our construction to
factorizations of height at most $\mh$.
Let $\forest \in \facto{\mu}{w}$ of height at most $\mh$.
Let $\{p_1 < \dots <p_{\ell(\forest)}\} \defined \fr{\forest}{\forest}$
be the frontier of its root node $\forest$.
By lemmas \ref{lem:sizef} and \ref{lem:preserv},
$\ell(\forest) \le 2^{\mh}$
and $p_1 = 1$ and $p_{\ell(\forest)} = |w|$.

Intuitively, the $p_k$ split the word $w$ in at most $\ell(\forest){-}1$
blocks (some of them can be empty),
and if $(\nodi, \nodj) \in D(\forest)$, then $\fr{\nodi}{\forest}$
and $\fr{\nodj}{\forest}$ are in two distinct blocks.

\begin{claim} \label{claim:fd} The following holds:
\begin{equation*}
f_D(\forest) = \sum_{\substack{1\le k' < k< \ell(\forest) \\ p_{k'} < i < p_{k'+1} \\ \substack{p_{k} < j < p_{k+1}}}}\pro{i,j}.
\end{equation*}
\end{claim}

\begin{proof} We first show that $\forall \nodi, \nodj \in \itera{\forest}$,
$\basis{\nodi}{\forest} = \basis{\nodj}{\forest}$ if and only if there exists $k$ such that
$p_k<\fr{\nodi}{\forest}<p_{k+1}$ and $p_k<\fr{\nodj}{\forest}<p_{k+1}$.
Second, using Lemma \ref{lem:partition} get for all
$1\le i \le |w|$ that $i \not \in \fr{\forest}{\forest}$ if and only if
$i \in \fr{\nodi}{\forest}$ for some $\nodi \in \itera{\forest}$.
The result follows by rewriting the sum which defines $f_D$.
\end{proof}

We now use Claim \ref{claim:fd}  to compute $f_D$ using a $1$-blind bimachine
in Algorithm \ref{algo:distinguish}. Intuitively, the
algorithm ranges over all possible $j$ and for each of them,
it calls an external function which ranges over all $i$ in blocks "on the left", and 
produces the $\pro{i,j}$.

We use the notation of \emph{bitypes}, see Definition \ref{def:bitype}
on page \pageref{def:bitype}.
Note that by definition
$\pro{i,j} = \pro{\mu(w[1{:}i{-}1]) \cro{w[i]} \mu(w[i{+}1{:}j{-}1])  \cro{w[j]} \mu(w[j{+}1{:}|w|])}.$

\begin{algorithm}[h!]
\SetKw{KwVar}{Variables:}
\SetKwProg{Fn}{Function}{}{}
\SetKw{In}{in}
\SetKw{Out}{Output}

 \Fn{$\operatorname{Main}(\forest)$}{

 		$w \leftarrow \operatorname{\text{word factored by}} \forest$;

		\For{$j$ \In $\{1, \dots, |w|\}$}{
		
			\If{$p_{k} <j< p_{k{+}1}\operatorname{\text{\normalfont for some }} 1 \le k < \ell(\forest)$}{
						
			$a \leftarrow w[j]$;
			\tcc{Letter in position $j$.}
			$m \leftarrow \mu(w[p_{k}{+}1{:}j{-}1])$;	\tcc{Left "context" of $j$ in block $k$.}
			$n \leftarrow \mu(w[j{+}1{:}p_{k{+}1}{-}1])$; 	\tcc{Right "context" of $j$ in block $k$.}

			\Out{$\exte^{m\cro{a}n}_k (\forest)$}
			
			}
			
		}}

 \Fn{$\exte^{m\cro{a}n}_k (\forest)$}{
 
 		$w \leftarrow \operatorname{\text{word factored by}} \forest$;
					
		\For{$i$ \In $\{1, \dots, |w|\}$}{
		
			\If{$p_{k'} <i < p_{k'{+}1}\operatorname{\text{\normalfont for some }} 0 \le k' < k $}{
						
			$a' \leftarrow w[i]$;
			\tcc{Letter in position $i$.}
			$m' \leftarrow \mu(w[p_{k'}{+}1{:}i{-}1])$;	\tcc{Left "context" of $i$ in block $k'$.}
			$n' \leftarrow \mu(w[i{+}1{:}p_{k'{+}1}{-}1])$; 	\tcc{Right "context" of $i$ in block $k'$.}
			$m_1 \leftarrow \mu(w[1{:}p_{k'}])$; 	\tcc{Context before block $k'$.}
			$m_2 \leftarrow \mu(w[p_{k'+1} {:}p_{k}])$; 	\tcc{Context between blocks $k'$ and $k$.}
			$m_3 \leftarrow \mu(w[p_{k+1} {:}|w|])$; 	\tcc{Context between blocks $k'$ and $k$.}

			\Out{$\pro{m_1 (m'\cro{a'}n') m_2 (m \cro{a} n) m_3}$}
			
			}
			
		}

}
	
 \caption{\label{algo:distinguish} Computing $f_D$ with a $1$-blind bimachine}
\end{algorithm}

\subparagraph*{Correctness of the algorithm.}
We first claim that if $\forest$ is a factorization of $w$ and $k < \ell(\forest)$, then
$\exte^{m\cro{a}n}_k (\forest)$ computes the sum of productions
over all $i$ which are in a block $k' < k$,
when called from a $j$ in block $k$ and in 
"context" $m,n$. Formally:
\begin{equation*}
\begin{aligned}
\exte^{m\cro{a}n}_k (\forest) = \sum_{\substack{1\le k' < k\\ p_{k'} < i < p_{k'+1}}}\pro{\mu(w[1{:}i{-}1]) \cro{w[i]} \mu(w[i{+}1{:}p_{k}]) m\cro{a}n \mu(w[p_{k+1}{:}|w|]) }.
\end{aligned}
\end{equation*}
Finally with Claim \ref{claim:fd}, it is easy to show that $\operatorname{Main}(\forest) = f_D(\forest)$.

\subparagraph*{Implementation by a $1$-blind bimachine.}
We justify how Algorithm \ref{algo:distinguish} can be implemented with a $1$-blind bimachine.
It uses finitely many external functions $\exte^{m\cro{a}n}_k$ for $1 \le k \le 2^{\mh}-1$, $m,n \in M$ and $a \in A$ and uses its morphism to "detect" 
the frontier of the root (recall that the frontier has a bounded size).
This detection is detailed below.

\begin{definition} \label{def:pi}
Given $w \in A^+$ and $\forest \in \facto{\mu}{w}$,
let $\pi: \{1, \dots, |w|\} \rightarrow \{1, \dots, |\forest|\}$ be the function mapping a position
of $w$ to the corresponding leaf of $\forest$, seen as a position.
\end{definition}

\begin{remark} The image of $\pi$ is exactly the set of positions of $\forest$ labelled by $A$.

\end{remark}

We explain how to "detect" if
the current position belongs to the frontier of the root.

\begin{claim} One can build a regular language $P \subseteq (\apar{A}\uplus \marq{\apar{A}})^*$
such that for all $w \in A^+$, $\forest \in \facto{\mu}{w}$ of height at most $\mh$ and $1 \le j \le |w|$:
\begin{equation*}
\omar{\forest}{\pi(j)} \in P \text{ if and only if } j \in \fr{\forest}{\forest}.
\end{equation*}
\end{claim}

\begin{proof} We build by induction a language $P^h$
which works for factorizations of height at most $h$.
For $h=1$, $P^1 \defined \{\underline{a}:a \in A\}$.
Assume that $h \ge 2$ and $P^{h-1}$  is built. Then we define
$P^{h} \defined P^{h-1} \cup \{(\forest_1) \dots (\forest_n):n \ge 0 \text{ and } (\forest_1 \in P^{h} \text{ or } \forest_n \in P^{h}) \}.$
\end{proof}

In a similar way, it is possible to detect the
number of the current block, and describe the "context" of
the current position in this block.

\begin{claim} \label{claim:regF}
One can build a regular language $P_k^{m,n} \subseteq (\apar{A} \uplus \marq{\apar{A}})^*$
such that for all $w \in A^+$, $\forest \in \ltm$ over $w$ and $1 \le j \le |w|$:
\begin{align*} 
&\omar{\forest}{\pi(j)} \in P_k^{m,n} \text{ if and only if }\\
&p_{k}<j<p_{k+1} \text{ and } \mu(w[p_{k}{+}1{:}j{-}1]) = m  \text{ and } \mu(w[j{+}1{:}p_{k+1}{-}1]) = n.
\end{align*}
\end{claim}

Using Claim \ref{claim:regF}, we can implement the function
Main using a bimachine with external blind functions. Its morphism
is built from the regular languages $P_k^{m,n}$
for $1 \le k \le 2^{\mh}$ and $m,n \in M$, and it calls 
the corresponding $\exte^{m \cro{a}n}_k$ as external function.

The construction of a bimachine the $\exte^{m \cro{a}n}_k$ is similar.

\section{Proof of Lemma \ref{lem:linked} - linked pairs}

We want to build a $0$-blind bimachine computing:
\begin{equation*}
f_L: (\apar{A})^+ \rightarrow \Nat, \forest \mapsto
\left\{
    \begin{array}{l}
        \displaystyle \sum_{(\nodi,\nodj) \in L(\forest)} \pro{\nodi,\nodj} \text{ if } \forest \text{ factorization of height at most $\mh$}; \\
        0 \text{ otherwise.}\\
    \end{array}
\right.
\end{equation*}

We shall build a $1$-pebble bimachine
computing $f_L$. Why not a $0$-blind?
Because from a $1$-pebble bimachine
we can build an equivalent $1$-marble bimachine 
(Corollary \ref{cor:pebmar}).
But as claimed below, $f_L$ has a "linear growth",
hence the results of \cite{doueneau20}
(which "minimize" the number of marbles used) will automatically
build a $0$-marble (= $0$-blind) bimachine\footnote{A direct
construction of a $0$-blind bimachine is possible,
but somehow more complex to explain.}.

\begin{claim} $f_{L}(\forest) = \mc{O}(|\forest|)$.
\end{claim}
\begin{proof} Let $w \in A^+$, $\forest \in \facto{\mu}{w}$ of height at
most $\mh$ (by definition of $f_L$, it is enough to only consider 
such inputs). Given $\nodi \in \nodes{\forest}$, we denote by $\up{\nodi}$ the
set of nodes $\nodj \in \parti{\forest}$ which are either an ancestor of $\nodi$,
or the right/left sibling of an ancestor
of $\nodi$. Since the height is at most $\mh$, it follows that
$|\up{\nodi}| \le 3\mh$. Furthermore:
\begin{equation*}
\begin{aligned}
f_L(\forest) = \sum_{\substack{\nodi  \in \parti{\forest}}} \left( \pro{\nodi, \nodi} + \sum_{\nodj \in \up{\nodi}\smallsetminus \{\nodi\}} \pro{\nodi,\nodj} + \pro{\nodj,\nodi}\right).
\end{aligned}
\end{equation*}

Since the frontiers have size at most $2^{\mh}$ (Lemma \ref{lem:sizef}),
there exists $B \ge 0$ independent from $\forest$  and $w$ such that $\pro{\nodi,\nodj} \le B$
for all $\nodi, \nodj \in \nodes{\forest}$.

Finally $f_L(\forest) \le |\forest| \times (1 + 2 \times 3\mh) \times B.$
\end{proof}

It remains to describe how $f_L$ can be computed
by a $1$-pebble bimachine. We restrict our construction
to factorizations of height at most $\mh$ (using Lemma \ref{lem:restriction}).
The $1$-pebble bimachine is described by Algorithm \ref{algo:linked}.
Intuitively, it ranges
over all possible $i$ and $j$ positions of $w$,
and for each pair, it checks whether $i \in \fr{\nodi}{\forest}$
and $j \in \fr{\nodj}{\forest}$ for some $(\nodi, \nodj) \in L(\forest)$.
If it is the case, it outputs $\pro{i,j}$.
Note that we use an external \emph{pebble} function and not a blind one:
it can "see" the calling position $i$ on the input $\forest$.

\begin{algorithm}[h!]
\SetKw{KwVar}{Variables:}
\SetKwProg{Fn}{Function}{}{}
\SetKw{In}{in}
\SetKw{Out}{Output}

 \Fn{$\operatorname{Main}(\forest)$}{

 		$w \leftarrow \operatorname{\text{word factored by}} \forest$;

		\For{$i$ \In $\{1, \dots, |w|\}$}{

			\Out{$\exte(\forest,i)$}
			
			}
			
		}

 \Fn{$\operatorname{\exte}(\forest,i)$}{

 		$w \leftarrow \operatorname{\text{word factored by}} \forest$;

		\For{$j$ \In $\{1, \dots, |w|\}$}{
			
			$\nodi \leftarrow \text{ unique node in } \parti{\forest} \text{ such that } i \in \fr{\nodi}{\forest}$;
			
			$\nodj \leftarrow \text{ unique node in } \parti{\forest} \text{ such that }  j \in \fr{\nodj}{\forest}$;		
			
			\If{$(\nodi, \nodj) \in L(\forest)$}{
				\Out{$\pro{i,j}$}
			
			}
			
		}}
	
 \caption{\label{algo:linked} Computing $f_L$ with a $1$-pebble bimachine}
\end{algorithm}

Correctness of Algorithm \ref{algo:linked} follows from
the definitions. To see how it can be implemented by a $1$-pebble
bimachine, we only need to show that the condition $(\nodi, \nodj) \in L(\forest)$
can be checked using a regular language, as claimed below.
We re-use the notation $\pi$ from the proof of Lemma \ref{lem:distinguishable}.
The claim follows from similar constructions of regular languages.

\begin{claim} 
One can build a regular language $P \subseteq (\apar{A}\uplus \marq{\apar{A}})^*$
such that for all $w \in A^+$, $1 \le i\le j \le |w|$ and factorization $\forest$ of height at most $\mh$ over $w$:
\begin{align*}
&\omar{\omar{\forest}{\pi(i)}}{\pi(j)} \in P \text{ if and only if}\\
 &\text{when }\nodi, \nodj  \in \parti{\forest} \text{ are such that } i \in \fr{\nodi}{\forest}, j \in \fr{\nodj}{\forest}, \text{we have } (\nodi, \nodj) \in L(\forest).
\end{align*}
\end{claim}

\section{Proof of Lemma \ref{lem:independent} - Independent pairs}

Assume that $\trans$ is symmetrical. We build a $1$-blind bimachine computing:
\begin{equation*}
f_I: (\apar{A})^+ \rightarrow \Nat, \forest \mapsto
\left\{
    \begin{array}{l}
        \displaystyle \sum_{(\nodi,\nodj) \in I(\forest)} \pro{\nodi,\nodj} \text{ if } \forest \text{ factorization of height at most $\mh$}; \\
        0 \text{ otherwise.}\\
    \end{array}
\right.
\end{equation*}

Once more, using Lemma \ref{lem:restriction},
we only consider "valid" inputs.

As explained in the  body of this paper, we shall see
that since $\trans$ is symmetrical,
$\pro{\nodi, \nodj}$ does not depend
on the relative position of $\nodi$ and $\nodj$,
when $(\nodi, \nodj) \in I(\forest)$.
It remains to formalize this intuition by defining the 
\emph{type}. We use the notations of Subsection \ref{subs:nota}.

\subparagraph*{Left and right.} In this paragraph, we
define the "left" and "right" of a node,
as the monoid elements which
are "before" and "after" it.

\begin{definition} Let $w \in A^+$, $\forest \in \facto{\mu}{w}$, and $\nodi \in \nodes{\forest}$.
We define the left and right of $\nodi$ in $\forest$ as follows:
\begin{equation*}
\left\{
    \begin{array}{l}
        \lefe{\nodi}{\forest} \defined \mu(w[1{:}\mine{\nodi}{\forest}{-}1]\\
        \rige{\nodi}{\forest} \defined \mu(w[\maxe{\nodi}{\forest}{+}1{:}|w|])
    \end{array}
\right.
\end{equation*}
\end{definition}

\begin{remark}
$\lefe{\forest}{\forest} = \rige{\forest}{\forest} = \mu(\vide)$.
\end{remark}

\begin{example}
In the factorization $\forest$ of Figure \ref{fig:cont},
we consider the basis $\nodi \in \nodes{\forest}$ colored in blue.
Then $\lefe{\nodi}{\forest} = \mu(w[1{:}2]) = \mu(aa)$ and
$\rige{\nodi}{\forest} = \mu(w[12{:}11]) = \mu(\vide)$.
\end{example}

\begin{figure}[h!]

\centering
\begin{tikzpicture}{scale=1}

	\newcommand{\couleur}{blue}
	\newcommand{\texte}{\small \bfseries \sffamily \mathversion{bold} }

	% Input 1
	\fill[fill=blue]  (4,9)  circle (0.15);
		
	\draw (1.5,9.5) -- (-1,9);
	\draw (1.5,9.5) -- (4,9);
	\draw (1.75,9) -- (6.25,9);
	\draw (6.25,8.5) -- (6.25,9);
	\node[above] at (6.25,8.1) {$b$};
	\draw (4.75,8.5) -- (4.75,9);

	\draw (3.25,8.5) -- (3.25,9);
	\node[above] at (3.25,8.1) {$c$};
	\draw (1.75,8.5) -- (1.75,9);
	\node[above] at (1.75,8.1) {$b$};

	\draw (-1,9) -- (-2,8.5);
	\node[above] at (-2,8.1) {$a$};
	\draw (-1,9) -- (0,8.5);
	\node[above] at (0,8.1) {$a$};
	
	\draw (4.75,8.5) -- (3.75,8);
	\node[above] at (3.75,7.6) {$a$};	
	\draw (4.75,8.5) -- (5.75,8);

	\draw (4.75,8) -- (6.75,8);
	\draw (4.75,8) -- (6.75,8);
	
	\draw (6.75,7.5) -- (6.75,8);
	\node[above] at (6.75,7.1)  {$b$};	
	\draw (6.25,7.5) -- (6.25,8);
	\node[above] at (6.25,7.1)  {$c$};	
	\draw (5.75,7.5) -- (5.75,8);
	\node[above] at (5.75,7.1)  {$b$};	
	\draw (5.25,7.5) -- (5.25,8);
	\node[above] at (5.25,7.1) {$b$};	
	\draw (4.75,7.5) -- (4.75,8);
	\node[above] at  (4.75,7.1)  {$c$};

\end{tikzpicture}

\caption{\label{fig:cont}  The factorization $(aa)(bc(a(cbbcb))b)$ of $aabcacbbcbb$}

\end{figure}

To get our result, we shall define a more precise
abstraction of the position of an iterable node,
that we call its \emph{type}.

\subparagraph*{Types.} Let $\nodi \in \itera{\forest}$,
it is the descendant of a middle child of $\basis{\nodi}{\forest}$.
We denote this middle child $\mimi{\nodi}{\forest}$ (note that it is an iterable node).
 Intuitively, the \emph{type} of $\nodi$
is an abstraction of the left and right
of $\basis{\nodi}{\forest}$,
plus the context of the idempotent parent of $\nodi$
within $\mimi{\nodi}{\forest}$ (seen as a subtree).

\begin{definition} Let $w \in A^+$, $\forest \in \facto{\mu}{w}$ and
$\nodi \in \itera{\forest}$.
We define \linebreak $\type{\nodi}{\forest} =  (d, m,n,e,m',n', u) \in \Nat \times M^5 \times A^+$
as follows:
\item
\begin{itemize}
\item $1 \le d $ is the depth of the node $\nodi$ in $\forest$;
\item $ m \defined \lefe{\basis{\nodi}{\forest}}{\forest}$ and $n \defined \rige{\basis{\nodi}{\forest}}{\forest}$;
\item  $e \defined \valu{\mimi{\nodi}{\forest}} = \valu{\basis{\nodi}{\forest}}$ (it is an idempotent);
\item 
\begin{itemize}
\item if $\mimi{\nodi}{\forest} = \nodi$ (i.e. $\nodi$ is a middle children of its basis),
$m' \defined n'  \defined \mu(\vide)$;
\item if $\mimi{\nodi}{\forest} \neq \nodi$ (i.e. $\nodi \in \itera{\mimi{\nodi}{\forest}}$),
let $\nodp$ be its parent, then $m' \defined \lefe{\nodp}{\mimi{\nodi}{\forest}}$
and $n' \defined \rige{\nodp}{\mimi{\nodi}{\forest}}$;
\end{itemize}
\item $u  = w[\fr{\nodi}{\forest}]$.
\end{itemize}
\end{definition}

We denote by $T \defined [1,\mh] \times M^5 \times \{w\in A^+:|w| \le 2^{\mh}\}$.
This set is finite, furthermore, if $\forest$ has height at most $\mh$
then $\type{\nodi}{\forest} \in T$.

\begin{example}
In the factorization $\forest$ of Figure \ref{fig:type},
we consider $\nodi \in \itera{\forest}$ colored
in red. Then $\basis{\nodi}{\forest}$ is colored in blue and
$\mimi{\nodi}{\forest}$ is colored in gray
(the associated subtree is dotted).
Then $\type{\nodi}{\forest} = (d,m,n,e,m',n',u)$ where
 $d = 5$;
 $m = \mu(aa)$;
 $n = \mu(\vide)$;
 $e = \mu(bcabcbbcbb) = \mu(b) = \mu(c) = \mu(acbbcb)$;
$m' = \mu(a)$;
$n' = \mu(\vide)$;
$u = b$.
\end{example}

\begin{figure}[h!]

\centering
\begin{tikzpicture}{scale=1}

	\newcommand{\couleur}{blue}
	\newcommand{\texte}{\small \bfseries \sffamily \mathversion{bold} }

	% Input 1
	\fill[fill=gray]  (4.75,8.5)  circle (0.15);
	\fill[fill=blue]  (4,9)  circle (0.15);
	\fill[fill=red]  (5.25,7.5)  circle (0.15);
		
	\draw (1.5,9.5) -- (-1,9);
	\draw (1.5,9.5) -- (4,9);
	\draw (1.75,9) -- (6.25,9);
	\draw (6.25,8.5) -- (6.25,9);
	\node[above] at (6.25,8.1) {$b$};
	\draw (4.75,8.5) -- (4.75,9);

	\draw (3.25,8.5) -- (3.25,9);
	\node[above] at (3.25,8.1) {$c$};
	\draw (1.75,8.5) -- (1.75,9);
	\node[above] at (1.75,8.1) {$b$};

	\draw (-1,9) -- (-2,8.5);
	\node[above] at (-2,8.1) {$a$};
	\draw (-1,9) -- (0,8.5);
	\node[above] at (0,8.1) {$a$};
	
	\draw[thick,dotted] (4.75,8.5) -- (3.75,8);
	\node[above] at (3.75,7.6) {$a$};	
	\draw[thick,dotted] (4.75,8.5) -- (5.75,8);

	\draw[thick,dotted] (4.75,8) -- (6.75,8);
	
	\draw[thick,dotted] (6.75,7.5) -- (6.75,8);
	\node[above] at (6.75,7.1)  {$b$};	
	\draw[thick,dotted] (6.25,7.5) -- (6.25,8);
	\node[above] at (6.25,7.1)  {$c$};	
	\draw[thick,dotted] (5.75,7.5) -- (5.75,8);
	\node[above] at (5.75,7.1)  {$b$};	
	\draw[thick,dotted] (5.25,7.5) -- (5.25,8);
	\node[above] at (5.25,7.1) {$b$};	
	\draw[thick,dotted] (4.75,7.5) -- (4.75,8);
	\node[above] at  (4.75,7.1)  {$c$};	
	
	\draw[thick, dashed] (-2.5,8.5)--(-2.5,8) -- (0.5,8) -- (0.5,8.5);
	\node[above] at (-1,7.5) {$m$};

	\draw[thick, dashed] (3.5,8)--(3.5,7.5) -- (4,7.5) -- (4,8);
	\node[above] at (3.75,7) {$m'$};	

	\draw[thick, dashed] (3,8.5)--(3,8) -- (3.5,8) -- (3.5,8.5);
	\node[above] at (3.25,7.5) {$e$};	

	\draw[thick, dashed] (1.5,8.5)--(1.5,8) -- (2,8) -- (2,8.5);
	\node[above] at (1.75,7.5) {$e$};	

	\draw[thick, dashed] (3.5,7.5)--(3.5,7) -- (7,7) -- (7,7.5);
	\node[above] at (5.25,6.5) {$e$};	
	
\end{tikzpicture}

\caption{\label{fig:type}  The factorization $(aa)(bc(a(cbbcb))b)$ of $aabcacbbcbb$}

\end{figure}

\subparagraph*{Relating types and bitypes.}
We now show that when $\nodi, \nodj \in I(\forest)$,
then $\pro{\nodi, \nodj}$ only depends on the
types of $\nodi$ and $\nodj$.
This is the purpose of Lemma \ref{lem:indep-prof}
below, whose proof (given in the next subsection)
crucially relies on the symmetry of $\trans$.

\begin{definition}
If $(\nodi, \nodj) \in I(\forest)$, we write $\nodi < \nodj$ if $\maxe{\nodi}{\forest} < \mine{\nodj}{\forest}$.
\end{definition}

\begin{remark} Since $(\nodi, \nodj) \in I(\forest)$, $\nodi$ and $\nodj$
are not on the same branch of $\forest$,
thus either $\nodi < \nodj$ or $\nodj < \nodi$
(their frontiers cannot be interleaved).
\end{remark}

In the following, we write $\{\tau_1, \tau_2\} \subseteq T$
to describe a set of $2$ or $1$ elements of $T$.
Caution: we can have $\tau_1 = \tau_2$;
this abuse of notation makes the statements more readable.

\begin{lemma} \label{lem:indep-prof}
Let $\{\tau_1, \tau_2 \} \subseteq T$. There exists
$K \ge 0$ such that
for all factorization $\forest$ of height at most $\mh$, if $(\nodi, \nodj) \in I(\forest)$,
 $\nodi < \nodj$ and $\{\tau_1, \tau_2\} = \{\type{\nodi}{\forest}, \type{\nodj}{\forest}\}$
then $\pro{\nodi, \nodj} = K$.
We thus define 
$\pro{\{\tau_1, \tau_2\}} \defined K$.
\end{lemma}

\begin{proof} See Subsection \ref{proof:indep-prof}.
\end{proof}

\begin{remark} 
Since we use sets, $\pro{\{\tau_1, \tau_2\}} = \pro{\{\tau_2, \tau_1\}}$.
The symmetry of the definition means that the calls from $\tau_1$ to $\tau_2$
give the same production as those from $\tau_2$ to $\tau_1$.
\end{remark}

Lemma \ref{lem:indep-prof} still holds
if we let the depths of the nodes be variable.
However, the depths will be used for
the computation by a $1$-blind bimachine.

\subparagraph*{Computing $f_I$ with types.}
We can now decompose $f_I$ as a linear combinaison.
Recall that in our notations, we can have $\tau_1 = \tau_2$.
Equation \ref{eq:partiti} follows by partitioning
the sum defining $f_I$ depending on $\{\type{\nodi}{\forest}, \type{\nodj}{\forest}\}$
and applying Lemma \ref{lem:indep-prof}.

\begin{equation}
\label{eq:partiti}
\begin{aligned}
f_I(\forest) &= \sum_{\{\tau_1, \tau_2\} \subseteq T}
\sum_{\substack{(\nodi, \nodj) \in I(\forest) \\ \nodi < \nodj \\ \{\type{\nodi}{\forest}, \type{\nodj}{\forest}\} =  \{\tau_1, \tau_2\}} }\pro{\nodi, \nodj}\\
&= \sum_{\{\tau_1, \tau_2\} \subseteq T}
f_{\{\tau_1, \tau_2\}} \times \pro{\{\tau_1, \tau_2\}}\\
 \text{where }& f_{\{\tau_1, \tau_2\}}\defined \big| \{(\nodi, \nodj) \in I(\forest): \nodi < \nodj, \{\type{\nodi}{\forest}, \type{\nodj}{\forest}\} =  \{\tau_1, \tau_2\} \}\big|.
\end{aligned}
\end{equation}

Since the number of sets $\{\tau_1, \tau_2\} \subseteq T$ is bounded,
we only need to describe how to compute each $f_{\{\tau_1, \tau_2\}}$.
Indeed, we can recombine them using Claim \ref{claim:sss}.

\begin{claim} \label{claim:sss} Let $\alpha, \beta \in \Nat$ and
$g,h: B^* \rightarrow \Nat$ computable by a $k$-blind bimachine, so is:
\begin{equation*}
\alpha g+\beta h: w \mapsto \alpha \times g(w) + \beta \times h(w).
\end{equation*}
\end{claim}
\begin{proof}[Proof idea.]
We build a $k$-blind transducer which
simulates sequentially $g$, and then $h$.
\end{proof}

From now we fix $\tau_1 = (d_1,m_1,n_1,e_1,m_1',n_1',u_1)$ and $\tau_2 = (d_2,m_2,n_2,e_2,m_2',n_2',u_2)$.
We assume that $\tau_1 \neq \tau_2$ (the case of equality
is similar and even easier). Up to switching $1$ and $2$,
we assume that $d_1 \ge d_2$, that is $\tau_1$ is "deeper" than $\tau_2$.
Then:
\begin{equation*}
\begin{aligned}
f_{\{\tau_1, \tau_2\}}(\forest)& = \left| \{(\nodi, \nodj) \in I(\forest): \nodi < \nodj, \{\type{\nodi}{\forest}, \type{\nodj}{\forest}\} =  \{\tau_1, \tau_2\} \}\right|\\
&= \frac{1}{2}\big| \{(\nodi, \nodj) \in I(\forest):\{\type{\nodi}{\forest}, \type{\nodj}{\forest}\} =  \{\tau_1, \tau_2\} \}\big|\\
& = \big| \{(\nodi, \nodj) \in I(\forest): \type{\nodi}{\forest} = \tau_1, \type{\nodj}{\forest} =  \tau_2 \}\big|\\
\end{aligned}
\end{equation*}
The second line is justified since the function
$\sigma = (\nodi, \nodj) \mapsto (\nodj, \nodi)$ is an involution without fixpoint
(since $(\nodi, \nodi) \not \in I(F)$) of the set
$\{(\nodi, \nodj) \in I(\forest):\{\type{\nodi}{\forest}, \type{\nodj}{\forest}\} =  \{\tau_1, \tau_2\}\}$
which reverses the ordering of the pair (that is, if $\nodi < \nodj$ then $\nodj > \nodi$).
For the third line, we similarly use $\sigma$ since it
reverses the types $\type{\nodi}{\forest}$ and $\type{\nodj}{\forest}$
(we use $\tau_1 \neq \tau_2$ here).

By partitioning the last set, we then get:
\begin{equation}
\begin{aligned}
f_{\{\tau_1, \tau_2\}}(\forest) & = \sum_{\substack{\nodi \in \itera{\forest} \\ \type{\nodi}{\forest} = \tau_1}}  \big| \{(\nodi, \nodj) \in I(\forest): \type{\nodj}{\forest} = \tau_2 \}\big|\\
& = \sum_{\substack{\nodi \in \itera{\forest} \\ \type{\nodi}{\forest} = \tau_1}} \big| \{\nodj: (\nodi, \nodj) \in I(\forest), \type{\nodj}{\forest} =  \tau_2 \}\big|\\
\end{aligned}
\label{eq:tautau}
\end{equation}

For all $\nodi$, we thus count all the $\nodj \in \itera{\nodj}$ such that $\type{\nodj}{\forest} = \tau_2$
and $\basis{\nodj}{\forest} = \basis{\nodi}{\forest} $,
except those which are in "nearly" on the same branch as $\nodi$.
However this set of $\nodj$ seems to depend on $\nodi$, which we want to avoid.

We want to rewrite Equation \ref{eq:tautau} in order to remove this dependency.

\begin{claim} \label{claim:disju} Let $\forest$ of height
at most $\mh$, for all $\nodi \in \itera{\forest}$ with $\type{\nodi}{\forest} = \tau_1$;
then
$\{\nodj: (\nodi, \nodj) \in I(\forest), \type{\nodj}{\forest} =  \tau_2\} = J_\forest(\basis{\nodi}{\forest}) \smallsetminus A_\forest(\nodi)$ where:
\item
\begin{itemize}
\item $J_F(\nodb) \defined \{\nodj \in \itera{\forest}: \basis{\nodj}{\forest} = \nodb, \type{\nodj}{\forest} =  \tau_2 \}$;
\item $A_\forest(\nodi)$ is the set of nodes $\nodj$ such that $\type{\nodj}{\forest} = \tau_2$,
$\basis{\nodj}{\forest} = \basis{\nodi}{\forest}$ and $\nodj$ is either an ancestor of $\nodi$, or the right/left
sibling of an ancestor of $\nodi$.
\end{itemize}
\end{claim}

\begin{proof}
Let $\nodi$ be fixed as in the claim.
Let $\nodj \in J_\forest(\nodi)$.
Since $d_1 \ge d_2$, then $\nodi$ is "deeper" than $\nodj$ and
either $(\nodi, \nodj) \in I(\forest)$;
or $\nodj$ is either an ancestor of $\nodi$, or the right/left
sibling of an ancestor of $\nodi$.
This case disjunction gives a disjoint union.
\end{proof}

Using Claim \ref{claim:disju} and Equation \ref{eq:tautau}, we finally get:
\begin{equation}
\begin{aligned}
f_{\{\tau_1, \tau_2\}}(\forest) & = \sum_{\nodi \in \itera{\forest}} \big| J_\forest(\basis{\nodi}{\forest})\big|- \big|A_\forest(\nodi)\big|
\label{eq:indep}
\end{aligned}
\end{equation}

Let us explain why Equation \ref{eq:indep} is satisfactory.
First, the set $J_\forest(\basis{\nodi}{\forest})$ does not depend on $\nodi$,
but only on $\basis{\nodi}{\forest}$, which is a bounded information that can
be given to an external blind function.  Second, the set $A_{\forest}(\nodi)$
still depends on $\nodi$, but:

\begin{claim}
$|A_{\forest}(\nodi)| \le 3$
\end{claim}

\begin{proof} The nodes of  $A_{\forest}(\nodi)$ have type $\tau_2$.
Thus  $A_{\forest}(\nodi)$  is included in the set containing the ancestor of $\nodi$
at height $d_2$, and its potential right/left siblings.
\end{proof}

Hence, it is also a bounded information which can be
given to a external function.
We can now describe  in Algorithm \ref{algo:indep} how to compute $f_{\tau_1, \tau_2}$
with a $1$-blind bimachine.

\begin{algorithm}[h!]
\SetKw{KwVar}{Variables:}
\SetKwProg{Fn}{Function}{}{}
\SetKw{In}{in}
\SetKw{Out}{Output}

 	%\tcc{Right "context" of $j$ in block $k$.}	
 \Fn{$\operatorname{Main}(\forest)$}{

	\For{$\nodi \in \itera{\forest}$}{

		\If{$\type{\nodi}{\forest} = \tau_1$}{
				$c \leftarrow |A_\forest(\nodi)|$;
				\tcc{Number of nodes of type $\tau_2$ "above" $\nodi$.}
				
				$\nodb \leftarrow \basis{\nodi}{\forest}$;
				\tcc{Base of $\nodi$.}
				
				\Out{$\exte^{-c}_{\nodb}(\forest)$}
			
			}
		}
			
		}

 \Fn{$\operatorname{\exte}^{-c}_{\nodb}(\forest)$}{
 
 		$\operatorname{\textsf{counter}} \leftarrow 0$;
 	 
		\For{$\nodj \in \itera{\forest}$}{
			\If{\normalfont $\type{\nodj}{\forest} = \tau_2$ and $\basis{\nodj}{\forest} = \nodb$}{
			
			\eIf{\normalfont $\operatorname{\textsf{counter}} < c  $}{
				 $\operatorname{\textsf{counter}} \leftarrow \operatorname{\textsf{counter}}  + 1$;
				 \tcc{No output on the $c$ first nodes.}	
			}{
				\Out{$1$}
			}	
			}		
		}}
	
 \caption{\label{algo:indep} Computing $f_{\{\tau_1, \tau_2\}}$ with a $1$-blind bimachine}
\end{algorithm}

\subparagraph*{Correctness of the algorithm.}
It follows by definition that $\exte_{\nodb}^{-c}(\forest) = \max(0,\big| J_\nodb(\forest)\big|{-} c)$.
Using Equation \ref{eq:indep}, we get $\operatorname{Main}(\forest) = f_{\{\tau_1,\tau_2\}}(\forest)$.

\subparagraph*{Implementation by a $1$-blind bimachine.}
We justify how Algorithm \ref{algo:indep} can be implemented with a $1$-blind bimachine.
First, it uses finitely many external functions $\exte^{-c}_\nodb$ for $0 \le c \le 3$,
and $\nodb$ a base of the forest (there is a bounded number of bases).

It seems however that $\nodb$ depends on the input $\forest$, but
we can naturally order the bases depending on "the first position below an iterable node having this base":

\begin{claim} The bases of $\forest$ can be ordered $\nodb_1, \dots, \nodb_{\ell(\forest)}$,
in the strictly increasing order of the values
$\min\left\{ \mine{\nodi}{\forest}:\nodi \in \itera{\forest} \text{ and } \basis{\nodi}{\forest} = \nodb_i\right\}$.
\end{claim}

\begin{proof} The set $S_{\nodb} \defined \{\nodi \in \itera{\forest} \text{ and } \basis{\nodi}{\forest} = \nodb\}$
is not empty since idempotent nodes have at least one middle children.
We only need to show that $S_{\nodb}$ are $S_{\nodb'}$ are disjoint when
$\nodb \neq \nodb'$, what is clear since every iterable node has
a unique base.
\end{proof}
Hence we do not really index the external function by
the bases, but by their number in this ordering
(the range of numbers is bounded independently from $\forest$).

The second issue is that Algorithm \ref{algo:indep} iterates on
 $\itera{\forest}$, but not on the positions $\{1, \dots, |F|\}$
 (that's what a bimachine does). Recall the function $\pi$ from
 Definition \ref{def:pi}. 
 The idea is to code
 the iterable node $\nodi$ by the position $\xi(\nodi) \defined \pi(\mine{\nodi}{\forest}) \in \{1, \dots, |\forest|\}$.
 Due to the partition, each position of $\{1, \dots, |\forest|\}$ codes
 at most one iterable node.
 
 We then need to check if a position codes for an iterable node $\nodi$
 of type $\tau_1$, such that $|A_{\forest}| = c$ and $\nodb_i = \basis{\nodi}{\forest}$.
 This is the purpose of Claim \ref{claim:check-tau},
 which is used to build the monoid morphism of the main bimachine
 (similarly to what we did for Lemma \ref{lem:distinguishable}).
 Its proof is an easy construction of regular languages.
 
\begin{claim} \label{claim:check-tau}
One can build a regular language $P^{-c}_{i} \subseteq (\apar{A} \uplus \marq{\apar{A}})^*$
such that for all $w \in A^+$, $\forest \in \ltm$ over $w$ and $1 \le j \le |\forest|$:
\begin{align*} 
&\omar{\forest}{j} \in P^{-c}_{i} \text{ if and only if } \xi(\nodi) = j
\text{ for some } \nodi \in \itera{\forest} \text{ such that }\\
& \type{\nodi}{\forest} = \tau_1, |A_{\forest}(\nodi)|= c  \text{ and }  \basis{\nodi}{\forest} = \nodb_i.
\end{align*}
\end{claim}

The construction of the bimachines for the external functions $\exte^{-c}_{\nodb}$ is similar.
However, we need to discuss the implementation of $\operatorname{\textsf{counter}}$.
The idea is to use the ordering given by $\xi(\nodj)$ on the nodes $\nodj \in \itera{\forest}$,
and to produce $1$ each time $\type{\nodj}{\forest} = \tau_2$
and $\basis{\nodj}{\forest} = \nodb$, except for the $c$ first ones
(which is a regular property since $c \le 3$).

\subsection{Proof of Lemma \ref{lem:indep-prof}}

\label{proof:indep-prof}

Let $\tau_1, \tau_2 \in T$ be two types.

\begin{claim} \label{claim:same} If there exists a factorization $\forest$
of height at most $\mh$ and $(\nodi, \nodj) \in I(\forest)$ such that
$\{\tau_1, \tau_2\} = \{\type{\nodi}{\forest}, \type{\nodj}{\forest}\}$,
then $\exists 1 \le d_1, d_2 \le \mh$, $m,n,e,m_1,n_1,m_2,n_2 \in M$,
$u_1, u_2 \in A^+$ with $|u_1|, |u_2| \le 2^{\mh}$
such that:
\begin{equation*}
\left\{
    \begin{array}{l}
        \tau_1 = (d_1, m,n,e,m_1,n_1,u_1);\\
        \tau_2 = (d_2, m,n,e,m_2,n_2,u_2).\\
    \end{array}
\right.
\end{equation*}
\end{claim}

\begin{proof} If $(\nodi, \nodj) \in I(\forest)$ then $\basis{\nodi}{\forest} = \basis{\nodj}{\forest}$.
Hence they these nodes have the same left $m$, right $n$ and value $e$.
\end{proof}

We thus assume that $\tau_1$ and $\tau_2$ are
as in Claim \ref{claim:same}, since otherwise Lemma \ref{lem:indep-prof} is true by emptiness.
The general idea is to reduce the proof to the
definition of a symmetrical $1$-marble bimachine. For this,
we relate pairs of nodes with bitypes.

Let $w \in A^+$ and $ \forest \in \facto{\mu}{w}$ be a factorization of height at most $\mh$.
If $(\nodi, \nodj) \in I(\forest)$ with $\nodi < \nodj$,
we let:
\begin{equation*}
\begin{aligned}
\bitype{\nodi, \nodj} \defined &\mu(w[1{:}{\mine{\nodi}{\forest}{-1}}]) \cro{w[\fr{\nodi}{\forest}]} \mu(w[\maxe{\nodi}{\forest}{+}1{:}{\mine{\nodj}{\forest}{-1}}]) \\
&\cro{w[\fr{\nodj}{\forest}]} \mu(w[\maxe{\nodj}{\forest}{+}1{:}|w|]).
\end{aligned}
\end{equation*}

The following lemma, shown in Subsubsection \ref{proof:lembi}, justifies this definition.

\begin{lemma} \label{lem:bitype-prod}
If $(\nodi, \nodj) \in I(\forest)$ and $\nodi < \nodj$, then $\pro{\nodi,\nodj} = \pro{\bitype{\nodi, \nodj}}$.
\end{lemma}

We finally want to relate $\bitype{\nodi, \nodj}$ with the types $\tau_1$ and $\tau_2$.

From now assume without loss of generality
that $ \type{\nodi}{\forest}= \tau_1$ and $\type{\nodj}{\forest}  = \tau_2$.
Let $e_1 \defined \mu(u_1)$ and $e_2 \defined \mu(u_2)$, 
the following claim is immediate by definition of a type.

\begin{claim}
$u_1 = w[\fr{\nodi}{\forest}] $, $u_2 = w[\fr{\nodj}{\forest}]$ and $m_1 e_1 n_1 = m_2 e_2 n_2 = e$.
\end{claim}

Similarly, we can show the following (the proof uses ideas from that of Lemma \ref{lem:cas},
it is in fact easier, hence we do not detail it).
\begin{claim} $  \mu(w[1{:}\mine{\nodi}{\forest}{-}1]) = m e m_1 e_1$ and $ \mu(w[\maxe{\nodj}{\forest}{+}1{:}|w|])=e_2n_2 en$.
\end{claim}

We finally create some $p \in M$, which
fits the conditions of Definition \ref{def:symmetrical}.
This result relies on the definition of $I(\forest)$, that is
that the nodes are not "nearly on the same branch".

\begin{lemma} \label{lem:cas} There exists $p \in M$ such that:
\begin{equation*}
\left\{
    \begin{array}{l}
        m_1 e_1 p e_2 n_2 = e;\\
        em_1e_1p e_2= em_2e_2;\\
        e_1p e_2 n_2 e= e_1n_1e;\\
         \mu(w[\maxe{\nodi}{\forest}{+}1{:}{\mine{\nodj}{\forest}{-1}}])  = e_1 p e_2;\\
    \end{array}
\right.
\end{equation*}
\end{lemma}

\begin{proof} Since $\nodi, \nodj \in \itera{\forest}$, then
$\nodi$ (resp. $\nodj$) has a right sibling $\nodi'$ (resp. a left sibling $\nodj'$) and
$\valu{\nodi'} = \valu{\nodi} = \mu(u_1) = e_1$
(resp. $\valu{\nodj'} = \valu{\nodj} = \mu(u_2) = e_2$).
Two cases occur:
\item
\begin{itemize}

\item either $\nodi' = \nodj'$, that is $\nodi, \nodi', \nodj$
are successive siblings. In that case
$m_1 = m_2$, $n_1 = n_2$ and $e_1 = e_2$.
We set $p \defined \mu(\vide)$, the conditions are easy to check;

\item or $\nodi' \neq \nodj'$. In that case:
\begin{claim}
$\maxe{\nodi'}{\forest} < \mine{\nodj'}{\forest}$
\end{claim}
\begin{proof}
Consider the least common ancestor $\noda$
of $\nodi$ and $\nodj$. It cannot be $\nodi$ nor $\nodj$
itself by definition of $I(\forest)$,
thus one of the following occurs:
\item
\begin{itemize}
\item either $\noda$ is neither the direct parent of $\nodi$ nor that of $\nodj$,
the result is immediate;
\item or $\noda$ it is the parent of $\nodi$. But by definition of $I(\forest)$,
$\nodj$ is not a descendant of $\nodi'$, hence $\nodj$
is the descendant of a sibling on the right of $\nodi'$ and the result follows;
\item or $\noda$ it is the parent of $\nodj$, treated symmetrically.
\end{itemize}
\vspace*{-1\baselineskip}
\end{proof}
\item
We define:
\begin{equation*}
p \defined \mu(w[\maxe{\nodi'}{\forest}{+}1{:}\mine{\nodj'}{\forest}{-}1]).
\end{equation*}

It is always clear that $\mu(w[\maxe{\nodi}{\forest}{+}1{:}{\mine{\nodj}{\forest}{-1}}])  = e_1 p e_2$.
The other equations can be checked by a (rather long) case disjunction.
Let us detail one case.

Assume that $\nodi$ is a middle child of $\basis{\nodi}{\forest} = \basis{\nodj}{\forest}$, 
i.e. $\nodi = \mimi{\nodi}{\forest}$.
Then $m_1 = n_1  = \mu(\epsilon)$ and $e_1 = e$.
Let $\nodi, \nodi', \nodi'', \nodi'''$ be the successive right siblings below $\basis{\nodi}{\forest}$
and assume that $\nodj$ is a descendant of $\nodi'''$.
Then $p = em_2 e_2$ or $p = em_2$ (depending on whether $\nodj'$
is a middle child or not).
Then $m_1 e_1 p e_2 n_2 = e m_2 e_2 n_2 = e$; $em_1e_1p e_2= em_2e_2$;
and  $e_1p e_2 n_2 e=  e m_2 e_2 n_2 e = e= e_1n_1e$.
\end{itemize}
\vspace*{-1\baselineskip}
\end{proof}

Finally we get $\bitype{\nodi, \nodj} = em_1 e_1 \cro{\uu} e_1 p e_2 \cro{\ud} e_2 n_2 e$.
To conclude the proof, we use the $K$ of Definition \ref{def:symmetrical}
whose conditions are met. Switching $\tau_1$ and $\tau_2$ gives the same $K$,
using the second case of Definition \ref{def:symmetrical}.

\subsubsection{Proof of Lemma \ref{lem:bitype-prod}}

\label{proof:lembi}

Let $\nodi < \nodj$ and $(\nodi, \nodj) \in I(F)$.
Let $\{i_1, \dots, i_{\ell}\} \defined \fr{\nodi}{\forest}$
and $\{j_1, \dots, j_{\ell'}\} \defined \fr{\nodj}{\forest}$.
Hence $i_{\ell} = \maxe{\nodi}{\forest} < \mine{\nodi}{\forest} = j_1$.
Furthermore:
\begin{equation*}
\begin{aligned}
\pro{\nodi, \nodj} &= \sum_{\substack{1 \le k \le \ell \\ 1 \le k' \le \ell'}} \pro{i_k,j_{k'}}
\end{aligned}
\end{equation*}

Let $q \defined \mu([1{:}i_1{-}1])$,  $q' \defined \mu(w[i_{\ell}{+}1{:}j_{1}{-}1]) $
and $q'' \defined \mu(w[j_{\ell'}{+}1{:}|w|])$.
Let $u \defined w[\fr{\nodi}{\forest}]$ and $u' \defined w[\fr{\nodj}{\forest}]$
be the words described by their frontiers.

By definition, $\bitype{\nodi, \nodj} = q\cro{u} q' \cro{u'} q''$.
Hence:
\begin{equation}
\begin{aligned}
&\pro{\bitype{\nodi, \nodj}} = \sum_{\substack{1 \le  k \le \ell  \\ 1 \le k' \le \ell'}} P(k,k') \text{ where}\\
&P(k,k') \defined q \mu(u[1{:}{k}{-}1])\cro{u[k]} 
\mu(u[{k}{+}1{:}{\ell}] ) 
q' \mu(u'[{1}{:}{k'}{-}1] )\cro{u'[{k'}]} \mu(u'[{k'}{+}1{:}\ell']) q''.
\end{aligned}
\label{eq:Pkk}
\end{equation}

Hence to show Lemma  \ref{lem:bitype-prod},
is is enough to show $\forall k,k'$ $\pro{i_k, i_{k'}} = P(k,k')$. And:
\begin{equation}
\begin{aligned}
\pro{i_k, i_k'}&=  \pro{\mu(w[1{:}i_{k}{-}1])\cro{w[i_k]} \mu(w[i_{k}{+}1{:}j_{k'}{-}1] )\cro{w[j_{k'}]} \mu(w[j_{k'}{+}1{:}|w|])}\\
&=  \pro{\mu(w[1{:}i_{k}{-}1])\cro{u[k]} \mu(w[i_{k}{+}1{:}j_{k'}{-}1] )\cro{u'[{k'}]} \mu(w[j_{k'}{+}1{:}|w|])}\\
\end{aligned}
\label{eq:prodkk}
\end{equation}

Comparing Equations \ref{eq:Pkk} and \ref{eq:prodkk}, we see that we only need to show that $\forall k,k'$:
\begin{equation*}
\left\{
    \begin{array}{l}
        \mu(u[1{:}k{-}1]) = \mu(w[i_1{:}i_{k}{-}1])\\
 	 \mu(u[k{+}1{:}\ell]) = \mu(w[i_{k}{+}1{:}i_{\ell}])\\
 	 \mu(u'[1{:}k'{-}1]) = \mu(w[j_{1}{:}j_{k'}{-}1])\\
 	 \mu(u'[k'{+}1{:}\ell']) = \mu(w[j_{k'}{+}1{:}j_{\ell'}])\\
    \end{array}
\right.
\end{equation*}

We only treat the first case, since the others
are similar. It follows by applying the claim
below to the node $\nodi \in \facto{\mu}{w[i_1{:}i_{\ell}]}$
seen as a factorization.

\begin{claim} \label{claim:skel-monoid} Let $\forest \in \facto{\mu}{x}$,
$y \defined x[\fr{\forest}{\forest}]$ and
$\{1 = i_1 < \cdots < i_{\ell} = |x|\} \defined \fr{\forest}{\forest}$.
\\
Then $\forall 1 \le k \le \ell$, $\mu(x[1{:}i_k{-1}]) = \mu(y[1{:}k{-}1])$.
\end{claim}

\begin{proof} The result is shown by induction on the factorization.
The only interesting case is when
$\forest = (\forest_1) \dots (\forest_n)$ starts with an idempotent node.
Let $x = x_1 \cdots x_n$ be such
that $\forest_i$ is a factorization of $x_i$. Then $ e \defined \mu(x_i)$ is an idempotent.

By definition of the dependency and Lemma \ref{lem:preserv},
there exists $1 \le p \le \ell$ such that $i_{p} = |x_1|$ and $i_{p+1} = |x| - |x_n| +1$.
Hence:
\begin{equation*}
y = x_1[i_1] \cdots x_1[i_p] x_n[i_{p+1} {-} |x|{+}|x_n|] \cdots x_n[i_{\ell} {-} |x|{+}|x_n|]
\end{equation*}

Let $1 \le k \le \ell$, we assume that $p+1 \le k$
(the case $k \le p$ is easier).
Then:
\begin{equation*}
	y[1{:}k{-}1] = x_1[i_1] \cdots x_n[i_p] x_n[i_{p+1} {-} |x|{+}|x_n|] \cdots x_n[i_{k{-}1} {-} |x|{+}|x_n|].
\end{equation*}
By Lemma \ref{lem:preserv}:
\begin{equation*}
	\mu(x_1[i_1] \cdots x_1[i_{p-1}]x_1[i_p]) = \mu(x_1) = e
\end{equation*}
And by induction hypothesis:
\begin{equation*}
	\mu(x_n[i_{p+1}{-}|x|{+}|x_n|] \cdots x_n[i_{k-1} {-} |x|{+}|x_n|]) = \mu(x_n[i_{p+1}{-}|x|{+}|x_n|{:}i_{k} {-} 1 {-} |x|{+}|x_n|]) 
\end{equation*}

Finally since $e = \mu(x_2) = \cdots = \mu(x_{n-1}) $ is idempotent, it follows that:
\begin{equation*}
\begin{aligned}
\mu(x[1{:}i_k{-}1] ) &= e \mu(x_n[i_{p+1}{-}|x|{+}|x_n|{:}[i_{k} {-} 1 {-} |x|{+}|x_n|])\\
&= \mu(x_1[i_1] \cdots x_1[i_{p-1}]x_1[i_p]) \mu(x_n[i_{p+1}{-}|x|{+}|x_n|] \cdots x_n[i_{k-1} {-} |x|{+}|x_n|])\\
&=\mu(y[1{:}k{-}1]).\\
\end{aligned}
\end{equation*}

\vspace*{-1\baselineskip}

\end{proof}

\section{Proof of Theorem \ref{theo:characterization} - characterization}

We only show that $\ref{it:trt} \Rightarrow \ref{it:det}$.
Assume that the $1$-blind bimachine $\trans$ is symmetrical.

By Proposition \ref{prop:effective-factorizations}, one can
build a two-way transducer (with non-unary output) that
given $w \in A^+$, computes some  factorization $\forest \in \facto{\mu}{w}$ of height at most $\mh$,
 described as a word, that is $\forest \in \apar{A}^+$. We denote this function $f_{\ltm}$.
 
Using lemmas \ref{lem:distinguishable}, \ref{lem:linked}, \ref{lem:independent},
the functions $f_D, f_L$ and $f_I$ are computable by a $1$-blind transducer.
We want to compute their sum, in the sense of Claim \ref{claim:sss}.
It follows from this claim that $f_D + f_L + f_I$ can be computed
by a $1$-blind bimachine, or equivalently a $1$-blind transducer.
By Lemma \ref{lem:nodes-prod} and the definitions
of $f_D, f_L$ and $f_I$, it follows that:
\begin{equation*}
(f_D + f_L + f_I) \circ f_{\ltm} = f.
\end{equation*}

It follows from \cite{nguyen2020comparison}
that (even for non-unary alphabets) $1$-blind transducers are effectively closed under
composition with two-way transducers, hence $(f_D + f_L + f_I) \circ f_{\ltm}$
can be computed by a $1$-blind transducer.

\end{document}